%% file: main.tex
\documentclass[sn-mathphys,iicol]{sn-jnl}

\input{macro-sn}

\jyear{2022}%

\theoremstyle{thmstyleone}%
\newtheorem{theorem}{Theorem}
\theoremstyle{thmstyletwo}%
\newtheorem{example}{Example}%

\theoremstyle{thmstylethree}%
\newtheorem{definition}{Definition}%

\begin{document}

\title[Article Title]{Full-Program Induction: Verifying Array Programs sans Loop Invariants}


\author[1]{\fnm{Supratik} \sur{Chakraborty}}\email{supratik@cse.iitb.ac.in}

\author[1]{\fnm{Ashutosh} \sur{Gupta}}\email{akg@cse.iitb.ac.in}

\author*[1,2]{\fnm{Divyesh} \sur{Unadkat}}\email{divyesh.unadkat@tcs.com}

\affil[1]{\orgdiv{Computer Science and Engineering}, \orgname{Indian Institute of Technology Bombay}, \orgaddress{\street{Main Gate Rd, IIT Area, Powai}, \city{Mumbai}, \postcode{400076}, \state{Maharashtra}, \country{India}}}

\affil*[2]{\orgname{TCS Research}, \orgaddress{\street{54-B Hadapsar Industrial Estate, Hadapsar}, \city{Pune}, \postcode{411013}, \state{Maharashtra}, \country{India}}}

\abstract{ Arrays are commonly used in a variety of software to store
  and process data in loops.  Automatically proving safety properties
  of such programs that manipulate arrays is challenging.  We present
  a novel verification technique, called \emph{full-program
  induction}, for proving (a sub-class of) quantified as well as
  quantifier-free properties of programs manipulating arrays of
  parametric size $N$.  Instead of inducting over individual loops,
  our technique inducts over the entire program (possibly containing
  multiple loops) directly via the program parameter $N$.  The
  technique performs non-trivial transformations of the given program
  and pre-conditions during the inductive step.  The transformations
  assist in effectively reducing the assertion checking problem by
  transforming a program with multiple loops to a program which has
  fewer and simpler loops or is loop-free.  Significantly,
  \emph{full-program induction} does not require generation or use of
  loop-specific invariants.  To assess the efficacy of our technique,
  we have developed a prototype tool called {\ourtool}.  We
  demonstrate the performance of {\ourtool} vis-a-vis several
  state-of-the-art tools on a large set of array manipulating
  benchmarks from the international software verification competition
  (SV-COMP) and on several programs inspired by algebraic functions
  that perform polynomial computations.  }




\maketitle

\section{Introduction}
\label{sec:intro}
\input{intro}

\section{Overview of Full-Program Induction}
\label{sec:overview}
\input{overview}

\section{Preliminaries}
\label{sec:prelim}
\input{prelim}

\section{Difference Computation}
\label{sec:diffcomp}

\input{diffcomp}

\section{Algorithms for Full-Program Induction}
\label{sec:algorithms}
\input{algorithms}

\section{Progress Measures}
\label{sec:progress}
\input{progress}

\section{Full-program Induction in Generalized Settings}
\label{sec:generalized}
\input{generalized}

\section{Experimental Results}
\label{sec:experiments}
\input{experiments}

\section{Related Work}
\label{sec:related}
\input{related}

\section{Conclusion \& Future Work}
\label{sec:conc}
\input{conclusion}

\bibliography{fpi}

\end{document}

%% file: macro-sn.tex
\catcode`\|=12\relax  

\usepackage{tikz}
\usetikzlibrary{patterns}
\usetikzlibrary{chains}
\usetikzlibrary{shadows.blur}
\usetikzlibrary{plotmarks}
\usetikzlibrary{shapes}
\usetikzlibrary{arrows}
\usetikzlibrary{positioning}
\usetikzlibrary{automata}
\usetikzlibrary{calc}

\usepackage{pgf}
\usepackage{pgfplots}
\usepackage{pgfplotstable}
\usepackage{booktabs}
\usepackage{colortbl}

\usepackage{alltt}

\algdef{SE}[DOWHILE]{Do}{doWhile}{\algorithmicdo}[1]{\algorithmicwhile\ #1}%

\makeatletter
\newcommand\llarge{\@setfontsize\llarge{14}{16}}
\makeatother

\newcommand{\cube}{\ensuremath{\mathtt{^3}}}

\newcommand{\integers}{\mathbb{Z}}

\newcommand{\ppre}{\ensuremath{\mathsf{Pre}}}

\newcommand{\ltrue}{\mathbf{tt}}
\newcommand{\lfalse}{\mathbf{ff}}

\newcommand{\domain}{\mathcal{D}}

\newcommand{\smem}{\mathcal{S}}

\newcommand{\ourtool}{\textsc{Vajra}}
\newcommand{\tiler}{\textsc{Tiler}}
\newcommand{\freqhorn}{\textsc{FreqHorn}}
\newcommand{\expl}{\textsc{expl}}
\newcommand{\viap}{\textsc{VIAP}}
\newcommand{\veriabs}{\textsc{VeriAbs}}
\newcommand{\symdiff}{\textsc{SymDiff}}
\newcommand{\zthree}{\textsc{Z3}}

\newcommand{\booster}{\textsc{Booster}}

\newcommand{\vaphor}{\textsc{Vaphor}}

\newtheorem{lemma}{Lemma}

\newcommand{\PP}{\ensuremath{\mathsf{P}}}

\newcommand{\EE}{\ensuremath{\mathsf{E}}}
\newcommand{\IE}{\ensuremath{\mathsf{IndE}}}
\newcommand{\LBE}{\ensuremath{\mathsf{LBnd}}}

\newcommand{\PB}{\ensuremath{\mathsf{PB}}}
\newcommand{\Stmt}{\ensuremath{\mathsf{St}}}
\newcommand{\Stmta}{\ensuremath{\mathsf{StLF}}}
\newcommand{\scVar}{\ensuremath{v}}
\newcommand{\lpVar}{\ensuremath{\ell}}
\newcommand{\ArVar}{\ensuremath{A}}
\newcommand{\OP}{\ensuremath{\mathsf{op}}}
\newcommand{\RelOP}{\ensuremath{\mathsf{relop}}}
\newcommand{\BoolE}{\ensuremath{\mathsf{BoolE}}}
\newcommand{\iif}{\ensuremath{\mathbf{if}}}
\newcommand{\eelse}{\ensuremath{\mathbf{else}}}
\newcommand{\tthen}{\ensuremath{\mathbf{then}}}
\newcommand{\ffor}{\ensuremath{\mathbf{for}}}

\newcommand{\cconst}{\ensuremath{\mathsf{c}}}

\newcommand{\AssignStmts}{\ensuremath{\mathsf{AssignSt}}}

\newcommand{\Unlabeled}{\ensuremath{\mathsf{U}}}

\newcommand{\infocomment}[1]{{\ttfamily\textcolor{darkgray}{\newline$\triangleright$ #1}}}
\renewcommand{\true}{\ensuremath{\mathsf{True}}}
\renewcommand{\false}{\ensuremath{\mathsf{False}}}
\newcommand{\worklist}{\ensuremath{\mathsf{WorkList}}}
\newcommand{\peelcount}{\ensuremath{\mathsf{PeelCount}}}
\newcommand{\peelnodes}{\ensuremath{\mathsf{PeelNodes}}}
\newcommand{\gluenodes}{\ensuremath{\mathsf{GlueNodes}}}
\newcommand{\processednodes}{\ensuremath{\mathsf{ProcessedNodes}}}
\newcommand{\condnodes}{\ensuremath{\mathsf{CondNodes}}}
\newcommand{\affectedvars}{\ensuremath{\mathsf{AffectedVars}}}

\newcommand{\inv}{\ensuremath{\mathsf{Inv}}}

\newcommand{\NonDefV}{\ensuremath \mathit{NDefV}}
\newcommand{\DefV}{\ensuremath \mathit{DefV}}
\newcommand{\defindex}{\ensuremath \mathit{defIndex}}
\newcommand{\useindex}{\ensuremath \mathit{useIndexSet}}
\newcommand{\Branch}[2]{\ensuremath \mathit{Branch}({#1},{#2})}
\newcommand{\Dep}[2]{\ensuremath \mathit{Dep}({#1},{#2})}
\newcommand{\PPP}[2]{\ensuremath \PP_{{#1}}^\star}
\newcommand{\Peel}[1]{\ensuremath \mathsf{Peel}({#1})}
\newcommand{\Path}[1]{\ensuremath \pi(\sigma, {#1})}
\newcommand{\NodeP}{n}
\newcommand{\Skip}{\ensuremath{\mathsf{skip}}}
\newcommand{\BoolCond}{\ensuremath{\mathsf{C}}}
\newcommand{\ProgFrag}{\ensuremath{\mathsf{S}}}
\newcommand{\TRPP}{\ensuremath{\mathsf{T}}}
\newcommand{\branchnodes}{\ensuremath{\mathsf{BranchNodes}}}


%% file: intro.tex
Use of software controlled systems in industrial and household
appliances is constantly increasing.  Functionalities of such software
are programmed with extensive use of loops and conditional statements
that manipulate different data structures such as arrays, lists, and
vectors to store and process data during its operation.  Programs with
loops manipulating arrays are quite common in many such applications.
These programs are expected to be of immensely high-quality as their
erroneous functioning can cause adversities to businesses as well as
human lives.  Thus, ensuring correctness of these programs is of
paramount importance.  Unfortunately, assertion checking in such
programs is, in general, undecidable.  Existing tools therefore use a
combination of verification techniques that work well for certain
classes of programs and assertions, and yield conservative results
otherwise.

In this paper, we present a new verification technique,
called \emph{full-program induction}, to add to this arsenal of
techniques.  Specifically, we focus on programs with loops
manipulating arrays, where the size of each array is a symbolic
integer parameter $N~(> 0)$.  We allow (a sub-class of) quantified and
quantifier-free pre- and post-conditions that may depend on the
symbolic parameter $N$.  Thus, the problem we wish to solve can be
viewed as checking the validity of a parameterized Hoare triple
$\{\varphi(N)\} \;\PP_N\; \{\psi(N)\}$ for all values of $N~(> 0)$,
where the program $\PP_N$ computes with arrays of size $N$, and $N$ is
a free variable in $\varphi(\cdot)$ and $\psi(\cdot)$.

Like earlier verification approaches \cite{kind}, our technique also
relies on mathematical induction to reason about programs with loops.
However, the way in which the inductive claim is formulated and proved
differs significantly from the previous techniques.  Specifically, (i)
we \emph{induct on the full program} (possibly containing multiple
loops) with parameter $N$ and not on iterations of individual loops in
the program, (ii) we perform \emph{non-trivial correct-by-construction
transformation of the given program and the pre-condition}, whenever
feasible, to simplify the inductive step of reasoning, (iii)
we \emph{strengthen the pre- and post-condition simultaneously during
the inductive step} using the auxiliary inductive predicates obtained
by employing Dijkstra's weakest pre-condition computation, (iv)
we \emph{recursively apply} the technique to prove the inductive step
and most importantly (v) we \emph{do not require explicit or implicit
loop-specific inductive invariants} to be provided by the user or
generated by a solver (viz. by constrained Horn clause
solvers \cite{chc,quic3,freqhorn} or recurrence
solvers \cite{viap,aligators}).  The combination of these factors
often reduces reasoning about a program with multiple loops to
reasoning about one with fewer (sometimes even none) and ``simpler''
loops, thereby simplifying proof goals.  In this paper, we demonstrate
this, focusing on programs with sequentially composed, but non-nested
loops.

\subsection{Motivating Examples}

We present a couple of examples to illustrate our technique and
showcase its salient features.  The first example presents the basic
ideas behind full-program induction.  The second example highlights
various nuanced features of the technique and is used as a running
example in this paper.

\begin{figure}[h]
\begin{alltt}
// assume(\(\true\))

1. for (int t1=0; t1<N; t1=t1+1) \{
2.   if (t1==0) \{ A[t1] = 6; \}
3.   else \{ A[t1] = A[t1-1] + 6; \}
4. \}

5. for (int t2=0; t2<N; t2=t2+1) \{
6.   if (t2==0) \{ B[t2] = 1; \}
7.   else \{ B[t2] = B[t2-1] + A[t2-1]; \}
8. \}

9. for (int t3=0; t3<N; t3=t3+1) \{
10.  if (t3==0) \{ C[t3] = 0; \}
11.  else \{ C[t3] = C[t3-1] + B[t3-1]; \}
12.\}

// assert(\(\forall\)i \(\in\) [0,N), C[i] = i\(\cube\))
\end{alltt}
\caption{Original Hoare triple}
\label{fig:ex}
\end{figure}

Fig. \ref{fig:ex} shows an example of a Hoare triple, where the pre-
and post-conditions are specified using {\tt assume} and {\tt assert}
statements.  This triple effectively verifies the formula
$\sum\limits_{j=0}^{i-1} \left(1 + \sum\limits_{k=0}^{j-1}
6\cdot(k+1)\right) = i^3$ for all $i \in \{0 \ldots N-1\}$, and for
all $N > 0$.  Although each loop in Fig. \ref{fig:ex} is simple, their
sequential composition makes it difficult even for state-of-the-art
tools like {\viap} \cite{viap}, {\veriabs} \cite{veriabs20},
{\freqhorn} \cite{freqhorn}, {\tiler} \cite{sas17},
{\vaphor} \cite{vaphor}, or {\booster} \cite{booster} to prove the
post-condition correct.  In fact, none of the above tools succeed in
automatically proving the quantified post-condition in
Fig. \ref{fig:ex}.  In contrast, our technique \emph{full-program
induction} proves the post-condition in Fig. \ref{fig:ex} correct
within a few seconds.

\begin{figure}[h]
\begin{alltt}
// assume(\(\true\))

1.  A[0] = 6;
2.  B[0] = 1;
3.  C[0] = 0;

// assert((C[0] = 0\(\cube\)) \(\wedge\)
//        (B[0] = 1\(\cube\) - 0\(\cube\)) \(\wedge\)
//        (A[0] = 2\(\cube\) - 2\(\times\)1\(\cube\) + 0\(\cube\)))
\end{alltt}
\caption{Base-case Hoare triple}
\label{fig:ex-base}
\end{figure}

Full-program induction reduces checking the validity of the Hoare
triple in Fig. \ref{fig:ex} to checking the validity of two
``simpler'' Hoare triples, represented in Figs. \ref{fig:ex-base}
and \ref{fig:ex-ind}.  The base case of our inductive reasoning is
shown in Fig. \ref{fig:ex-base}, where every loop in the program is
statically unrolled a fixed number of times after instantiating the
program parameter $N$ to a small constant value (here $N=1$).  As the
induction hypothesis, we assume that the Hoare triple
$\{\varphi(N-1)\}$ $\;{\PP_{N-1}}\;$ $\{\psi(N-1)\}$ holds for values
of $N > 1$.  Note that this assumption does not relate to a specific
loop in the program, but to the entire program ${\PP_N}$.  For the
motivating example, the induction hypothesis states that the entire
Hoare triple in Fig. \ref{fig:ex}, after substituting $N$ with $N-1$,
holds.  Notice that the induction hypothesis is on the entire program
including all three loops and not on individual loops.  The inductive
step of the reasoning shown in Fig. \ref{fig:ex-ind} proves the
post-condition, by automatically generating the computation to be
performed after the program with parameter $N-1$ has executed and
strengthening the pre- and post-conditions using auxiliary predicates.
Note that all the program statements in Fig. \ref{fig:ex-ind} have
syntactic counterparts in Fig. \ref{fig:ex}, but this may not be the
case in general.  We conceptualize the computation in the inductive
step using the notions of {\em difference program} and {\em difference
pre-condition} in the later sections.  Effectively, we reasoned about
three sequentially composed loops in Fig. \ref{fig:ex} together,
without the need for any implicitly or explicitly specified loop
invariants.  We defer a discussion of how our technique computes these
Hoare triples and how auxiliary predicates are generated to
iteratively strengthen the pre- and post-conditions to
Sect. \ref{sec:algorithms}, where we present the details of our
algorithms.

\begin{figure}[h]
\begin{alltt}
// assume(
//  (N > 1) \(\wedge\) (C_Nm1[N-2] = (N-2)\(\cube\)) \(\wedge\)
//  (B_Nm1[N-2] = (N-1)\(\cube\) - (N-2)\(\cube\)) \(\wedge\)
//  (A_Nm1[N-2] = N\(\cube\) - 2\(\times\)(N-1)\(\cube\) + (N-2)\(\cube\))
// )

1.  A[N-1] = A_Nm1[N-2] + 6;
2.  B[N-1] = B_Nm1[N-2] + A_Nm1[N-2];
3.  C[N-1] = C_Nm1[N-2] + B_Nm1[N-2];

// assert(
//  (C[N-1] = (N-1)\(\cube\)) \(\wedge\)
//  (B[N-1] = N\(\cube\) - (N-1)\(\cube\)) \(\wedge\)
//  (A[N-1] = (N+1)\(\cube\) - 2\(\times\)N\(\cube\) + (N-1)\(\cube\))
// )
\end{alltt}
\caption{Inductive-step Hoare triple}
\label{fig:ex-ind}
\end{figure}

It is important to mention a few things here to highlight the
simplifications illustrated by the Hoare triples in
Figs. \ref{fig:ex-base} and \ref{fig:ex-ind} that resulted from the
application of the full-program induction technique on the problem in
Fig. \ref{fig:ex}.  First, the programs in Figs. \ref{fig:ex-base}
and \ref{fig:ex-ind} are loop-free.  Second, their pre- and
post-conditions are quantifier-free.  Third, the validity of these
Hoare triples (Figs. \ref{fig:ex-base} and \ref{fig:ex-ind}) can be
easily proved, e.g. by bounded model checking \cite{bmc} with a
back-end SMT solver like {\zthree} \cite{z3}.  Fourth, the value
computed in each iteration of each loop in Fig. \ref{fig:ex} is
data-dependent on previous iterations of the respective loops as well
as on the value computed in previous loops.  Even though none of these
loops can be trivially translated to a set of parallel assignments,
our method still succeeds in automating the inductive step of the
analysis.  Last, we did not require any specialized constraint solving
techniques like recurrence solving, theory of uninterpreted functions
or constrained Horn clause solving to verify these Hoare triples, thus
making our technique orthogonal to these approaches when proving
properties of array programs.


Now consider the Hoare triple shown in Fig. \ref{fig:ss}.  The program
updates a scalar variable {\tt S} and an array variable {\tt A}.  The
first loop adds the value of each element in array {\tt A} to variable
{\tt S}.  The second loop adds the value of {\tt S} to each element of
{\tt A}.  The last loop aggregates the updated content of {\tt A} in
{\tt S}.  The pre-condition $\varphi(N)$ is a universally quantified
formula on array {\tt A} stating that each element has the value $1$.
We need to establish the post-condition $\psi(N)$, which is a
predicate on {\tt S} and {\tt N}.  Note that the post-condition has
non-linear terms making it quite challenging to prove.  We will use
the Hoare triple in Fig. \ref{fig:ss} as our running example to
illustrate important aspects of the full-program induction technique.

\input{ex-ss}  

Since the program $\PP_N$ updates the same scalar variable {\tt S} and
the array {\tt A} in multiple sequentially composed loops, we rename
these scalars and arrays such that each loop in $\PP_N$ updates its own
copy of scalar variables and arrays.  This ensures that when
$\PP_{N-1}$ terminates we have access to the values of these variables
after each loop in the program.  The renamed program is shown in the
Hoare triple in Fig. \ref{fig:ex-ss-motex}(a).

In the base case of our inductive reasoning, we instantiate the
parameter $N$ to a small constant value (say $N=1$).  As a result,
every loop in the program $\PP_N$ in Fig. \ref{fig:ex-ss-motex}(a) can
be statically unrolled a fixed number of times.  The resulting Hoare
triple can be easily compiled to a first-order logic formula and
verified using an SMT solver.  As the induction hypothesis, we assume
that the Hoare triple $\{\varphi(N-1)\}$ $\;{\PP_{N-1}}\;$
$\{\psi(N-1)\}$, shown in Fig. \ref{fig:ex-ss-motex}(b), holds for
values of $N > 1$.  This Hoare triple is obtained by substituting $N$
with $N-1$ in the entire Hoare triple in
Fig. \ref{fig:ex-ss-motex}(a).

\input{ex-ss-motex}

The Hoare triple in Fig. \ref{fig:ex-ss-motex}(c) is computed for the
inductive step.  Intuitively, the difference program $\partial \PP_N$
recovers the effect of the computation in $\PP_N$ on all scalar
variables and arrays after the computation in $\PP_{N-1}$ has been
performed.  It includes the iterations of a loop in $\PP_N$ that are
missed by $\PP_{N-1}$.  When program statements are impervious to the
value of $N$, the values computed in such statements are the same in
$\PP_N$ and $\PP_{N-1}$, and hence, they may not need any modification.
However, $\partial \PP_N$ may contain code to ``rectify'' values of
variables and arrays that have different values at corresponding
statements in $\PP_N$ vis-a-vis $\PP_{N-1}$.  The code, possibly
consisting of loops, to rectify the values of variables and arrays is
further simplified whenever possible.  Consequently, not all program
statements of $\PP_N$ in Fig. \ref{fig:ex-ss-motex}(a) may have a
syntactic counterpart in Fig. \ref{fig:ex-ss-motex}(c) and vice-versa.
We present in detail the algorithms for the computation and
simplification of the difference program in Sects. \ref{sec:diff-prog}
and \ref{sec:simp-diff}.  The inductive step may not be immediately
established, in which case we strengthen the pre- and post-conditions
using automatically inferred auxiliary predicates as shown in
Fig. \ref{fig:ex-ss-motex}(c).

\subsection{Beyond Loop-Invariant based Proofs}

Techniques based on synthesis and use of loop invariants are popularly
used to reason about programs with loops.  These techniques have been
successfully applied to verify different classes of array manipulating
programs,
viz. \cite{Gopan,Halbwachs,Rival,ArrayCousotCL11,Gulwani,Srivastava09,Dirk07,Jhala,freqhorn}.
If we were to prove the assertion in Fig. \ref{fig:ex} using such
techniques, it would be necessary to use appropriate loop-specific
invariants for each of the three loops in Fig. \ref{fig:ex}. The
weakest loop invariants needed to prove the post-condition in this
example are: $\forall i \in [0,t1)\; (A[i] = 6i + 6)$ for the first
loop (lines $1$-$4$), $\forall j \in [0,t2)\; (B[j] = 3j^2 + 3j +
1) \wedge (A[j] = 6j + 6)$ for the second loop (lines $5$-$8$), and
$\forall k \in [0,t3)\; (C[k] = k^3) \wedge (B[k] = 3k^2 + 3k + 1)$
for the third loop (lines $9$-$12$).  Notice that these invariants are
quantified and have non-linear terms.  Similarly, the weakest loop
invariants needed to prove the post-condition for the program in
Fig. \ref{fig:ex-ss-motex}(a) are: $\forall j \in [0,i)\; (A[j] =
1) \wedge (S = j)$ for the first loop (lines $2$-$4$), $\forall k \in
[0,i)\; (A1[k] = N + 1) \wedge (A[k] = 1) \wedge (S = N)$ for the
second loop (lines $5$-$7$), and $\forall l \in [0,i)\; (A1[l] =
N+1) \wedge (S1 = l \times (N+1) + N)$ for the third loop (lines
$9$-$11$).

Unfortunately, automatically deriving such quantified non-linear
inductive invariants for each loop is far from trivial.
Template-based invariant generators, viz. \cite{houdini,daikon}, are
among the best-performers when generating such complex invariants.
However, their abilities are fundamentally limited by the set of
templates from which they choose.  We therefore choose not to depend
on inductive loop-invariants at all in our work.  Instead, we make use
of inductive pre- and post-conditions -- a notion that is related to,
yet significantly different from loop-specific invariants.
Specifically, inductive pre- and post-conditions are computed for the
entire program, possibly consisting of multiple loops, instead of for
each loop in the program.


As is clear from the discussion above, the primary difference between
a proof generated by an invariant synthesis technique and the proof
generated by our method is that we no longer need loop-specific safe
inductive invariants.  Instead, we generate and verify the Hoare
triples shown in Figs. \ref{fig:ex-base} and \ref{fig:ex-ind}
considering the entire program $\PP_N$ in Fig. \ref{fig:ex} and the
Hoare triple shown in Fig. \ref{fig:ex-ss-motex}(c) considering the
entire program $\PP_N$ in Fig. \ref{fig:ex-ss-motex}(a).
Automatically generating these Hoare triples in some cases may be more
difficult than automatically generating inductive invariants for each
loop and vice versa.  However, as demonstrated by the motivating
examples, there are several complex programs, where it may be easier
to generate these Hoare triples than compute safe inductive invariants
for individual loops.  It is a considerable challenge for verification
techniques to be able to automatically generate these invariants and
to the best of our knowledge none of the current state-of-the-art
techniques do so.

\subsection{Effectiveness of Full-Program Induction}

We have implemented the full-program induction technique in a
prototype tool called {\ourtool}.  Written in C++, the tool is built
on top of a compiler framework (LLVM/CLANG \cite{clang}) and uses an
off-the-shelf SMT solver ({\zthree} \cite{z3}) at the back-end.  Our
experiments show that the full-program induction technique is able to
solve several difficult problem instances, which other techniques
either fail to solve, or can solve only with the help of sophisticated
recurrence solvers.  {\ourtool} is significantly more efficient as
compared to other tools on a set of benchmarks.

Needless to say, each approach has its own strengths and limitations,
and the right choice always depends on the problem at hand.
Full-program induction is no exception, and despite its several
strengths, it has its own limitations, which we discuss in detail in
Sect.~\ref{sec:limitations}.

The full-program induction technique is orthogonal to other
verification approaches proposed in literature, making it suitable to
be a part of an arsenal of verification techniques.  It has already
been incorporated within a verification tool, namely
{\veriabs} \cite{veriabs20}.  Since the 2020 edition of the
international software verification competition (SV-COMP), {\veriabs}
invokes full-program induction (via our tool {\ourtool}) in its
pipeline of techniques for verifying programs with arrays from the set
of benchmarks in the verification competition
(refer \cite{veriabs20}).

\subsection{Primary Contributions of our Work}

This paper is a revised and extended version of \cite{tacas20}.  Our
main contributions can be summarized as follows:

\begin{enumerate}
\item We introduce \emph{full-program induction} as a technique for
      reasoning about assertions in programs with loops manipulating
      arrays with parametric size bounds. Full-program induction does
      not need loop-specific invariants in order to prove assertions,
      even when the program contains multiple sequentially composed
      loops.
\item We describe practical algorithms for performing \emph{full-program
      induction}.  We elaborate the generalized algorithms for
      computing the difference program and the difference
      pre-condition.
\item We present a new algorithm to compute a progress measure,
      based on the characteristics of the difference program. This
      gives a measure of how easy it is to prove the inductive step
      of our technique using constraint solving based techniques like
      bounded model checking.
\item We give rigorous proofs of correctness for the presented algorithms.
      We demonstrate these algorithms using a running example.
\item We present generalizations of the full-program induction technique
      to programs with multiple parameters and loops with increasing and/or
      decreasing loop counters.
\item We describe a prototype tool {\ourtool} that implements the
      algorithms for performing full-program induction, using (i) the
      compiler framework \textsc{LLVM/CLANG} for analysis and
      transformation of the input program and (ii) an off-the-shelf
      SMT solver, viz. {\zthree}, at the back-end to discharge
      verification conditions.
\item We present an extensive experimental evaluation on a large suite
      of benchmarks that manipulate arrays.  {\ourtool} outperforms
      the state-of-the-art tools {\viap}, {\veriabs}, {\booster},
      {\vaphor}, and {\freqhorn}, on the set of benchmark programs.
\end{enumerate}

Several contributions listed above are beyond those presented in
\cite{tacas20}.  These include the contributions $2$, $3$, $4$, $5$,
and $7$.

The remainder of the paper is structured as follows.  In
Sect. \ref{sec:overview}, we give a formal overview of the
full-program induction technique.  Sect. \ref{sec:prelim} presents the
syntactic restrictions on the program as well as the pre- and
post-conditions and the representation of programs as control flow
graphs.  Sect. \ref{sec:diffcomp} discusses the algorithms for
computing the difference program and the difference pre-condition, as
well as the pre-requisite analyses and transformations.  In
Sect. \ref{sec:algorithms}, we present the algorithms for full-program
induction, prove their correctness and demonstrate each algorithm on
the running example.  In Sect. \ref{sec:progress}, we give an
algorithm to check whether the recursive application of our technique
will eventually be able to verify the given program.
Sect. \ref{sec:generalized} talks of the generalizations of our
technique in different settings.  Sect. \ref{sec:experiments} presents
the implementation of our technique in {\ourtool}, its evaluation on a
set of benchmarks and comparison vis-a-vis state-of-the-art tools.  In
Sect. \ref{sec:related}, we discuss the related techniques from
literature.  Finally, Sect. \ref{sec:conc} presents concluding remarks
on our work and possible future directions.

%% file: ex-ss.tex

\begin{figure}[h]
\begin{alltt}
// assume(\(\forall\)i\(\in\)[0,N) A[i] = 1)

1.  S = 0;
2.  for(i=0; i<N; i++) \{
3.    S = S + A[i];
4.  \}

5.  for(i=0; i<N; i++) \{
6.    A[i] = A[i] + S;
7.  \}

8.  for(i=0; i<N; i++) \{
9.    S = S + A[i];
10. \}

// assert(S = N \(\times\) (N+2))
\end{alltt}
\caption{Running example}
\label{fig:ss}
\end{figure}


%% file: ex-ss-motex.tex
\begin{figure*}[tb]
\begin{minipage}{0.30\linewidth}
\begin{alltt}
// assume(\(\forall\)i\(\in\)[0,N) A[i]=1)

1.  S = 0;
2.  for(i=0; i<N; i++) \{
3.    S = S + A[i];
4.  \}

5.  for(i=0; i<N; i++) \{
6.    A1[i] = A[i] + S;
7.  \}

8.  S1 = S;
9.  for(i=0; i<N; i++) \{
10.   S1 = S1 + A1[i];
11. \}

// assert(S1=N\(\times\)(N+2))
\end{alltt}
\begin{center}
  (a)
\end{center}
\end{minipage}
\vrule\vrule
\begin{minipage}{0.35\linewidth}
\begin{alltt}
// assume(\(\forall\)i\(\in\)[0,N-1) A[i]=1)

1.  S_Nm1 = 0;
2.  for(i=0; i<N-1; i++) \{
3.    S_Nm1 = S_Nm1 + A[i];
4.  \}

5.  for(i=0; i<N-1; i++) \{
6.    A1_Nm1[i] = A[i]+S_Nm1;
7.  \}

8.  S1 = S;
9.  for(i=0; i<N-1; i++) \{
10.   S1_Nm1=S1_Nm1+A1_Nm1[i];
11. \}

// assert(S1_Nm1=(N-1)\(\times\)(N+1))
\end{alltt}
\begin{center}
  (b)
\end{center}
\end{minipage}
\vrule\vrule
\begin{minipage}{0.35\linewidth}
\begin{alltt}
// assume(N>1 \(\wedge\) A[N-1]=1
// \(\wedge\) S1_Nm1=(N-1)\(\times\)(N+1)
// \(\wedge\) \(\forall\)i\(\in\)[0,N-1) A1_Nm1[i]=N
// \(\wedge\) S_Nm1=N-1)

1.  S = S_Nm1 + A[N-1];

2.  for(i=0; i<N-1; i++) \{
3.    A1[i] = A1_Nm1[i] + 1;
4.  \}
5.  A1[N-1] = A[N-1] + S;

6.  S1 = S1_Nm1 + A[N-1];
7.  S1 = S1 + (N-1);
8.  S1 = S1 + A1[N-1];

// assert(S1=N\(\times\)(N+2) \(\wedge\) S=N
// \(\wedge\) \(\forall\)i\(\in\)[0,N) A1[i]=N+1)
\end{alltt}
\begin{center}
  (c)
\end{center}
\end{minipage}
\vspace{1ex}
\caption{(a) Given Hoare triple on $\PP_N$ (b) induction hypothesis
  Hoare triple on $\PP_{N-1}$ and (c) inductive step Hoare triple on
  $\partial \PP_N$ after simplification and strengthening}
\label{fig:ex-ss-motex}
\end{figure*}

%% file: overview.tex
\input{htfig-pn}

We now elaborate on the core idea behind the {\em full-program
induction} technique.  Our goal is to check the validity of the
parameterized Hoare triple $\{\varphi(N)\}$ $\;\PP_{N}\;$
$\{\psi(N)\}$ for all $N > 0$.  A visual representation of this Hoare
triple is shown in Fig. \ref{fig:ht-pn}, where the clouds represent
(possibly quantified) formulas and boxes represent programs/code
fragments.

Intuitively, at a conceptual level, our approach works like any other
inductive reasoning technique.  However, the induction is over the
entire program, via the program parameter $N$, and not on the
individual loops in the program.

We first check the base case, where we verify that the parameterized
Hoare triple holds for some small values of $N$, say $0 < N \le M$.
We rely on an important, yet reasonable, assumption that can be stated
as follows: \emph{For every value of $N~(> 0)$, every loop in
${\PP_N}$ can be statically unrolled a number (say $f(N)$) of times
that depends only on $N$, to yield a loop-free program
$\widehat{\PP_N}$ that is semantically equivalent to ${\PP_N}$.}  Note
that this does not imply that reasoning about loops can be translated
into loop-free reasoning.  In general, $f(N)$ is a non-constant
function, and hence, the number of unrollings of loops in $\PP_N$ may
strongly depend on $N$.  In our experience, loops in a vast majority
of array manipulating programs (including Figs. \ref{fig:ex} and \ref{fig:ss}
and all our benchmarks) satisfy the above assumption.  Consequently, the base
case of our induction reduces to checking a Hoare triple for a
loop-free program.  Checking a Hoare triple for a loop-free program is
easily achieved by compiling the pre-condition, program and
post-condition into an SMT formula, whose (un)satisfiability can be
checked with an off-the-shelf back-end SMT solver.

\input{htfig-pnm1}

Next we hypothesize that $\{\varphi(N-1)\}$ $\;\PP_{N-1}\;$
$\{\psi(N-1)\}$ holds for some $N > M$, visually depicted in
Fig. \ref{fig:ht-pnm1}.  A few things are worth mentioning
here. First, the entire Hoare triple is assumed not just the formula
in the post-condition.  Second, the assumption is not on a specific
loop in the program, but the entire program $\PP_N$.  Third, the
change in the parameter from $N$ to $N-1$ is uniform across the entire
Hoare triple and not on a specific part there-off.

We then try to show that the induction hypothesis implies
$\{\varphi(N)\}$ $\;\PP_{N}\;$ $\{\psi(N)\}$.  While this sounds
simple in principle, there are several technical difficulties
en-route.  Our contribution lies in overcoming these difficulties
algorithmically for a large class of programs and assertions, thereby
making {\em full-program induction} a viable and competitive technique
for proving properties of array manipulating programs.

The inductive step is the most complex one, and is the focus of the
rest of the paper.  Recall that the inductive hypothesis asserts that
$\{\varphi(N-1)\}$ $\;{\PP_{N-1}}\;$ $\{\psi(N-1)\}$ is valid.  To
make use of this hypothesis in the inductive step, we must relate the
validity of $\{\varphi(N)\}$ $\;{\PP_N}\;$ $\{\psi(N)\}$ to that of
$\{\varphi(N-1)\}$ $\;{\PP_{N-1}}\;$ $\{\psi(N-1)\}$.  We propose
doing this, whenever possible, via two key notions -- that of
``difference'' program and ``difference'' pre-condition.

\input{htfig-diff}

Given a parameterized program ${\PP_N}$, intuitively the
``difference'' program ${\partial \PP_N}$ is one such that
$\{\varphi(N)\}$ $\;{\PP_N}\;$ $\{\psi(N)\}$ holds iff
$\{\varphi(N)\}$ $\; {\PP_{N-1}};{\partial \PP_N}\;$ $\{\psi(N)\}$
holds, where ``;'' denotes sequential composition.  Refer to
Fig. \ref{fig:ht-diff} for a visual representation of the Hoare triple
after the decomposition of $\PP_N$ into $\PP_{N-1}$ and
$\partial \PP_N$.  We will use this interpretation of a ``difference''
program in the subsequent parts of this paper.

\input{htfig-pequiv}

A simple way of ensuring the correctness of this transformation is by
having a difference program ${\partial \PP_N}$ such that the
sequential composition ${\PP_{N-1}};{\partial \PP_N}$ is semantically
equivalent to ${\PP_N}$.  Decomposition of $\PP_N$ into $\PP_{N-1}$
and $\partial \PP_N$ is visually depicted in Fig. \ref{fig:ht-pequiv}.
It ensures that, upon termination, same program state is reached by
both $\PP_N$ and ${\PP_{N-1}};{\partial \PP_N}$.  The given
post-condition may not be impacted by the entire program state, and
hence, the semantic equivalence alluded to here may not be required,
in general.  Thus, the semantic equivalence of the decomposition is a
strong condition.  It is referred here only for the ease of explaining
the inductive setup and for an intuitive demonstration of soundness of
the decomposition.  For the purposes of full-program induction
semantic equivalence is not really necessary, and we do not refer to
this interpretation of the ``difference'' program further due to its
restrictive nature.

\input{htfig-dpre}

The ``difference'' pre-condition ${\partial \varphi(N)}$ is a formula
such that the following conditions hold.

\begin{enumerate}
\item $\varphi(N)$ $\rightarrow$ $(\varphi(N-1)$ $\odot$
      ${\partial \varphi(N)})$, where the boolean operator $\odot$ is
      $\wedge$ when $\varphi(N)$ is a universally quantified formula
      and it is $\vee$ when $\varphi(N)$ is a existentially quantified
      formula. We depict this decomposition of $\varphi(N)$ into
      $\varphi(N-1)$ and ${\partial \varphi(N)}$ in
      Fig. \ref{fig:ht-dpre}.
\item The execution of $\PP_{N-1}$ does not affect the truth of
      ${\partial \varphi(N)}$.  This can be visualized using
      Fig. \ref{fig:ht-diff}, where the dashed line indicates the
      propagation of the difference pre-condition across $\PP_{N-1}$
      when it is not affected.
\end{enumerate}

Computing the ``difference'' program ${\partial \PP_N}$ and the
``difference'' pre-condition ${\partial \varphi(N)}$ is not easy in
general.  In Sect. \ref{sec:algorithms}, we discuss ways to overcome
these problems and challenges.

\input{htfig-indstep}

Assuming we have ${\partial \PP_N}$ and ${\partial \varphi(N)}$ with
the properties stated above, the proof obligation $\{\varphi(N)\}$
$\;\PP_N\;$ $\{\psi(N)\}$ can now be reduced to proving the Hoare
triples $\{\varphi(N-1)\}$ $\;\PP_{N-1}\;$ $\{\psi(N-1)\}$ and
$\{\psi(N-1) \wedge {\partial \varphi(N)}\}$ $\;{\partial \PP_{N}}\;$
$\{\psi(N)\}$ in the inductive step.  Both these Hoare triples can be
easily visualized from Fig. \ref{fig:ht-diff}.
The first Hoare triple follows from the inductive hypothesis
(Fig. \ref{fig:ht-pnm1}), and hence, is available for free.  Thus, the
inductive step is reduced to proving the second Hoare triple as shown
in the Fig. \ref{fig:ht-indstep}.

Proving the inductive step may require strengthening the
pre-condition, say by a formula $\ppre(N-1)$, in general.  Since we
are in the inductive step of mathematical induction, we formulate the
new proof sub-goal in such a case as
$\{(\psi(N-1) \wedge \ppre(N-1)) \wedge {\partial \varphi(N)}\}$
${\partial \PP_{N}}$ $\{\psi(N) \wedge \ppre(N)\}$.  While this is
somewhat reminiscent of loop invariants, observe that $\ppre(N)$
is \emph{not} really a loop-specific invariant.  Instead, it is
analogous to computing an invariant for the entire program, possibly
containing multiple loops.  Specifically, the above process
strengthens both the pre- and post-condition of $\{\psi(N-1) \wedge
{\partial \varphi(N)}\}$ ${\partial \PP_{N}}$ $\{\psi(N)\}$
simultaneously using $\ppre(N-1)$ and $\ppre(N)$, respectively.
Fig. \ref{fig:ht-wp} shows the Hoare triple after the strengthening
step.  The strengthened post-condition of the resulting Hoare triple
may, in turn, require a new pre-condition $\ppre'(N-1)$ to be
satisfied. This process of strengthening the pre- and post-conditions
of the Hoare triple involving $\partial{\PP_N}$ can be iterated until
a fix-point is reached, i.e. no further pre-conditions are needed for
the parameterized Hoare triple to hold.  While the fix-point was
quickly reached for all benchmarks we experimented with, we also
discuss how to handle cases where the above process may not converge
easily.  Note that since we effectively strengthen the pre-condition
of the Hoare triple in the inductive step, for the overall induction
to go through, it is also necessary to check that the strengthened
assertions hold at the end of each base case check.  Automatically
computing $\ppre(N)$ to strengthen the pre- and post-conditions of the
Hoare triple may not always be straight forward, especially when the
difference program ${\partial \PP_{N}}$ has loops.  In such cases, we
recursively apply our technique on the Hoare triple
$\{\psi(N-1) \wedge {\partial \varphi(N)}\}$ ${\partial \PP_{N}}$
$\{\psi(N)\}$ generated during the inductive step.  This helps our
technique converge when the generated
difference program has one or more loops.  We check if the recursive
invocation of our technique will yield beneficial results using a
progress measure influenced by several characteristics of the
difference program.

\input{htfig-wp}

The technique outlined above is called \emph{full-program induction},
and the following theorem is the basis for the soundness
of \emph{full-program induction}.

\begin{theorem}
\label{thm:full-prog-ind-sound}
Given $\{\varphi(N)\}$ $\;\PP_N\;$ $\{\psi(N)\}$, suppose the
following are true:
\begin{enumerate}
\item For $N > 1$, $\{\varphi(N)\}$ $\;\PP_{N-1};{\partial \PP_N}\;$
      $\{\psi(N)\}$ holds iff $\{\varphi(N)\}$ $\;\PP_{N}\;$
      $\{\psi(N)\}$ holds.
\item For $N > 1$, there exists a formula ${\partial \varphi(N)}$ such
      that
      \begin{itemize}
        \item[(a)] ${\partial \varphi(N)}$ doesn't refer to any program
        variable or array element modified in $\PP_{N-1}$, and
        \item[(b)] $\varphi(N)$ $\rightarrow$ $\varphi(N-1)$ $\wedge$
        ${\partial \varphi(N)}$.
      \end{itemize}
\item There exists an integer $M \ge 1$ and a parameterized formula $\ppre(M)$
      such that
      \begin{itemize}
        \item[(a)] $\{\varphi(N)\}$ $\;\PP_N\;$ $\{\psi(N)\}$ holds for $0 < N \le M$,
        \item[(b)] $\{\varphi(M)\}$ $\;\PP_M\;$ $\{\psi(M)\wedge\ppre(M)\}$ holds, and
        \item[(c)] $\{\psi(N-1) \wedge \ppre(N-1) \wedge {\partial \varphi(N)}\}$
        $\;{\partial \PP_N}\;$ $\{\psi(N) \wedge \ppre(N)\}$ holds for $N > M$.
      \end{itemize}
\end{enumerate}
Then $\{\varphi_N\}$ $\;\PP_N\;$ $\{\psi_N\}$ holds for all $N \ge 1$.
\end{theorem}

\begin{proof}
For $0 < N \le M$, condition 3(a) (the base case) ensures that
$\{\varphi(N)\}$ $\;\PP_N\;$ $\{\psi(N)\}$ holds.  For $N > M$, note
that by virtue of conditions 1 and 2(b), $\{\varphi(N)\}$ $\;\PP_N\;$
$\{\psi(N)\}$ holds if $\{\varphi(N-1) \wedge {\partial \varphi(N)}\}$
$\;\PP_{N-1};{\partial \PP_N}\;$ $\{\psi(N) \wedge \ppre(N)\}$ holds.
With $\psi(N-1) \wedge \ppre(N-1)$ as a mid-condition, and by virtue
of condition 2(a), the latter Hoare triple holds for $N > M$ if
$\{\varphi(M)\}$ $\;\PP_{M}\;$ $\{\psi(M) \wedge \ppre(M)\}$ holds and
$\{\psi(N-1) \wedge \ppre(N-1) \wedge {\partial \varphi(N)}\}$
$\;{\partial \PP_N}\;$ $\{\psi(N) \wedge \ppre(N)\}$ holds for all $N
> M$.  Both these triples are seen to hold by virtue of conditions
3(b) and (c).
\end{proof}

%% file: htfig-pn.tex
\begin{wrapfigure}[15]{r}{0.13\textwidth}
\centering

\resizebox{0.13\textwidth}{!}{

\tikzset{every picture/.style={line width=0.75pt}} 

\begin{tikzpicture}[x=0.75pt,y=0.75pt,yscale=-1,xscale=1]

  \draw  [line width=3, fill=white] (24.46,33.67) .. controls (23.63,27.14) and (26.38,20.67) .. (31.55,17.02) .. controls (36.72,13.36) and (43.4,13.15) .. (48.76,16.49) .. controls (50.66,12.69) and (54.13,10.07) .. (58.14,9.41) .. controls (62.14,8.76) and (66.19,10.15) .. (69.08,13.17) .. controls (70.69,9.73) and (73.87,7.41) .. (77.48,7.05) .. controls (81.09,6.69) and (84.62,8.33) .. (86.81,11.39) .. controls (89.74,7.74) and (94.38,6.2) .. (98.75,7.44) .. controls (103.12,8.68) and (106.41,12.48) .. (107.21,17.19) .. controls (110.79,18.22) and (113.78,20.86) .. (115.39,24.41) .. controls (117.01,27.96) and (117.09,32.09) .. (115.63,35.71) .. controls (119.16,40.58) and (119.98,47.07) .. (117.8,52.76) .. controls (115.61,58.44) and (110.75,62.47) .. (105.02,63.34) .. controls (104.98,68.68) and (102.22,73.58) .. (97.8,76.15) .. controls (93.39,78.72) and (88.01,78.56) .. (83.74,75.73) .. controls (81.92,82.13) and (76.8,86.83) .. (70.59,87.81) .. controls (64.38,88.79) and (58.2,85.88) .. (54.71,80.33) .. controls (50.43,83.06) and (45.3,83.85) .. (40.47,82.51) .. controls (35.64,81.18) and (31.53,77.82) .. (29.04,73.22) .. controls (24.67,73.76) and (20.45,71.36) .. (18.46,67.2) .. controls (16.48,63.04) and (17.16,58.02) .. (20.17,54.62) .. controls (16.27,52.18) and (14.28,47.35) .. (15.24,42.64) .. controls (16.2,37.94) and (19.88,34.42) .. (24.38,33.92) ;
\draw  [line width=3, fill=white] (20.17,54.62) .. controls (22.01,55.77) and (24.13,56.29) .. (26.26,56.11)(29.04,73.22) .. controls (29.96,73.1) and (30.85,72.86) .. (31.71,72.5)(54.71,80.33) .. controls (54.07,79.3) and (53.53,78.21) .. (53.1,77.06)(83.74,75.73) .. controls (84.07,74.57) and (84.29,73.37) .. (84.38,72.15)(105.02,63.34) .. controls (105.06,57.66) and (102.02,52.46) .. (97.2,49.97)(115.63,35.71) .. controls (114.85,37.65) and (113.66,39.36) .. (112.15,40.73)(107.21,17.19) .. controls (107.35,17.97) and (107.41,18.76) .. (107.4,19.56)(86.81,11.39) .. controls (86.08,12.3) and (85.48,13.31) .. (85.03,14.41)(69.08,13.17) .. controls (68.69,14) and (68.4,14.87) .. (68.21,15.77)(48.76,16.49) .. controls (49.89,17.19) and (50.94,18.04) .. (51.88,19.01)(24.46,33.67) .. controls (24.58,34.57) and (24.76,35.46) .. (25.01,36.33) ;
\draw  [line width=3, fill=white] (24.46,230.67) .. controls (23.63,224.14) and (26.38,217.67) .. (31.55,214.02) .. controls (36.72,210.36) and (43.4,210.15) .. (48.76,213.49) .. controls (50.66,209.69) and (54.13,207.07) .. (58.14,206.41) .. controls (62.14,205.76) and (66.19,207.15) .. (69.08,210.17) .. controls (70.69,206.73) and (73.87,204.41) .. (77.48,204.05) .. controls (81.09,203.69) and (84.62,205.33) .. (86.81,208.39) .. controls (89.74,204.74) and (94.38,203.2) .. (98.75,204.44) .. controls (103.12,205.68) and (106.41,209.48) .. (107.21,214.19) .. controls (110.79,215.22) and (113.78,217.86) .. (115.39,221.41) .. controls (117.01,224.96) and (117.09,229.09) .. (115.63,232.71) .. controls (119.16,237.58) and (119.98,244.07) .. (117.8,249.76) .. controls (115.61,255.44) and (110.75,259.47) .. (105.02,260.34) .. controls (104.98,265.68) and (102.22,270.58) .. (97.8,273.15) .. controls (93.39,275.72) and (88.01,275.56) .. (83.74,272.73) .. controls (81.92,279.13) and (76.8,283.83) .. (70.59,284.81) .. controls (64.38,285.79) and (58.2,282.88) .. (54.71,277.33) .. controls (50.43,280.06) and (45.3,280.85) .. (40.47,279.51) .. controls (35.64,278.18) and (31.53,274.82) .. (29.04,270.22) .. controls (24.67,270.76) and (20.45,268.36) .. (18.46,264.2) .. controls (16.48,260.04) and (17.16,255.02) .. (20.17,251.62) .. controls (16.27,249.18) and (14.28,244.35) .. (15.24,239.64) .. controls (16.2,234.94) and (19.88,231.42) .. (24.38,230.92) ;
\draw  [line width=3, fill=white] (20.17,251.62) .. controls (22.01,252.77) and (24.13,253.29) .. (26.26,253.11)(29.04,270.22) .. controls (29.96,270.1) and (30.85,269.86) .. (31.71,269.5)(54.71,277.33) .. controls (54.07,276.3) and (53.53,275.21) .. (53.1,274.06)(83.74,272.73) .. controls (84.07,271.57) and (84.29,270.37) .. (84.38,269.15)(105.02,260.34) .. controls (105.06,254.66) and (102.02,249.46) .. (97.2,246.97)(115.63,232.71) .. controls (114.85,234.65) and (113.66,236.36) .. (112.15,237.73)(107.21,214.19) .. controls (107.35,214.97) and (107.41,215.76) .. (107.4,216.56)(86.81,208.39) .. controls (86.08,209.3) and (85.48,210.31) .. (85.03,211.41)(69.08,210.17) .. controls (68.69,211) and (68.4,211.87) .. (68.21,212.77)(48.76,213.49) .. controls (49.89,214.19) and (50.94,215.04) .. (51.88,216.01)(24.46,230.67) .. controls (24.58,231.57) and (24.76,232.46) .. (25.01,233.33) ;
\draw  [line width=3, fill=white, draw=black] (16,111) -- (123,111) -- (123,180) -- (16,180) -- cycle;

\draw (69,44) node  [font=\LARGE]  {$\boldsymbol{\varphi}${\bf(N)}};
\draw (68,241) node  [font=\LARGE]  {$\boldsymbol{\psi}${\bf(N)}};
\draw (71,147) node  [font=\LARGE]  {$\mathbf{P_{N}}$};

\end{tikzpicture}

}
\caption{Goal}
\label{fig:ht-pn}
\end{wrapfigure}

%% file: htfig-pnm1.tex
\begin{wrapfigure}[15]{r}{0.16\textwidth}
\centering

\resizebox{0.13\textwidth}{!}{

\tikzset{every picture/.style={line width=0.75pt}} 

\begin{tikzpicture}[x=0.75pt,y=0.75pt,yscale=-1,xscale=1]

  \draw  [line width=3, fill=white] (24.46,33.67) .. controls (23.63,27.14) and (26.38,20.67) .. (31.55,17.02) .. controls (36.72,13.36) and (43.4,13.15) .. (48.76,16.49) .. controls (50.66,12.69) and (54.13,10.07) .. (58.14,9.41) .. controls (62.14,8.76) and (66.19,10.15) .. (69.08,13.17) .. controls (70.69,9.73) and (73.87,7.41) .. (77.48,7.05) .. controls (81.09,6.69) and (84.62,8.33) .. (86.81,11.39) .. controls (89.74,7.74) and (94.38,6.2) .. (98.75,7.44) .. controls (103.12,8.68) and (106.41,12.48) .. (107.21,17.19) .. controls (110.79,18.22) and (113.78,20.86) .. (115.39,24.41) .. controls (117.01,27.96) and (117.09,32.09) .. (115.63,35.71) .. controls (119.16,40.58) and (119.98,47.07) .. (117.8,52.76) .. controls (115.61,58.44) and (110.75,62.47) .. (105.02,63.34) .. controls (104.98,68.68) and (102.22,73.58) .. (97.8,76.15) .. controls (93.39,78.72) and (88.01,78.56) .. (83.74,75.73) .. controls (81.92,82.13) and (76.8,86.83) .. (70.59,87.81) .. controls (64.38,88.79) and (58.2,85.88) .. (54.71,80.33) .. controls (50.43,83.06) and (45.3,83.85) .. (40.47,82.51) .. controls (35.64,81.18) and (31.53,77.82) .. (29.04,73.22) .. controls (24.67,73.76) and (20.45,71.36) .. (18.46,67.2) .. controls (16.48,63.04) and (17.16,58.02) .. (20.17,54.62) .. controls (16.27,52.18) and (14.28,47.35) .. (15.24,42.64) .. controls (16.2,37.94) and (19.88,34.42) .. (24.38,33.92) ;
\draw  [line width=3, fill=white] (20.17,54.62) .. controls (22.01,55.77) and (24.13,56.29) .. (26.26,56.11)(29.04,73.22) .. controls (29.96,73.1) and (30.85,72.86) .. (31.71,72.5)(54.71,80.33) .. controls (54.07,79.3) and (53.53,78.21) .. (53.1,77.06)(83.74,75.73) .. controls (84.07,74.57) and (84.29,73.37) .. (84.38,72.15)(105.02,63.34) .. controls (105.06,57.66) and (102.02,52.46) .. (97.2,49.97)(115.63,35.71) .. controls (114.85,37.65) and (113.66,39.36) .. (112.15,40.73)(107.21,17.19) .. controls (107.35,17.97) and (107.41,18.76) .. (107.4,19.56)(86.81,11.39) .. controls (86.08,12.3) and (85.48,13.31) .. (85.03,14.41)(69.08,13.17) .. controls (68.69,14) and (68.4,14.87) .. (68.21,15.77)(48.76,16.49) .. controls (49.89,17.19) and (50.94,18.04) .. (51.88,19.01)(24.46,33.67) .. controls (24.58,34.57) and (24.76,35.46) .. (25.01,36.33) ;
\draw  [line width=3, fill=white] (24.46,230.67) .. controls (23.63,224.14) and (26.38,217.67) .. (31.55,214.02) .. controls (36.72,210.36) and (43.4,210.15) .. (48.76,213.49) .. controls (50.66,209.69) and (54.13,207.07) .. (58.14,206.41) .. controls (62.14,205.76) and (66.19,207.15) .. (69.08,210.17) .. controls (70.69,206.73) and (73.87,204.41) .. (77.48,204.05) .. controls (81.09,203.69) and (84.62,205.33) .. (86.81,208.39) .. controls (89.74,204.74) and (94.38,203.2) .. (98.75,204.44) .. controls (103.12,205.68) and (106.41,209.48) .. (107.21,214.19) .. controls (110.79,215.22) and (113.78,217.86) .. (115.39,221.41) .. controls (117.01,224.96) and (117.09,229.09) .. (115.63,232.71) .. controls (119.16,237.58) and (119.98,244.07) .. (117.8,249.76) .. controls (115.61,255.44) and (110.75,259.47) .. (105.02,260.34) .. controls (104.98,265.68) and (102.22,270.58) .. (97.8,273.15) .. controls (93.39,275.72) and (88.01,275.56) .. (83.74,272.73) .. controls (81.92,279.13) and (76.8,283.83) .. (70.59,284.81) .. controls (64.38,285.79) and (58.2,282.88) .. (54.71,277.33) .. controls (50.43,280.06) and (45.3,280.85) .. (40.47,279.51) .. controls (35.64,278.18) and (31.53,274.82) .. (29.04,270.22) .. controls (24.67,270.76) and (20.45,268.36) .. (18.46,264.2) .. controls (16.48,260.04) and (17.16,255.02) .. (20.17,251.62) .. controls (16.27,249.18) and (14.28,244.35) .. (15.24,239.64) .. controls (16.2,234.94) and (19.88,231.42) .. (24.38,230.92) ;
\draw  [line width=3, fill=white] (20.17,251.62) .. controls (22.01,252.77) and (24.13,253.29) .. (26.26,253.11)(29.04,270.22) .. controls (29.96,270.1) and (30.85,269.86) .. (31.71,269.5)(54.71,277.33) .. controls (54.07,276.3) and (53.53,275.21) .. (53.1,274.06)(83.74,272.73) .. controls (84.07,271.57) and (84.29,270.37) .. (84.38,269.15)(105.02,260.34) .. controls (105.06,254.66) and (102.02,249.46) .. (97.2,246.97)(115.63,232.71) .. controls (114.85,234.65) and (113.66,236.36) .. (112.15,237.73)(107.21,214.19) .. controls (107.35,214.97) and (107.41,215.76) .. (107.4,216.56)(86.81,208.39) .. controls (86.08,209.3) and (85.48,210.31) .. (85.03,211.41)(69.08,210.17) .. controls (68.69,211) and (68.4,211.87) .. (68.21,212.77)(48.76,213.49) .. controls (49.89,214.19) and (50.94,215.04) .. (51.88,216.01)(24.46,230.67) .. controls (24.58,231.57) and (24.76,232.46) .. (25.01,233.33) ;
\draw  [line width=3, fill=white, draw=black] (16,111) -- (123,111) -- (123,180) -- (16,180) -- cycle;

\draw (69,44) node  [font=\LARGE]  {$\boldsymbol{\varphi}${\bf(N-1)}};
\draw (68,241) node  [font=\LARGE]  {$\boldsymbol{\psi}${\bf(N-1)}};
\draw (71,147) node  [font=\LARGE]  {$\mathbf{P_{N\text{-}1}}$};

\end{tikzpicture}

}
\caption{Hypothesis}
\label{fig:ht-pnm1}
\end{wrapfigure}
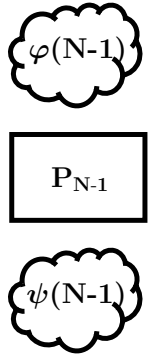

%% file: htfig-diff.tex
\begin{wrapfigure}[20]{r}{0.25\textwidth}
\centering

\resizebox{0.24\textwidth}{!}{

\tikzset{every picture/.style={line width=0.75pt}} 

\begin{tikzpicture}[x=0.75pt,y=0.75pt,yscale=-1,xscale=1]

  \draw  [line width=3, fill=white] (44.46,36.67) .. controls (43.63,30.14) and (46.38,23.67) .. (51.55,20.02) .. controls (56.72,16.36) and (63.4,16.15) .. (68.76,19.49) .. controls (70.66,15.69) and (74.13,13.07) .. (78.14,12.41) .. controls (82.14,11.76) and (86.19,13.15) .. (89.08,16.17) .. controls (90.69,12.73) and (93.87,10.41) .. (97.48,10.05) .. controls (101.09,9.69) and (104.62,11.33) .. (106.81,14.39) .. controls (109.74,10.74) and (114.38,9.2) .. (118.75,10.44) .. controls (123.12,11.68) and (126.41,15.48) .. (127.21,20.19) .. controls (130.79,21.22) and (133.78,23.86) .. (135.39,27.41) .. controls (137.01,30.96) and (137.09,35.09) .. (135.63,38.71) .. controls (139.16,43.58) and (139.98,50.07) .. (137.8,55.76) .. controls (135.61,61.44) and (130.75,65.47) .. (125.02,66.34) .. controls (124.98,71.68) and (122.22,76.58) .. (117.8,79.15) .. controls (113.39,81.72) and (108.01,81.56) .. (103.74,78.73) .. controls (101.92,85.13) and (96.8,89.83) .. (90.59,90.81) .. controls (84.38,91.79) and (78.2,88.88) .. (74.71,83.33) .. controls (70.43,86.06) and (65.3,86.85) .. (60.47,85.51) .. controls (55.64,84.18) and (51.53,80.82) .. (49.04,76.22) .. controls (44.67,76.76) and (40.45,74.36) .. (38.46,70.2) .. controls (36.48,66.04) and (37.16,61.02) .. (40.17,57.62) .. controls (36.27,55.18) and (34.28,50.35) .. (35.24,45.64) .. controls (36.2,40.94) and (39.88,37.42) .. (44.38,36.92) ;
  \draw  [line width=3, fill=white] (40.17,57.62) .. controls (42.01,58.77) and (44.13,59.29) .. (46.26,59.11)(49.04,76.22) .. controls (49.96,76.1) and (50.85,75.86) .. (51.71,75.5)(74.71,83.33) .. controls (74.07,82.3) and (73.53,81.21) .. (73.1,80.06)(103.74,78.73) .. controls (104.07,77.57) and (104.29,76.37) .. (104.38,75.15)(125.02,66.34) .. controls (125.06,60.66) and (122.02,55.46) .. (117.2,52.97)(135.63,38.71) .. controls (134.85,40.65) and (133.66,42.36) .. (132.15,43.73)(127.21,20.19) .. controls (127.35,20.97) and (127.41,21.76) .. (127.4,22.56)(106.81,14.39) .. controls (106.08,15.3) and (105.48,16.31) .. (105.03,17.41)(89.08,16.17) .. controls (88.69,17) and (88.4,17.87) .. (88.21,18.77)(68.76,19.49) .. controls (69.89,20.19) and (70.94,21.04) .. (71.88,22.01)(44.46,36.67) .. controls (44.58,37.57) and (44.76,38.46) .. (45.01,39.33) ;
  \draw  [line width=3, fill=white] (44.46,233.67) .. controls (43.63,227.14) and (46.38,220.67) .. (51.55,217.02) .. controls (56.72,213.36) and (63.4,213.15) .. (68.76,216.49) .. controls (70.66,212.69) and (74.13,210.07) .. (78.14,209.41) .. controls (82.14,208.76) and (86.19,210.15) .. (89.08,213.17) .. controls (90.69,209.73) and (93.87,207.41) .. (97.48,207.05) .. controls (101.09,206.69) and (104.62,208.33) .. (106.81,211.39) .. controls (109.74,207.74) and (114.38,206.2) .. (118.75,207.44) .. controls (123.12,208.68) and (126.41,212.48) .. (127.21,217.19) .. controls (130.79,218.22) and (133.78,220.86) .. (135.39,224.41) .. controls (137.01,227.96) and (137.09,232.09) .. (135.63,235.71) .. controls (139.16,240.58) and (139.98,247.07) .. (137.8,252.76) .. controls (135.61,258.44) and (130.75,262.47) .. (125.02,263.34) .. controls (124.98,268.68) and (122.22,273.58) .. (117.8,276.15) .. controls (113.39,278.72) and (108.01,278.56) .. (103.74,275.73) .. controls (101.92,282.13) and (96.8,286.83) .. (90.59,287.81) .. controls (84.38,288.79) and (78.2,285.88) .. (74.71,280.33) .. controls (70.43,283.06) and (65.3,283.85) .. (60.47,282.51) .. controls (55.64,281.18) and (51.53,277.82) .. (49.04,273.22) .. controls (44.67,273.76) and (40.45,271.36) .. (38.46,267.2) .. controls (36.48,263.04) and (37.16,258.02) .. (40.17,254.62) .. controls (36.27,252.18) and (34.28,247.35) .. (35.24,242.64) .. controls (36.2,237.94) and (39.88,234.42) .. (44.38,233.92) ;
  \draw  [line width=3, fill=white] (40.17,254.62) .. controls (42.01,255.77) and (44.13,256.29) .. (46.26,256.11)(49.04,273.22) .. controls (49.96,273.1) and (50.85,272.86) .. (51.71,272.5)(74.71,280.33) .. controls (74.07,279.3) and (73.53,278.21) .. (73.1,277.06)(103.74,275.73) .. controls (104.07,274.57) and (104.29,273.37) .. (104.38,272.15)(125.02,263.34) .. controls (125.06,257.66) and (122.02,252.46) .. (117.2,249.97)(135.63,235.71) .. controls (134.85,237.65) and (133.66,239.36) .. (132.15,240.73)(127.21,217.19) .. controls (127.35,217.97) and (127.41,218.76) .. (127.4,219.56)(106.81,211.39) .. controls (106.08,212.3) and (105.48,213.31) .. (105.03,214.41)(89.08,213.17) .. controls (88.69,214) and (88.4,214.87) .. (88.21,215.77)(68.76,216.49) .. controls (69.89,217.19) and (70.94,218.04) .. (71.88,219.01)(44.46,233.67) .. controls (44.58,234.57) and (44.76,235.46) .. (45.01,236.33) ;
  \draw  [line width=3, fill=white] (36,114) -- (143,114) -- (143,183) -- (36,183) -- cycle ;
  \draw  [line width=3, fill=white] (44.46,413.67) .. controls (43.63,407.14) and (46.38,400.67) .. (51.55,397.02) .. controls (56.72,393.36) and (63.4,393.15) .. (68.76,396.49) .. controls (70.66,392.69) and (74.13,390.07) .. (78.14,389.41) .. controls (82.14,388.76) and (86.19,390.15) .. (89.08,393.17) .. controls (90.69,389.73) and (93.87,387.41) .. (97.48,387.05) .. controls (101.09,386.69) and (104.62,388.33) .. (106.81,391.39) .. controls (109.74,387.74) and (114.38,386.2) .. (118.75,387.44) .. controls (123.12,388.68) and (126.41,392.48) .. (127.21,397.19) .. controls (130.79,398.22) and (133.78,400.86) .. (135.39,404.41) .. controls (137.01,407.96) and (137.09,412.09) .. (135.63,415.71) .. controls (139.16,420.58) and (139.98,427.07) .. (137.8,432.76) .. controls (135.61,438.44) and (130.75,442.47) .. (125.02,443.34) .. controls (124.98,448.68) and (122.22,453.58) .. (117.8,456.15) .. controls (113.39,458.72) and (108.01,458.56) .. (103.74,455.73) .. controls (101.92,462.13) and (96.8,466.83) .. (90.59,467.81) .. controls (84.38,468.79) and (78.2,465.88) .. (74.71,460.33) .. controls (70.43,463.06) and (65.3,463.85) .. (60.47,462.51) .. controls (55.64,461.18) and (51.53,457.82) .. (49.04,453.22) .. controls (44.67,453.76) and (40.45,451.36) .. (38.46,447.2) .. controls (36.48,443.04) and (37.16,438.02) .. (40.17,434.62) .. controls (36.27,432.18) and (34.28,427.35) .. (35.24,422.64) .. controls (36.2,417.94) and (39.88,414.42) .. (44.38,413.92) ;
  \draw  [line width=3, fill=white] (40.17,434.62) .. controls (42.01,435.77) and (44.13,436.29) .. (46.26,436.11)(49.04,453.22) .. controls (49.96,453.1) and (50.85,452.86) .. (51.71,452.5)(74.71,460.33) .. controls (74.07,459.3) and (73.53,458.21) .. (73.1,457.06)(103.74,455.73) .. controls (104.07,454.57) and (104.29,453.37) .. (104.38,452.15)(125.02,443.34) .. controls (125.06,437.66) and (122.02,432.46) .. (117.2,429.97)(135.63,415.71) .. controls (134.85,417.65) and (133.66,419.36) .. (132.15,420.73)(127.21,397.19) .. controls (127.35,397.97) and (127.41,398.76) .. (127.4,399.56)(106.81,391.39) .. controls (106.08,392.3) and (105.48,393.31) .. (105.03,394.41)(89.08,393.17) .. controls (88.69,394) and (88.4,394.87) .. (88.21,395.77)(68.76,396.49) .. controls (69.89,397.19) and (70.94,398.04) .. (71.88,399.01)(44.46,413.67) .. controls (44.58,414.57) and (44.76,415.46) .. (45.01,416.33) ;
  \draw  [line width=3, fill=white] (36,309) -- (143,309) -- (143,363) -- (36,363) -- cycle ;
  \draw  [line width=3, fill=white] (179.46,36.67) .. controls (178.63,30.14) and (181.38,23.67) .. (186.55,20.02) .. controls (191.72,16.36) and (198.4,16.15) .. (203.76,19.49) .. controls (205.66,15.69) and (209.13,13.07) .. (213.14,12.41) .. controls (217.14,11.76) and (221.19,13.15) .. (224.08,16.17) .. controls (225.69,12.73) and (228.87,10.41) .. (232.48,10.05) .. controls (236.09,9.69) and (239.62,11.33) .. (241.81,14.39) .. controls (244.74,10.74) and (249.38,9.2) .. (253.75,10.44) .. controls (258.12,11.68) and (261.41,15.48) .. (262.21,20.19) .. controls (265.79,21.22) and (268.78,23.86) .. (270.39,27.41) .. controls (272.01,30.96) and (272.09,35.09) .. (270.63,38.71) .. controls (274.16,43.58) and (274.98,50.07) .. (272.8,55.76) .. controls (270.61,61.44) and (265.75,65.47) .. (260.02,66.34) .. controls (259.98,71.68) and (257.22,76.58) .. (252.8,79.15) .. controls (248.39,81.72) and (243.01,81.56) .. (238.74,78.73) .. controls (236.92,85.13) and (231.8,89.83) .. (225.59,90.81) .. controls (219.38,91.79) and (213.2,88.88) .. (209.71,83.33) .. controls (205.43,86.06) and (200.3,86.85) .. (195.47,85.51) .. controls (190.64,84.18) and (186.53,80.82) .. (184.04,76.22) .. controls (179.67,76.76) and (175.45,74.36) .. (173.46,70.2) .. controls (171.48,66.04) and (172.16,61.02) .. (175.17,57.62) .. controls (171.27,55.18) and (169.28,50.35) .. (170.24,45.64) .. controls (171.2,40.94) and (174.88,37.42) .. (179.38,36.92) ;
  \draw  [line width=3, fill=white] (175.17,57.62) .. controls (177.01,58.77) and (179.13,59.29) .. (181.26,59.11)(184.04,76.22) .. controls (184.96,76.1) and (185.85,75.86) .. (186.71,75.5)(209.71,83.33) .. controls (209.07,82.3) and (208.53,81.21) .. (208.1,80.06)(238.74,78.73) .. controls (239.07,77.57) and (239.29,76.37) .. (239.38,75.15)(260.02,66.34) .. controls (260.06,60.66) and (257.02,55.46) .. (252.2,52.97)(270.63,38.71) .. controls (269.85,40.65) and (268.66,42.36) .. (267.15,43.73)(262.21,20.19) .. controls (262.35,20.97) and (262.41,21.76) .. (262.4,22.56)(241.81,14.39) .. controls (241.08,15.3) and (240.48,16.31) .. (240.03,17.41)(224.08,16.17) .. controls (223.69,17) and (223.4,17.87) .. (223.21,18.77)(203.76,19.49) .. controls (204.89,20.19) and (205.94,21.04) .. (206.88,22.01)(179.46,36.67) .. controls (179.58,37.57) and (179.76,38.46) .. (180.01,39.33) ;
  \draw  [line width=3, fill=white] (179.46,232.67) .. controls (178.63,226.14) and (181.38,219.67) .. (186.55,216.02) .. controls (191.72,212.36) and (198.4,212.15) .. (203.76,215.49) .. controls (205.66,211.69) and (209.13,209.07) .. (213.14,208.41) .. controls (217.14,207.76) and (221.19,209.15) .. (224.08,212.17) .. controls (225.69,208.73) and (228.87,206.41) .. (232.48,206.05) .. controls (236.09,205.69) and (239.62,207.33) .. (241.81,210.39) .. controls (244.74,206.74) and (249.38,205.2) .. (253.75,206.44) .. controls (258.12,207.68) and (261.41,211.48) .. (262.21,216.19) .. controls (265.79,217.22) and (268.78,219.86) .. (270.39,223.41) .. controls (272.01,226.96) and (272.09,231.09) .. (270.63,234.71) .. controls (274.16,239.58) and (274.98,246.07) .. (272.8,251.76) .. controls (270.61,257.44) and (265.75,261.47) .. (260.02,262.34) .. controls (259.98,267.68) and (257.22,272.58) .. (252.8,275.15) .. controls (248.39,277.72) and (243.01,277.56) .. (238.74,274.73) .. controls (236.92,281.13) and (231.8,285.83) .. (225.59,286.81) .. controls (219.38,287.79) and (213.2,284.88) .. (209.71,279.33) .. controls (205.43,282.06) and (200.3,282.85) .. (195.47,281.51) .. controls (190.64,280.18) and (186.53,276.82) .. (184.04,272.22) .. controls (179.67,272.76) and (175.45,270.36) .. (173.46,266.2) .. controls (171.48,262.04) and (172.16,257.02) .. (175.17,253.62) .. controls (171.27,251.18) and (169.28,246.35) .. (170.24,241.64) .. controls (171.2,236.94) and (174.88,233.42) .. (179.38,232.92) ;
  \draw  [line width=3, fill=white] (175.17,253.62) .. controls (177.01,254.77) and (179.13,255.29) .. (181.26,255.11)(184.04,272.22) .. controls (184.96,272.1) and (185.85,271.86) .. (186.71,271.5)(209.71,279.33) .. controls (209.07,278.3) and (208.53,277.21) .. (208.1,276.06)(238.74,274.73) .. controls (239.07,273.57) and (239.29,272.37) .. (239.38,271.15)(260.02,262.34) .. controls (260.06,256.66) and (257.02,251.46) .. (252.2,248.97)(270.63,234.71) .. controls (269.85,236.65) and (268.66,238.36) .. (267.15,239.73)(262.21,216.19) .. controls (262.35,216.97) and (262.41,217.76) .. (262.4,218.56)(241.81,210.39) .. controls (241.08,211.3) and (240.48,212.31) .. (240.03,213.41)(224.08,212.17) .. controls (223.69,213) and (223.4,213.87) .. (223.21,214.77)(203.76,215.49) .. controls (204.89,216.19) and (205.94,217.04) .. (206.88,218.01)(179.46,232.67) .. controls (179.58,233.57) and (179.76,234.46) .. (180.01,235.33) ;
\draw [line width=3, draw=black]    (227,105) -- (227,117) ;
\draw [line width=3, draw=black]    (227,159) -- (227,171) ;
\draw [line width=3, draw=black]    (227,141) -- (227,153) ;
\draw [line width=3, draw=black]    (227,123) -- (227,135) ;
\draw [line width=3, draw=black]    (227,177) -- (227,189) ;

\draw (89,47) node  [font=\LARGE]  {$\boldsymbol{\varphi}${\bf(N-1)}};
\draw (89,244) node  [font=\LARGE]  {$\boldsymbol{\psi}${\bf(N-1)}};
\draw (91,150) node  [font=\LARGE]  {$\mathbf{P_{N\text{-}1}}$};
\draw (89,424) node  [font=\LARGE]  {$\boldsymbol{\psi}${\bf(N)}};
\draw (91,335) node  [font=\LARGE]  {$\boldsymbol{\partial} \mathbf{P_{N}}$};
\draw (224,47) node  [font=\LARGE]  {$\boldsymbol{\partial \varphi}${\bf(N)}};
\draw (224,243) node  [font=\LARGE]  {$\boldsymbol{\partial \varphi}${\bf(N)}};

\end{tikzpicture}

}
\caption{Transformations}
\label{fig:ht-diff}
\end{wrapfigure}
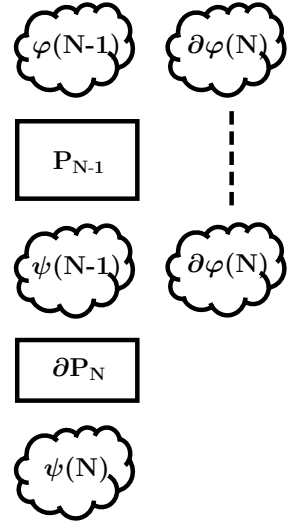

%% file: htfig-pequiv.tex
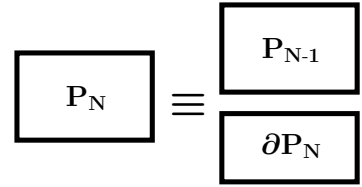
\begin{figure}[H]
\centering

\resizebox{0.30\textwidth}{!}{

\tikzset{every picture/.style={line width=0.75pt}} 

\begin{tikzpicture}[x=0.75pt,y=0.75pt,yscale=-1,xscale=1]

  \draw  [line width=3, fill=white] (176,13.31) -- (283,13.31) -- (283,82.31) -- (176,82.31) -- cycle ;
  \draw  [line width=3, fill=white] (176,99.31) -- (283,99.31) -- (283,153.31) -- (176,153.31) -- cycle ;
  \draw  [line width=3, fill=white] (14,51.31) -- (121,51.31) -- (121,120.31) -- (14,120.31) -- cycle ;

\draw (69,87) node  [font=\LARGE]  {$\mathbf{P_{N}}$};
\draw (231,49) node  [font=\LARGE]  {$\mathbf{P_{N\text{-}1}}$};
\draw (231,125) node  [font=\LARGE]  {$\boldsymbol{\partial}\mathbf{P_{N}}$};
\draw (133,82) node [anchor=north west][inner sep=0.75pt]  [font=\Huge]  {$\boldsymbol{\equiv}$};

\end{tikzpicture}

}
\caption{Decomposition of $\PP_N$ and semantic equivalence}
\label{fig:ht-pequiv}
\end{figure}

%% file: htfig-dpre.tex
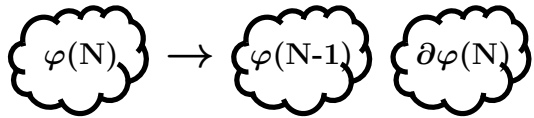
\begin{figure}[H]
\centering

\resizebox{0.45\textwidth}{!}{

\tikzset{every picture/.style={line width=0.75pt}} 

\begin{tikzpicture}[x=0.75pt,y=0.75pt,yscale=-1,xscale=1]

  \draw  [line width=3, fill=white] (20.46,35.67) .. controls (19.63,29.14) and (22.38,22.67) .. (27.55,19.02) .. controls (32.72,15.36) and (39.4,15.15) .. (44.76,18.49) .. controls (46.66,14.69) and (50.13,12.07) .. (54.14,11.41) .. controls (58.14,10.76) and (62.19,12.15) .. (65.08,15.17) .. controls (66.69,11.73) and (69.87,9.41) .. (73.48,9.05) .. controls (77.09,8.69) and (80.62,10.33) .. (82.81,13.39) .. controls (85.74,9.74) and (90.38,8.2) .. (94.75,9.44) .. controls (99.12,10.68) and (102.41,14.48) .. (103.21,19.19) .. controls (106.79,20.22) and (109.78,22.86) .. (111.39,26.41) .. controls (113.01,29.96) and (113.09,34.09) .. (111.63,37.71) .. controls (115.16,42.58) and (115.98,49.07) .. (113.8,54.76) .. controls (111.61,60.44) and (106.75,64.47) .. (101.02,65.34) .. controls (100.98,70.68) and (98.22,75.58) .. (93.8,78.15) .. controls (89.39,80.72) and (84.01,80.56) .. (79.74,77.73) .. controls (77.92,84.13) and (72.8,88.83) .. (66.59,89.81) .. controls (60.38,90.79) and (54.2,87.88) .. (50.71,82.33) .. controls (46.43,85.06) and (41.3,85.85) .. (36.47,84.51) .. controls (31.64,83.18) and (27.53,79.82) .. (25.04,75.22) .. controls (20.67,75.76) and (16.45,73.36) .. (14.46,69.2) .. controls (12.48,65.04) and (13.16,60.02) .. (16.17,56.62) .. controls (12.27,54.18) and (10.28,49.35) .. (11.24,44.64) .. controls (12.2,39.94) and (15.88,36.42) .. (20.38,35.92) ;
  \draw  [line width=3, fill=white] (16.17,56.62) .. controls (18.01,57.77) and (20.13,58.29) .. (22.26,58.11)(25.04,75.22) .. controls (25.96,75.1) and (26.85,74.86) .. (27.71,74.5)(50.71,82.33) .. controls (50.07,81.3) and (49.53,80.21) .. (49.1,79.06)(79.74,77.73) .. controls (80.07,76.57) and (80.29,75.37) .. (80.38,74.15)(101.02,65.34) .. controls (101.06,59.66) and (98.02,54.46) .. (93.2,51.97)(111.63,37.71) .. controls (110.85,39.65) and (109.66,41.36) .. (108.15,42.73)(103.21,19.19) .. controls (103.35,19.97) and (103.41,20.76) .. (103.4,21.56)(82.81,13.39) .. controls (82.08,14.3) and (81.48,15.31) .. (81.03,16.41)(65.08,15.17) .. controls (64.69,16) and (64.4,16.87) .. (64.21,17.77)(44.76,18.49) .. controls (45.89,19.19) and (46.94,20.04) .. (47.88,21.01)(20.46,35.67) .. controls (20.58,36.57) and (20.76,37.46) .. (21.01,38.33) ;
  \draw  [line width=3, fill=white] (191.46,35.67) .. controls (190.63,29.14) and (193.38,22.67) .. (198.55,19.02) .. controls (203.72,15.36) and (210.4,15.15) .. (215.76,18.49) .. controls (217.66,14.69) and (221.13,12.07) .. (225.14,11.41) .. controls (229.14,10.76) and (233.19,12.15) .. (236.08,15.17) .. controls (237.69,11.73) and (240.87,9.41) .. (244.48,9.05) .. controls (248.09,8.69) and (251.62,10.33) .. (253.81,13.39) .. controls (256.74,9.74) and (261.38,8.2) .. (265.75,9.44) .. controls (270.12,10.68) and (273.41,14.48) .. (274.21,19.19) .. controls (277.79,20.22) and (280.78,22.86) .. (282.39,26.41) .. controls (284.01,29.96) and (284.09,34.09) .. (282.63,37.71) .. controls (286.16,42.58) and (286.98,49.07) .. (284.8,54.76) .. controls (282.61,60.44) and (277.75,64.47) .. (272.02,65.34) .. controls (271.98,70.68) and (269.22,75.58) .. (264.8,78.15) .. controls (260.39,80.72) and (255.01,80.56) .. (250.74,77.73) .. controls (248.92,84.13) and (243.8,88.83) .. (237.59,89.81) .. controls (231.38,90.79) and (225.2,87.88) .. (221.71,82.33) .. controls (217.43,85.06) and (212.3,85.85) .. (207.47,84.51) .. controls (202.64,83.18) and (198.53,79.82) .. (196.04,75.22) .. controls (191.67,75.76) and (187.45,73.36) .. (185.46,69.2) .. controls (183.48,65.04) and (184.16,60.02) .. (187.17,56.62) .. controls (183.27,54.18) and (181.28,49.35) .. (182.24,44.64) .. controls (183.2,39.94) and (186.88,36.42) .. (191.38,35.92) ;
  \draw  [line width=3, fill=white] (187.17,56.62) .. controls (189.01,57.77) and (191.13,58.29) .. (193.26,58.11)(196.04,75.22) .. controls (196.96,75.1) and (197.85,74.86) .. (198.71,74.5)(221.71,82.33) .. controls (221.07,81.3) and (220.53,80.21) .. (220.1,79.06)(250.74,77.73) .. controls (251.07,76.57) and (251.29,75.37) .. (251.38,74.15)(272.02,65.34) .. controls (272.06,59.66) and (269.02,54.46) .. (264.2,51.97)(282.63,37.71) .. controls (281.85,39.65) and (280.66,41.36) .. (279.15,42.73)(274.21,19.19) .. controls (274.35,19.97) and (274.41,20.76) .. (274.4,21.56)(253.81,13.39) .. controls (253.08,14.3) and (252.48,15.31) .. (252.03,16.41)(236.08,15.17) .. controls (235.69,16) and (235.4,16.87) .. (235.21,17.77)(215.76,18.49) .. controls (216.89,19.19) and (217.94,20.04) .. (218.88,21.01)(191.46,35.67) .. controls (191.58,36.57) and (191.76,37.46) .. (192.01,38.33) ;
  \draw  [line width=3, fill=white] (313.46,35.67) .. controls (312.63,29.14) and (315.38,22.67) .. (320.55,19.02) .. controls (325.72,15.36) and (332.4,15.15) .. (337.76,18.49) .. controls (339.66,14.69) and (343.13,12.07) .. (347.14,11.41) .. controls (351.14,10.76) and (355.19,12.15) .. (358.08,15.17) .. controls (359.69,11.73) and (362.87,9.41) .. (366.48,9.05) .. controls (370.09,8.69) and (373.62,10.33) .. (375.81,13.39) .. controls (378.74,9.74) and (383.38,8.2) .. (387.75,9.44) .. controls (392.12,10.68) and (395.41,14.48) .. (396.21,19.19) .. controls (399.79,20.22) and (402.78,22.86) .. (404.39,26.41) .. controls (406.01,29.96) and (406.09,34.09) .. (404.63,37.71) .. controls (408.16,42.58) and (408.98,49.07) .. (406.8,54.76) .. controls (404.61,60.44) and (399.75,64.47) .. (394.02,65.34) .. controls (393.98,70.68) and (391.22,75.58) .. (386.8,78.15) .. controls (382.39,80.72) and (377.01,80.56) .. (372.74,77.73) .. controls (370.92,84.13) and (365.8,88.83) .. (359.59,89.81) .. controls (353.38,90.79) and (347.2,87.88) .. (343.71,82.33) .. controls (339.43,85.06) and (334.3,85.85) .. (329.47,84.51) .. controls (324.64,83.18) and (320.53,79.82) .. (318.04,75.22) .. controls (313.67,75.76) and (309.45,73.36) .. (307.46,69.2) .. controls (305.48,65.04) and (306.16,60.02) .. (309.17,56.62) .. controls (305.27,54.18) and (303.28,49.35) .. (304.24,44.64) .. controls (305.2,39.94) and (308.88,36.42) .. (313.38,35.92) ;
  \draw  [line width=3, fill=white] (309.17,56.62) .. controls (311.01,57.77) and (313.13,58.29) .. (315.26,58.11)(318.04,75.22) .. controls (318.96,75.1) and (319.85,74.86) .. (320.71,74.5)(343.71,82.33) .. controls (343.07,81.3) and (342.53,80.21) .. (342.1,79.06)(372.74,77.73) .. controls (373.07,76.57) and (373.29,75.37) .. (373.38,74.15)(394.02,65.34) .. controls (394.06,59.66) and (391.02,54.46) .. (386.2,51.97)(404.63,37.71) .. controls (403.85,39.65) and (402.66,41.36) .. (401.15,42.73)(396.21,19.19) .. controls (396.35,19.97) and (396.41,20.76) .. (396.4,21.56)(375.81,13.39) .. controls (375.08,14.3) and (374.48,15.31) .. (374.03,16.41)(358.08,15.17) .. controls (357.69,16) and (357.4,16.87) .. (357.21,17.77)(337.76,18.49) .. controls (338.89,19.19) and (339.94,20.04) .. (340.88,21.01)(313.46,35.67) .. controls (313.58,36.57) and (313.76,37.46) .. (314.01,38.33) ;

\draw (65,46) node  [font=\LARGE]  {$\boldsymbol{\varphi}${\bf(N)}};
\draw (234,46) node  [font=\LARGE]  {$\boldsymbol{\varphi}${\bf(N-1)}};
\draw (358,46) node  [font=\LARGE]  {$\boldsymbol{\partial \varphi}${\bf(N)}};
\draw (128,40) node [anchor=north west][inner sep=0.75pt]  [font=\Huge]  {$\boldsymbol{\rightarrow}$};

\end{tikzpicture}

}
\caption{Difference pre-condition}
\label{fig:ht-dpre}
\end{figure}

%% file: htfig-indstep.tex
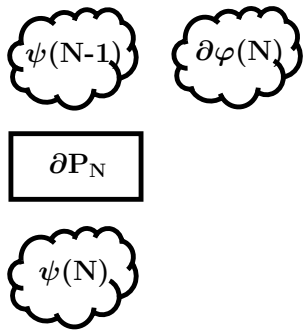
\begin{wrapfigure}[14]{r}{0.28\textwidth}
\centering

\resizebox{0.26\textwidth}{!}{

\tikzset{every picture/.style={line width=0.75pt}} 

\begin{tikzpicture}[x=0.75pt,y=0.75pt,yscale=-1,xscale=1]

  \draw  [line width=3, fill=white] (21.46,40.67) .. controls (20.63,34.14) and (23.38,27.67) .. (28.55,24.02) .. controls (33.72,20.36) and (40.4,20.15) .. (45.76,23.49) .. controls (47.66,19.69) and (51.13,17.07) .. (55.14,16.41) .. controls (59.14,15.76) and (63.19,17.15) .. (66.08,20.17) .. controls (67.69,16.73) and (70.87,14.41) .. (74.48,14.05) .. controls (78.09,13.69) and (81.62,15.33) .. (83.81,18.39) .. controls (86.74,14.74) and (91.38,13.2) .. (95.75,14.44) .. controls (100.12,15.68) and (103.41,19.48) .. (104.21,24.19) .. controls (107.79,25.22) and (110.78,27.86) .. (112.39,31.41) .. controls (114.01,34.96) and (114.09,39.09) .. (112.63,42.71) .. controls (116.16,47.58) and (116.98,54.07) .. (114.8,59.76) .. controls (112.61,65.44) and (107.75,69.47) .. (102.02,70.34) .. controls (101.98,75.68) and (99.22,80.58) .. (94.8,83.15) .. controls (90.39,85.72) and (85.01,85.56) .. (80.74,82.73) .. controls (78.92,89.13) and (73.8,93.83) .. (67.59,94.81) .. controls (61.38,95.79) and (55.2,92.88) .. (51.71,87.33) .. controls (47.43,90.06) and (42.3,90.85) .. (37.47,89.51) .. controls (32.64,88.18) and (28.53,84.82) .. (26.04,80.22) .. controls (21.67,80.76) and (17.45,78.36) .. (15.46,74.2) .. controls (13.48,70.04) and (14.16,65.02) .. (17.17,61.62) .. controls (13.27,59.18) and (11.28,54.35) .. (12.24,49.64) .. controls (13.2,44.94) and (16.88,41.42) .. (21.38,40.92) ;
  \draw  [line width=3, fill=white] (17.17,61.62) .. controls (19.01,62.77) and (21.13,63.29) .. (23.26,63.11)(26.04,80.22) .. controls (26.96,80.1) and (27.85,79.86) .. (28.71,79.5)(51.71,87.33) .. controls (51.07,86.3) and (50.53,85.21) .. (50.1,84.06)(80.74,82.73) .. controls (81.07,81.57) and (81.29,80.37) .. (81.38,79.15)(102.02,70.34) .. controls (102.06,64.66) and (99.02,59.46) .. (94.2,56.97)(112.63,42.71) .. controls (111.85,44.65) and (110.66,46.36) .. (109.15,47.73)(104.21,24.19) .. controls (104.35,24.97) and (104.41,25.76) .. (104.4,26.56)(83.81,18.39) .. controls (83.08,19.3) and (82.48,20.31) .. (82.03,21.41)(66.08,20.17) .. controls (65.69,21) and (65.4,21.87) .. (65.21,22.77)(45.76,23.49) .. controls (46.89,24.19) and (47.94,25.04) .. (48.88,26.01)(21.46,40.67) .. controls (21.58,41.57) and (21.76,42.46) .. (22.01,43.33) ;
  \draw  [line width=3, fill=white] (21.46,220.67) .. controls (20.63,214.14) and (23.38,207.67) .. (28.55,204.02) .. controls (33.72,200.36) and (40.4,200.15) .. (45.76,203.49) .. controls (47.66,199.69) and (51.13,197.07) .. (55.14,196.41) .. controls (59.14,195.76) and (63.19,197.15) .. (66.08,200.17) .. controls (67.69,196.73) and (70.87,194.41) .. (74.48,194.05) .. controls (78.09,193.69) and (81.62,195.33) .. (83.81,198.39) .. controls (86.74,194.74) and (91.38,193.2) .. (95.75,194.44) .. controls (100.12,195.68) and (103.41,199.48) .. (104.21,204.19) .. controls (107.79,205.22) and (110.78,207.86) .. (112.39,211.41) .. controls (114.01,214.96) and (114.09,219.09) .. (112.63,222.71) .. controls (116.16,227.58) and (116.98,234.07) .. (114.8,239.76) .. controls (112.61,245.44) and (107.75,249.47) .. (102.02,250.34) .. controls (101.98,255.68) and (99.22,260.58) .. (94.8,263.15) .. controls (90.39,265.72) and (85.01,265.56) .. (80.74,262.73) .. controls (78.92,269.13) and (73.8,273.83) .. (67.59,274.81) .. controls (61.38,275.79) and (55.2,272.88) .. (51.71,267.33) .. controls (47.43,270.06) and (42.3,270.85) .. (37.47,269.51) .. controls (32.64,268.18) and (28.53,264.82) .. (26.04,260.22) .. controls (21.67,260.76) and (17.45,258.36) .. (15.46,254.2) .. controls (13.48,250.04) and (14.16,245.02) .. (17.17,241.62) .. controls (13.27,239.18) and (11.28,234.35) .. (12.24,229.64) .. controls (13.2,224.94) and (16.88,221.42) .. (21.38,220.92) ;
  \draw  [line width=3, fill=white] (17.17,241.62) .. controls (19.01,242.77) and (21.13,243.29) .. (23.26,243.11)(26.04,260.22) .. controls (26.96,260.1) and (27.85,259.86) .. (28.71,259.5)(51.71,267.33) .. controls (51.07,266.3) and (50.53,265.21) .. (50.1,264.06)(80.74,262.73) .. controls (81.07,261.57) and (81.29,260.37) .. (81.38,259.15)(102.02,250.34) .. controls (102.06,244.66) and (99.02,239.46) .. (94.2,236.97)(112.63,222.71) .. controls (111.85,224.65) and (110.66,226.36) .. (109.15,227.73)(104.21,204.19) .. controls (104.35,204.97) and (104.41,205.76) .. (104.4,206.56)(83.81,198.39) .. controls (83.08,199.3) and (82.48,200.31) .. (82.03,201.41)(66.08,200.17) .. controls (65.69,201) and (65.4,201.87) .. (65.21,202.77)(45.76,203.49) .. controls (46.89,204.19) and (47.94,205.04) .. (48.88,206.01)(21.46,220.67) .. controls (21.58,221.57) and (21.76,222.46) .. (22.01,223.33) ;
  \draw  [line width=3, fill=white] (13,116) -- (120,116) -- (120,170) -- (13,170) -- cycle ;
  \draw  [line width=3, fill=white] (156.46,39.67) .. controls (155.63,33.14) and (158.38,26.67) .. (163.55,23.02) .. controls (168.72,19.36) and (175.4,19.15) .. (180.76,22.49) .. controls (182.66,18.69) and (186.13,16.07) .. (190.14,15.41) .. controls (194.14,14.76) and (198.19,16.15) .. (201.08,19.17) .. controls (202.69,15.73) and (205.87,13.41) .. (209.48,13.05) .. controls (213.09,12.69) and (216.62,14.33) .. (218.81,17.39) .. controls (221.74,13.74) and (226.38,12.2) .. (230.75,13.44) .. controls (235.12,14.68) and (238.41,18.48) .. (239.21,23.19) .. controls (242.79,24.22) and (245.78,26.86) .. (247.39,30.41) .. controls (249.01,33.96) and (249.09,38.09) .. (247.63,41.71) .. controls (251.16,46.58) and (251.98,53.07) .. (249.8,58.76) .. controls (247.61,64.44) and (242.75,68.47) .. (237.02,69.34) .. controls (236.98,74.68) and (234.22,79.58) .. (229.8,82.15) .. controls (225.39,84.72) and (220.01,84.56) .. (215.74,81.73) .. controls (213.92,88.13) and (208.8,92.83) .. (202.59,93.81) .. controls (196.38,94.79) and (190.2,91.88) .. (186.71,86.33) .. controls (182.43,89.06) and (177.3,89.85) .. (172.47,88.51) .. controls (167.64,87.18) and (163.53,83.82) .. (161.04,79.22) .. controls (156.67,79.76) and (152.45,77.36) .. (150.46,73.2) .. controls (148.48,69.04) and (149.16,64.02) .. (152.17,60.62) .. controls (148.27,58.18) and (146.28,53.35) .. (147.24,48.64) .. controls (148.2,43.94) and (151.88,40.42) .. (156.38,39.92) ;
  \draw  [line width=3, fill=white] (152.17,60.62) .. controls (154.01,61.77) and (156.13,62.29) .. (158.26,62.11)(161.04,79.22) .. controls (161.96,79.1) and (162.85,78.86) .. (163.71,78.5)(186.71,86.33) .. controls (186.07,85.3) and (185.53,84.21) .. (185.1,83.06)(215.74,81.73) .. controls (216.07,80.57) and (216.29,79.37) .. (216.38,78.15)(237.02,69.34) .. controls (237.06,63.66) and (234.02,58.46) .. (229.2,55.97)(247.63,41.71) .. controls (246.85,43.65) and (245.66,45.36) .. (244.15,46.73)(239.21,23.19) .. controls (239.35,23.97) and (239.41,24.76) .. (239.4,25.56)(218.81,17.39) .. controls (218.08,18.3) and (217.48,19.31) .. (217.03,20.41)(201.08,19.17) .. controls (200.69,20) and (200.4,20.87) .. (200.21,21.77)(180.76,22.49) .. controls (181.89,23.19) and (182.94,24.04) .. (183.88,25.01)(156.46,39.67) .. controls (156.58,40.57) and (156.76,41.46) .. (157.01,42.33) ;

\draw (65,51) node  [font=\LARGE]  {$\boldsymbol{\psi}${\bf(N-1)}};
\draw (65,231) node  [font=\LARGE]  {$\boldsymbol{\psi}${\bf(N)}};
\draw (68,142) node  [font=\LARGE]  {$\boldsymbol{\partial}\mathbf{P_{N}}$};
\draw (201,50) node  [font=\LARGE]  {$\boldsymbol{\partial \varphi}${\bf(N)}};

\end{tikzpicture}

}
\caption{Inductive step}
\label{fig:ht-indstep}
\end{wrapfigure}

%% file: htfig-wp.tex
\begin{figure}[H]
\centering

\resizebox{0.39\textwidth}{!}{

\tikzset{every picture/.style={line width=0.75pt}} 

\begin{tikzpicture}[x=0.75pt,y=0.75pt,yscale=-1,xscale=1]

  \draw  [line width=3, fill=white] (152.46,37.67) .. controls (151.63,31.14) and (154.38,24.67) .. (159.55,21.02) .. controls (164.72,17.36) and (171.4,17.15) .. (176.76,20.49) .. controls (178.66,16.69) and (182.13,14.07) .. (186.14,13.41) .. controls (190.14,12.76) and (194.19,14.15) .. (197.08,17.17) .. controls (198.69,13.73) and (201.87,11.41) .. (205.48,11.05) .. controls (209.09,10.69) and (212.62,12.33) .. (214.81,15.39) .. controls (217.74,11.74) and (222.38,10.2) .. (226.75,11.44) .. controls (231.12,12.68) and (234.41,16.48) .. (235.21,21.19) .. controls (238.79,22.22) and (241.78,24.86) .. (243.39,28.41) .. controls (245.01,31.96) and (245.09,36.09) .. (243.63,39.71) .. controls (247.16,44.58) and (247.98,51.07) .. (245.8,56.76) .. controls (243.61,62.44) and (238.75,66.47) .. (233.02,67.34) .. controls (232.98,72.68) and (230.22,77.58) .. (225.8,80.15) .. controls (221.39,82.72) and (216.01,82.56) .. (211.74,79.73) .. controls (209.92,86.13) and (204.8,90.83) .. (198.59,91.81) .. controls (192.38,92.79) and (186.2,89.88) .. (182.71,84.33) .. controls (178.43,87.06) and (173.3,87.85) .. (168.47,86.51) .. controls (163.64,85.18) and (159.53,81.82) .. (157.04,77.22) .. controls (152.67,77.76) and (148.45,75.36) .. (146.46,71.2) .. controls (144.48,67.04) and (145.16,62.02) .. (148.17,58.62) .. controls (144.27,56.18) and (142.28,51.35) .. (143.24,46.64) .. controls (144.2,41.94) and (147.88,38.42) .. (152.38,37.92) ;
  \draw  [line width=3, fill=white] (148.17,58.62) .. controls (150.01,59.77) and (152.13,60.29) .. (154.26,60.11)(157.04,77.22) .. controls (157.96,77.1) and (158.85,76.86) .. (159.71,76.5)(182.71,84.33) .. controls (182.07,83.3) and (181.53,82.21) .. (181.1,81.06)(211.74,79.73) .. controls (212.07,78.57) and (212.29,77.37) .. (212.38,76.15)(233.02,67.34) .. controls (233.06,61.66) and (230.02,56.46) .. (225.2,53.97)(243.63,39.71) .. controls (242.85,41.65) and (241.66,43.36) .. (240.15,44.73)(235.21,21.19) .. controls (235.35,21.97) and (235.41,22.76) .. (235.4,23.56)(214.81,15.39) .. controls (214.08,16.3) and (213.48,17.31) .. (213.03,18.41)(197.08,17.17) .. controls (196.69,18) and (196.4,18.87) .. (196.21,19.77)(176.76,20.49) .. controls (177.89,21.19) and (178.94,22.04) .. (179.88,23.01)(152.46,37.67) .. controls (152.58,38.57) and (152.76,39.46) .. (153.01,40.33) ;
  \draw  [line width=3, fill=white] (152.46,234.67) .. controls (151.63,228.14) and (154.38,221.67) .. (159.55,218.02) .. controls (164.72,214.36) and (171.4,214.15) .. (176.76,217.49) .. controls (178.66,213.69) and (182.13,211.07) .. (186.14,210.41) .. controls (190.14,209.76) and (194.19,211.15) .. (197.08,214.17) .. controls (198.69,210.73) and (201.87,208.41) .. (205.48,208.05) .. controls (209.09,207.69) and (212.62,209.33) .. (214.81,212.39) .. controls (217.74,208.74) and (222.38,207.2) .. (226.75,208.44) .. controls (231.12,209.68) and (234.41,213.48) .. (235.21,218.19) .. controls (238.79,219.22) and (241.78,221.86) .. (243.39,225.41) .. controls (245.01,228.96) and (245.09,233.09) .. (243.63,236.71) .. controls (247.16,241.58) and (247.98,248.07) .. (245.8,253.76) .. controls (243.61,259.44) and (238.75,263.47) .. (233.02,264.34) .. controls (232.98,269.68) and (230.22,274.58) .. (225.8,277.15) .. controls (221.39,279.72) and (216.01,279.56) .. (211.74,276.73) .. controls (209.92,283.13) and (204.8,287.83) .. (198.59,288.81) .. controls (192.38,289.79) and (186.2,286.88) .. (182.71,281.33) .. controls (178.43,284.06) and (173.3,284.85) .. (168.47,283.51) .. controls (163.64,282.18) and (159.53,278.82) .. (157.04,274.22) .. controls (152.67,274.76) and (148.45,272.36) .. (146.46,268.2) .. controls (144.48,264.04) and (145.16,259.02) .. (148.17,255.62) .. controls (144.27,253.18) and (142.28,248.35) .. (143.24,243.64) .. controls (144.2,238.94) and (147.88,235.42) .. (152.38,234.92) ;
  \draw  [line width=3, fill=white] (148.17,255.62) .. controls (150.01,256.77) and (152.13,257.29) .. (154.26,257.11)(157.04,274.22) .. controls (157.96,274.1) and (158.85,273.86) .. (159.71,273.5)(182.71,281.33) .. controls (182.07,280.3) and (181.53,279.21) .. (181.1,278.06)(211.74,276.73) .. controls (212.07,275.57) and (212.29,274.37) .. (212.38,273.15)(233.02,264.34) .. controls (233.06,258.66) and (230.02,253.46) .. (225.2,250.97)(243.63,236.71) .. controls (242.85,238.65) and (241.66,240.36) .. (240.15,241.73)(235.21,218.19) .. controls (235.35,218.97) and (235.41,219.76) .. (235.4,220.56)(214.81,212.39) .. controls (214.08,213.3) and (213.48,214.31) .. (213.03,215.41)(197.08,214.17) .. controls (196.69,215) and (196.4,215.87) .. (196.21,216.77)(176.76,217.49) .. controls (177.89,218.19) and (178.94,219.04) .. (179.88,220.01)(152.46,234.67) .. controls (152.58,235.57) and (152.76,236.46) .. (153.01,237.33) ;
  \draw  [line width=3, fill=white] (144,115) -- (251,115) -- (251,184) -- (144,184) -- cycle ;
  \draw  [line width=3, fill=white] (152.46,414.67) .. controls (151.63,408.14) and (154.38,401.67) .. (159.55,398.02) .. controls (164.72,394.36) and (171.4,394.15) .. (176.76,397.49) .. controls (178.66,393.69) and (182.13,391.07) .. (186.14,390.41) .. controls (190.14,389.76) and (194.19,391.15) .. (197.08,394.17) .. controls (198.69,390.73) and (201.87,388.41) .. (205.48,388.05) .. controls (209.09,387.69) and (212.62,389.33) .. (214.81,392.39) .. controls (217.74,388.74) and (222.38,387.2) .. (226.75,388.44) .. controls (231.12,389.68) and (234.41,393.48) .. (235.21,398.19) .. controls (238.79,399.22) and (241.78,401.86) .. (243.39,405.41) .. controls (245.01,408.96) and (245.09,413.09) .. (243.63,416.71) .. controls (247.16,421.58) and (247.98,428.07) .. (245.8,433.76) .. controls (243.61,439.44) and (238.75,443.47) .. (233.02,444.34) .. controls (232.98,449.68) and (230.22,454.58) .. (225.8,457.15) .. controls (221.39,459.72) and (216.01,459.56) .. (211.74,456.73) .. controls (209.92,463.13) and (204.8,467.83) .. (198.59,468.81) .. controls (192.38,469.79) and (186.2,466.88) .. (182.71,461.33) .. controls (178.43,464.06) and (173.3,464.85) .. (168.47,463.51) .. controls (163.64,462.18) and (159.53,458.82) .. (157.04,454.22) .. controls (152.67,454.76) and (148.45,452.36) .. (146.46,448.2) .. controls (144.48,444.04) and (145.16,439.02) .. (148.17,435.62) .. controls (144.27,433.18) and (142.28,428.35) .. (143.24,423.64) .. controls (144.2,418.94) and (147.88,415.42) .. (152.38,414.92) ;
  \draw  [line width=3, fill=white] (148.17,435.62) .. controls (150.01,436.77) and (152.13,437.29) .. (154.26,437.11)(157.04,454.22) .. controls (157.96,454.1) and (158.85,453.86) .. (159.71,453.5)(182.71,461.33) .. controls (182.07,460.3) and (181.53,459.21) .. (181.1,458.06)(211.74,456.73) .. controls (212.07,455.57) and (212.29,454.37) .. (212.38,453.15)(233.02,444.34) .. controls (233.06,438.66) and (230.02,433.46) .. (225.2,430.97)(243.63,416.71) .. controls (242.85,418.65) and (241.66,420.36) .. (240.15,421.73)(235.21,398.19) .. controls (235.35,398.97) and (235.41,399.76) .. (235.4,400.56)(214.81,392.39) .. controls (214.08,393.3) and (213.48,394.31) .. (213.03,395.41)(197.08,394.17) .. controls (196.69,395) and (196.4,395.87) .. (196.21,396.77)(176.76,397.49) .. controls (177.89,398.19) and (178.94,399.04) .. (179.88,400.01)(152.46,414.67) .. controls (152.58,415.57) and (152.76,416.46) .. (153.01,417.33) ;
  \draw  [line width=3, fill=white] (144,310) -- (251,310) -- (251,364) -- (144,364) -- cycle ;
  \draw  [line width=3, fill=white] (287.46,37.67) .. controls (286.63,31.14) and (289.38,24.67) .. (294.55,21.02) .. controls (299.72,17.36) and (306.4,17.15) .. (311.76,20.49) .. controls (313.66,16.69) and (317.13,14.07) .. (321.14,13.41) .. controls (325.14,12.76) and (329.19,14.15) .. (332.08,17.17) .. controls (333.69,13.73) and (336.87,11.41) .. (340.48,11.05) .. controls (344.09,10.69) and (347.62,12.33) .. (349.81,15.39) .. controls (352.74,11.74) and (357.38,10.2) .. (361.75,11.44) .. controls (366.12,12.68) and (369.41,16.48) .. (370.21,21.19) .. controls (373.79,22.22) and (376.78,24.86) .. (378.39,28.41) .. controls (380.01,31.96) and (380.09,36.09) .. (378.63,39.71) .. controls (382.16,44.58) and (382.98,51.07) .. (380.8,56.76) .. controls (378.61,62.44) and (373.75,66.47) .. (368.02,67.34) .. controls (367.98,72.68) and (365.22,77.58) .. (360.8,80.15) .. controls (356.39,82.72) and (351.01,82.56) .. (346.74,79.73) .. controls (344.92,86.13) and (339.8,90.83) .. (333.59,91.81) .. controls (327.38,92.79) and (321.2,89.88) .. (317.71,84.33) .. controls (313.43,87.06) and (308.3,87.85) .. (303.47,86.51) .. controls (298.64,85.18) and (294.53,81.82) .. (292.04,77.22) .. controls (287.67,77.76) and (283.45,75.36) .. (281.46,71.2) .. controls (279.48,67.04) and (280.16,62.02) .. (283.17,58.62) .. controls (279.27,56.18) and (277.28,51.35) .. (278.24,46.64) .. controls (279.2,41.94) and (282.88,38.42) .. (287.38,37.92) ;
  \draw  [line width=3, fill=white] (283.17,58.62) .. controls (285.01,59.77) and (287.13,60.29) .. (289.26,60.11)(292.04,77.22) .. controls (292.96,77.1) and (293.85,76.86) .. (294.71,76.5)(317.71,84.33) .. controls (317.07,83.3) and (316.53,82.21) .. (316.1,81.06)(346.74,79.73) .. controls (347.07,78.57) and (347.29,77.37) .. (347.38,76.15)(368.02,67.34) .. controls (368.06,61.66) and (365.02,56.46) .. (360.2,53.97)(378.63,39.71) .. controls (377.85,41.65) and (376.66,43.36) .. (375.15,44.73)(370.21,21.19) .. controls (370.35,21.97) and (370.41,22.76) .. (370.4,23.56)(349.81,15.39) .. controls (349.08,16.3) and (348.48,17.31) .. (348.03,18.41)(332.08,17.17) .. controls (331.69,18) and (331.4,18.87) .. (331.21,19.77)(311.76,20.49) .. controls (312.89,21.19) and (313.94,22.04) .. (314.88,23.01)(287.46,37.67) .. controls (287.58,38.57) and (287.76,39.46) .. (288.01,40.33) ;
  \draw  [line width=3, fill=white] (287.46,234.67) .. controls (286.63,228.14) and (289.38,221.67) .. (294.55,218.02) .. controls (299.72,214.36) and (306.4,214.15) .. (311.76,217.49) .. controls (313.66,213.69) and (317.13,211.07) .. (321.14,210.41) .. controls (325.14,209.76) and (329.19,211.15) .. (332.08,214.17) .. controls (333.69,210.73) and (336.87,208.41) .. (340.48,208.05) .. controls (344.09,207.69) and (347.62,209.33) .. (349.81,212.39) .. controls (352.74,208.74) and (357.38,207.2) .. (361.75,208.44) .. controls (366.12,209.68) and (369.41,213.48) .. (370.21,218.19) .. controls (373.79,219.22) and (376.78,221.86) .. (378.39,225.41) .. controls (380.01,228.96) and (380.09,233.09) .. (378.63,236.71) .. controls (382.16,241.58) and (382.98,248.07) .. (380.8,253.76) .. controls (378.61,259.44) and (373.75,263.47) .. (368.02,264.34) .. controls (367.98,269.68) and (365.22,274.58) .. (360.8,277.15) .. controls (356.39,279.72) and (351.01,279.56) .. (346.74,276.73) .. controls (344.92,283.13) and (339.8,287.83) .. (333.59,288.81) .. controls (327.38,289.79) and (321.2,286.88) .. (317.71,281.33) .. controls (313.43,284.06) and (308.3,284.85) .. (303.47,283.51) .. controls (298.64,282.18) and (294.53,278.82) .. (292.04,274.22) .. controls (287.67,274.76) and (283.45,272.36) .. (281.46,268.2) .. controls (279.48,264.04) and (280.16,259.02) .. (283.17,255.62) .. controls (279.27,253.18) and (277.28,248.35) .. (278.24,243.64) .. controls (279.2,238.94) and (282.88,235.42) .. (287.38,234.92) ;
  \draw  [line width=3, fill=white] (283.17,255.62) .. controls (285.01,256.77) and (287.13,257.29) .. (289.26,257.11)(292.04,274.22) .. controls (292.96,274.1) and (293.85,273.86) .. (294.71,273.5)(317.71,281.33) .. controls (317.07,280.3) and (316.53,279.21) .. (316.1,278.06)(346.74,276.73) .. controls (347.07,275.57) and (347.29,274.37) .. (347.38,273.15)(368.02,264.34) .. controls (368.06,258.66) and (365.02,253.46) .. (360.2,250.97)(378.63,236.71) .. controls (377.85,238.65) and (376.66,240.36) .. (375.15,241.73)(370.21,218.19) .. controls (370.35,218.97) and (370.41,219.76) .. (370.4,220.56)(349.81,212.39) .. controls (349.08,213.3) and (348.48,214.31) .. (348.03,215.41)(332.08,214.17) .. controls (331.69,215) and (331.4,215.87) .. (331.21,216.77)(311.76,217.49) .. controls (312.89,218.19) and (313.94,219.04) .. (314.88,220.01)(287.46,234.67) .. controls (287.58,235.57) and (287.76,236.46) .. (288.01,237.33) ;
  \draw  [line width=3, fill=white] (20.46,233.67) .. controls (19.63,227.14) and (22.38,220.67) .. (27.55,217.02) .. controls (32.72,213.36) and (39.4,213.15) .. (44.76,216.49) .. controls (46.66,212.69) and (50.13,210.07) .. (54.14,209.41) .. controls (58.14,208.76) and (62.19,210.15) .. (65.08,213.17) .. controls (66.69,209.73) and (69.87,207.41) .. (73.48,207.05) .. controls (77.09,206.69) and (80.62,208.33) .. (82.81,211.39) .. controls (85.74,207.74) and (90.38,206.2) .. (94.75,207.44) .. controls (99.12,208.68) and (102.41,212.48) .. (103.21,217.19) .. controls (106.79,218.22) and (109.78,220.86) .. (111.39,224.41) .. controls (113.01,227.96) and (113.09,232.09) .. (111.63,235.71) .. controls (115.16,240.58) and (115.98,247.07) .. (113.8,252.76) .. controls (111.61,258.44) and (106.75,262.47) .. (101.02,263.34) .. controls (100.98,268.68) and (98.22,273.58) .. (93.8,276.15) .. controls (89.39,278.72) and (84.01,278.56) .. (79.74,275.73) .. controls (77.92,282.13) and (72.8,286.83) .. (66.59,287.81) .. controls (60.38,288.79) and (54.2,285.88) .. (50.71,280.33) .. controls (46.43,283.06) and (41.3,283.85) .. (36.47,282.51) .. controls (31.64,281.18) and (27.53,277.82) .. (25.04,273.22) .. controls (20.67,273.76) and (16.45,271.36) .. (14.46,267.2) .. controls (12.48,263.04) and (13.16,258.02) .. (16.17,254.62) .. controls (12.27,252.18) and (10.28,247.35) .. (11.24,242.64) .. controls (12.2,237.94) and (15.88,234.42) .. (20.38,233.92) ;
  \draw  [line width=3, fill=white] (16.17,254.62) .. controls (18.01,255.77) and (20.13,256.29) .. (22.26,256.11)(25.04,273.22) .. controls (25.96,273.1) and (26.85,272.86) .. (27.71,272.5)(50.71,280.33) .. controls (50.07,279.3) and (49.53,278.21) .. (49.1,277.06)(79.74,275.73) .. controls (80.07,274.57) and (80.29,273.37) .. (80.38,272.15)(101.02,263.34) .. controls (101.06,257.66) and (98.02,252.46) .. (93.2,249.97)(111.63,235.71) .. controls (110.85,237.65) and (109.66,239.36) .. (108.15,240.73)(103.21,217.19) .. controls (103.35,217.97) and (103.41,218.76) .. (103.4,219.56)(82.81,211.39) .. controls (82.08,212.3) and (81.48,213.31) .. (81.03,214.41)(65.08,213.17) .. controls (64.69,214) and (64.4,214.87) .. (64.21,215.77)(44.76,216.49) .. controls (45.89,217.19) and (46.94,218.04) .. (47.88,219.01)(20.46,233.67) .. controls (20.58,234.57) and (20.76,235.46) .. (21.01,236.33) ;
  \draw  [line width=3, fill=white] (19.46,414.67) .. controls (18.63,408.14) and (21.38,401.67) .. (26.55,398.02) .. controls (31.72,394.36) and (38.4,394.15) .. (43.76,397.49) .. controls (45.66,393.69) and (49.13,391.07) .. (53.14,390.41) .. controls (57.14,389.76) and (61.19,391.15) .. (64.08,394.17) .. controls (65.69,390.73) and (68.87,388.41) .. (72.48,388.05) .. controls (76.09,387.69) and (79.62,389.33) .. (81.81,392.39) .. controls (84.74,388.74) and (89.38,387.2) .. (93.75,388.44) .. controls (98.12,389.68) and (101.41,393.48) .. (102.21,398.19) .. controls (105.79,399.22) and (108.78,401.86) .. (110.39,405.41) .. controls (112.01,408.96) and (112.09,413.09) .. (110.63,416.71) .. controls (114.16,421.58) and (114.98,428.07) .. (112.8,433.76) .. controls (110.61,439.44) and (105.75,443.47) .. (100.02,444.34) .. controls (99.98,449.68) and (97.22,454.58) .. (92.8,457.15) .. controls (88.39,459.72) and (83.01,459.56) .. (78.74,456.73) .. controls (76.92,463.13) and (71.8,467.83) .. (65.59,468.81) .. controls (59.38,469.79) and (53.2,466.88) .. (49.71,461.33) .. controls (45.43,464.06) and (40.3,464.85) .. (35.47,463.51) .. controls (30.64,462.18) and (26.53,458.82) .. (24.04,454.22) .. controls (19.67,454.76) and (15.45,452.36) .. (13.46,448.2) .. controls (11.48,444.04) and (12.16,439.02) .. (15.17,435.62) .. controls (11.27,433.18) and (9.28,428.35) .. (10.24,423.64) .. controls (11.2,418.94) and (14.88,415.42) .. (19.38,414.92) ;
  \draw  [line width=3, fill=white] (15.17,435.62) .. controls (17.01,436.77) and (19.13,437.29) .. (21.26,437.11)(24.04,454.22) .. controls (24.96,454.1) and (25.85,453.86) .. (26.71,453.5)(49.71,461.33) .. controls (49.07,460.3) and (48.53,459.21) .. (48.1,458.06)(78.74,456.73) .. controls (79.07,455.57) and (79.29,454.37) .. (79.38,453.15)(100.02,444.34) .. controls (100.06,438.66) and (97.02,433.46) .. (92.2,430.97)(110.63,416.71) .. controls (109.85,418.65) and (108.66,420.36) .. (107.15,421.73)(102.21,398.19) .. controls (102.35,398.97) and (102.41,399.76) .. (102.4,400.56)(81.81,392.39) .. controls (81.08,393.3) and (80.48,394.31) .. (80.03,395.41)(64.08,394.17) .. controls (63.69,395) and (63.4,395.87) .. (63.21,396.77)(43.76,397.49) .. controls (44.89,398.19) and (45.94,399.04) .. (46.88,400.01)(19.46,414.67) .. controls (19.58,415.57) and (19.76,416.46) .. (20.01,417.33) ;
\draw [line width=3, draw=black]    (334,108) -- (334,120) ;
\draw [line width=3, draw=black]    (334,162) -- (334,174) ;
\draw [line width=3, draw=black]    (334,144) -- (334,156) ;
\draw [line width=3, draw=black]    (334,126) -- (334,138) ;
\draw [line width=3, draw=black]    (334,180) -- (334,192) ;

\draw (197,48) node  [font=\LARGE]  {$\boldsymbol{\varphi}${\bf(N-1)}};
\draw (197,245) node  [font=\LARGE]  {$\boldsymbol{\psi}${\bf(N-1)}};
\draw (199,151) node  [font=\LARGE]  {$\mathbf{P_{N\text{-}1}}$};
\draw (196,425) node  [font=\LARGE]  {$\boldsymbol{\psi}${\bf(N)}};
\draw (199,336) node  [font=\LARGE]  {$\boldsymbol{\partial}\mathbf{P_{N}}$};
\draw (332,48) node  [font=\LARGE]  {$\boldsymbol{\partial \varphi}${\bf(N)}};
\draw (332,245) node  [font=\LARGE]  {$\boldsymbol{\partial \varphi}${\bf(N)}};
\draw (64,244) node  [font=\LARGE]  {$\boldsymbol{\ppre}${\bf(N-1)}};
\draw (64,425) node  [font=\LARGE]  {$\boldsymbol{\ppre}${\bf(N)}};

\end{tikzpicture}

}
\caption{Strengthing pre- and post-conditions}
\label{fig:ht-wp}
\end{figure}

%% file: prelim.tex

\begin{figure*}[tb]
\begin{center}
\begin{tabular}{rcl}
 \PB    & ::= & \Stmt \\
 \Stmt  & ::= & {\AssignStmts} ~$\mid$~ \Stmt~;~\Stmt ~$\mid$~ {\iif}(\BoolE) {\tthen} {\Stmt} {\eelse} \Stmt ~$\mid$~
                {\ffor} ({\lpVar} := 0; {\lpVar} $<$ \LBE; {\lpVar} := {\lpVar}+1) \{{\Stmta}\}\\
 \AssignStmts & ::= & {\scVar} = {\EE} ~$\mid$~ {\ArVar}[\IE] = \EE \\
 \Stmta & ::= & {\AssignStmts} ~$\mid$~ \Stmta~;~\Stmta ~$\mid$~ {\iif}(\BoolE) {\tthen} {\Stmta} {\eelse} \Stmta \\
 \EE    & ::= & {\EE} ~\OP~ {\EE} ~$\mid$~ {\ArVar}[\IE] ~$\mid$~ {\scVar}
                ~$\mid$~ {\lpVar} ~$\mid$~ {\cconst} ~$\mid$~ $N$ \\
 \IE    & ::= & {\IE} ~\OP~ {\IE} ~$\mid$~ {\scVar} ~$\mid$~ {\lpVar} ~$\mid$~ {\cconst} ~$\mid$~ $N$ \\
 \LBE    & ::= & {\LBE} ~\OP~ {\LBE} ~$\mid$~ {\cconst} ~$\mid$~ $N$ \\
 \OP    & ::= & + ~$\mid$~ - ~$\mid$~ $\times$ ~$\mid$~ $\div$ ~$\mid$~ $\ldots$ \\
 \BoolE & ::= & {\EE} ~$\mathsf{relop}$ {\EE} ~$\mid$~ {\BoolE} $\mathsf{AND}$ {\BoolE} ~$\mid$~ $\mathsf{NOT}$ {\BoolE} ~$\mid$~ {\BoolE} $\mathsf{OR}$ {\BoolE} \\
 \RelOP & ::= & $==$ ~$\mid$~ $<$ ~$\mid$~ $\leq$ ~$\mid$~ $>$ ~$\mid$~ $\geq$ \\
\end{tabular}
\end{center}
\label{fig:grammar}
\caption{Program grammar}
\end{figure*}

We consider array manipulating programs generated by the grammar shown
in Fig. \ref{fig:grammar} (adapted from \cite{sas17}).  This grammar
restricts programs to have non-nested loops.  Specifically, programs
generated starting from {\Stmta} are loop free.  The non-terminal
{\Stmt} can generate programs with loops but their bodies are
generated from {\Stmta}, thereby forbidding nesting of loops.  Note
also that expressions for indexing arrays are generated from the
non-terminal {\IE}, and such expressions cannot refer to other array
elements.  However, this is not really a restriction on the expressive
power of programs since every array index expression that depends on
other array elements, say $A[e]$, can be replaced by an array index
expression that depends on temporary variables, say $v$, that are
pre-assigned to the respective array elements, viz. $A[e]$.  For
example, {\tt A[B[i]] = C[D[i]];} can be rewritten as {\tt v1 = B[i];
  v2 = D[i]; A[v1] = C[v2];}.  Finally, note that loop bound
expressions are generated using the non-terminal $\LBE$, and such
expressions can only involve constants and the parameter $N$.  While
the above restrictions limit the class of programs to which our
technique currently applies, there is still a large collection of
useful programs, with possibly long sequences of loops, that are
included in the scope of our work.  In reality, our technique also
applies to a sub-class of programs with nested loops and with loop
bound expressions that involve scalar variables. However,
characterizing this sub-class of programs through a grammar is a bit
unwieldy, and we avoid doing so for reasons of clarity.

A program $\PP_N$ is a tuple $(\mathcal{V}, \mathcal{L}, \mathcal{A},
{\PB}, N)$, where $\mathcal{V}$ is a set of scalar variables,
$\mathcal{L} \subseteq \mathcal{V}$ is a set of loop counter
variables, $\mathcal{A}$ is a set of array variables, ${\PB}$ is the
program body, and $N$ is a special symbol denoting a positive integer
parameter. In the grammar shown above, we assume ${\ArVar} \in
\mathcal{A}$, ${\scVar} \in \mathcal{V} \setminus \mathcal{L}$,
${\lpVar} \in \mathcal{L}$ and ${\cconst}\in \integers$.  Furthermore,
``$\RelOP$'' is assumed to be one of the relational operators and
``$\OP$'' is an arithmetic operator.  We discuss more about the
operators supported by our technique in Sect. \ref{sec:diff-prog}.  We
also assume that each loop $L$ has a unique loop counter variable
$\ell$ that is initialized at the beginning of $L$ and is incremented
by $1$ at the end of each iteration. Assignments in the body of $L$
are assumed not to update $\ell$.  Finally, for each loop with
termination condition $\ell < \LBE$, we assume that $\LBE$ is an
expression in terms of $N$.  We denote by $k_L(N)$ the number of times
loop $L$ iterates in the program with parameter $N$.

We admit Hoare triples of the form $\{\varphi(N)\}$ $\;\PP_N\;$
$\{\psi(N)\}$, where $\varphi(N)$ and $\psi(N)$ are either universally
quantified, existentially quantified or quantifier-free formulas of
the form $\forall I\,$
$\left(\Phi(I,N) \implies \Psi(\mathcal{A}, \mathcal{V}, I,
N)\right)$, $\exists I\,$ $\left(\Phi(I,
N) \wedge \Psi(\mathcal{A}, \mathcal{V}, I, N)\right)$ and
$\Upsilon(\mathcal{A}, \mathcal{V}, N)$ respectively.  In the above,
$I$ is an array index variable, $\Phi$ is a quantifier-free formula in
the theory of arithmetic over integers, and $\Psi$ and $\Upsilon$ are
quantifier-free formulas in the combined theory of arrays and
arithmetic over integers.  Our technique also works for conjunctions
and disjunctions of such formulas as pre- and post-conditions in some
cases as discussed in Sect. \ref{sec:diff-pre}.

\subsection{Tracking Control Flow}
\label{sec:cfg}

We represent a program $\PP_N$ using its {\em control flow graph} (or
CFG) $G_C = (\mathit{Locs}, CE, \mu)$, where $Locs$ denotes the set of
control locations (nodes) of the program, $CE \subseteq \mathit{Locs}
\times \mathit{Locs} \times \{\ltrue, \lfalse, \Unlabeled\}$ are the
control-flow edges, and $\mu: \mathit{Locs} \rightarrow
{\AssignStmts}$ $\union$ ${\BoolE}$ annotates every node in $Locs$
with either an assignment statement (of the form ${\scVar} = \EE$ or
${\ArVar}[\IE] = \EE$) from those represented by {\AssignStmts}, or a
Boolean expression from those represented by {\BoolE}.  Two
distinguished control locations, called $n_{start}$ and $n_{end}$ in
$\mathit{Locs}$ represent the entry and exit points of the program.
An edge $(n_1, n_2, label)$ represents flow of control from $n_1$ to
$n_2$ without any other intervening node. It is labeled $\ltrue$ or
$\lfalse$ if $\mu(n_1)$ is a Boolean condition, and is labeled
$\Unlabeled$ otherwise.  If $\mu(n_1)$ is a Boolean condition, there
are two outgoing edges from $n_1$, labeled $\ltrue$ and $\lfalse$
respectively, and control flows from $n_1$ to $n_2$ along $(n_1, n_2,
label)$ only if $\mu(n_1)$ evaluates to $label$. If $\mu(n_1)$ is an
assignment statement, there is a single outgoing edge from $n_1$, and
it is labeled $\Unlabeled$.  Henceforth, we use CFG to refer to a
control flow graph, and use $\PP_N$ to refer to both a program and its
CFG, when there is no confusion.

\input{cfg}  

A CFG may have cycles due to the presence of loops in the program.
A \emph{back-edge} of a loop is an edge from the node corresponding to
the last statement in the loop body to the node representing the loop
head.  An \emph{exit-edge} is an edge from the loop head to a node
outside the loop body. An \emph{incoming-edge} is an edge to the loop
head from a node outside the loop body. We assume that every loop has
exactly one \emph{back-edge}, one \emph{incoming-edge} and
one \emph{exit-edge}.

A node $n$ in a control flow graph \emph{strictly post-dominates} a
node $m$ if all control flow paths from node $m$ pass through $n$
before reaching the exit node and $m$ is not the same as $n$.  The
\emph{immediate post-dominator} of node $n$ is a node that strictly
post-dominates $n$ but does not strictly post-dominate any other node
that strictly post-dominates $n$.

For every node $n$ in the CFG, we use $\mathit{def}(n)$ and
$\mathit{uses}(n)$ to refer to the set of scalar variables and arrays
(not loop counter variables) that are defined and used, respectively,
in the statement or boolean expression at $n$.  We include the
symbolic parameter $N$ in the set $\mathit{uses}(n)$ if the statement
at node $n$ makes use of $N$.  Since the parameter $N$ cannot be
re-defined by any program generated according to the grammar in
Fig.~\ref{fig:grammar}, it never appears in $\mathit{def}(n)$ for any
node $n$.  If $A$ represents an array in $\mathit{def}(n)$, we use
$\defindex(A,n)$ to refer to the index expression of the element of
$A$ updated at $n$.  Similarly, if $A \in \mathit{uses}(n)$, we use
$\useindex(A,n)$ to refer to the set of index expression(s) of
element(s) of $A$ read at $n$.

\begin{example}
The CFG of the program in Fig. \ref{fig:ss} is shown in
Fig. \ref{fig:cfg}.  The nodes are numbered such that they coincide
with the line numbers in the program.  The graph has three cycles each
corresponding to a loop in the given program. \{ $(1,2)$, $(2,5)$,
$(5,8)$ \} are \emph{incoming-edges}, \{ $(4,2)$, $(7,5)$, $(10,8)$ \}
are \emph{back-edges} and \{ $(2,5)$, $(5,8)$, $(8,End)$ \} are
\emph{exit-edges}.

Node $2$ strictly post-dominates nodes $1$, $4$ and $Start$.  Node $2$
is an immediate post-dominator of $1$ and $4$.  Node $8$ strictly
post-dominates all nodes except itself and $End$.  Node $8$ is an
immediate post-dominator of $5$ and $10$.  $End$ node strictly
post-dominates all nodes except itself.  $Start$ node does not
strictly post-dominate any other node.  Similarly, the post-domination
relations for other nodes can be computed.

The set of scalars and arrays defined at nodes $1$, $3$ and $9$ is
$\mathit{def}(n)$ := \{{\tt S}\} and the set at node $6$ is
$\mathit{def}(n)$ = \{{\tt A}\}.  The index of {\tt A} updated at node
$6$ is $\defindex$({\tt A},$6$) = {\tt i}.  The set of scalars and
arrays used at nodes $3$, $6$ and $9$ is $\mathit{uses}(n)$ = \{{\tt
  S, A}\}.  For node $1$, as there are no uses of scalars or arrays,
$\mathit{uses}(n)$ = $\emptyset$.  The set of indices of array {\tt A}
used at nodes $3$, $6$ and $9$ is $\useindex$({\tt A},$n$) = \{{\tt
  i}\}.
\end{example}

%% file: cfg.tex
\begin{wrapfigure}[22]{r}{0.23\textwidth}
\begin{center}
{\scriptsize
  \begin{tikzpicture}[%
    ->,
    >=stealth,
    node distance=0.8cm,
    noname/.style={%
      ellipse,
      very thick,
      fill=white,
      minimum width=1em,
      minimum height=1em,
      draw
    }
  ]
    \node[noname] (S)              {Start};
    \node[noname] (1) [below=of S] {1};
    \node[noname] (2) [below=of 1] {2};
    \node[noname] (3) [right=of 2] {3};
    \node[noname] (4) [right=of 3] {4};
    \node[noname] (5) [below=of 2] {5};
    \node[noname] (6) [right=of 5] {6};
    \node[noname] (7) [right=of 6] {7};
    \node[noname] (8) [below=of 5] {8};
    \node[noname] (9) [right=of 8] {9};
    \node[noname] (10) [right=of 9] {10};
    \node[noname] (E) [below=of 8] {End};

    \begin{scope}[very thick,-latex]
    \path
    (S) edge                   node {} (1)
    (1) edge                   node {} (2)

    (2) edge                   node [above, xshift=-2mm] {$~~~~\ltrue$} (3)
    (3) edge                   node {} (4)
    (4) edge [bend right=55pt] node {} (2)
    (2) edge node {$\lfalse~~~~$} (5)

    (5) edge                   node [above, xshift=-2mm] {$~~~~\ltrue$} (6)
    (6) edge                   node {} (7)
    (7) edge [bend right=55pt] node {} (5)
    (5) edge node {$\lfalse~~~~$} (8)

    (8) edge                   node [above, xshift=-2mm] {$~~~~\ltrue$} (9)
    (9) edge                   node {} (10)
    (10) edge [bend right=55pt] node {} (8)
    (8) edge node {$\lfalse~~~~$} (E);
    \end{scope}
  \end{tikzpicture}
}
\end{center}
\caption{CFG of Fig. \ref{fig:ss}}
\label{fig:cfg}
\end{wrapfigure}
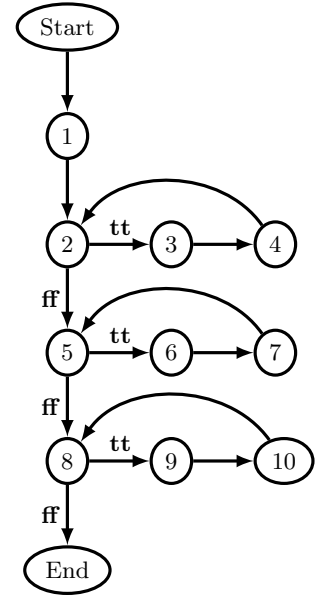

%% file: diffcomp.tex
\input{fig-flow-overview}

In this section, we focus on the generation of two crucial components
for performing \emph{full-program induction}, namely (i)
the \emph{difference program} ${\partial \PP_N}$ and (ii)
the \emph{difference pre-condition} ${\partial \varphi(N)}$.

Computing the \emph{difference program} ${\partial \PP_N}$ is a
non-trivial endeavor.  Fig. \ref{fig:flow-overview} presents a high
level overview of the sequence of steps involved in the generation of
a difference program.  We first carefully rename the variables and
arrays such that each loop in the renamed program refers to its own
copy of variables/arrays.  Note that this is similar in spirit to SSA
renaming, although there are important differences that will become
clear in Sect. \ref{sec:renaming}.  We next peel the last (in some
cases the last few) iteration(s) of each loop in the program such that
the remaining part of each loop in the peeled version of $\PP_N$
iterates exactly the same number of times as the corresponding loop in
$\PP_{N-1}$.  Throughout this paper, we use the term \emph{peel} to
denote the last (or last few as the case may be) iteration(s) of a
loop that have been removed from the loop.  The motivation for this
peeling is that the difference program can often be constructed by
moving the peels of individual loops to the end of the program and
stitching them up in appropriate ways, as will be discussed in detail
in Sects. \ref{sec:diff-prog:peel} and \ref{sec:diff-prog:affected}.
In order to ensure that the semantics of the program is preserved even
after moving the peels to the end of the program, we need to do a
careful analysis of the data dependencies between between variables
and array elements updated/read in statements within loops and those
updated/read in the peeled iterations.  This is achieved by computing
a customized data dependence graph, details of which are presented in
Sect. \ref{sec:refine-pdg}.  In general, variables and array elements
in the program $\PP_N$ can have data/control dependencies on the
parameter $N$ beyond those attributable to the iteration counts of
loops being possibly determined by $N$.  We call such variables/array
elements as ``affected'' by $N$ and identify them using a special
data-flow analysis and the data dependencies computed above.  Details
of this analysis are presented in Sect. \ref{sec:affected}.  Finally,
we move the peels of loops to the end of the program and use the
information about data dependencies and affected variables computed
above to appropriately stitch and modify them to obtain an unoptimized
version of the difference program.  In general, this modification may
involve adding carefully constructed loops in the difference program
itself.  It turns out that the difference program obtained in this way
can often be significantly optimized using simple optimization
techniques.  This includes things like pruning superfluous
computational steps and accelerating loops among others.  We include
this optimization as the last step in our flow for generating the
difference program.

Our empirical studies show that for the success of full-program
induction, it is very important that the difference program
$\partial \PP_N$ be ``simpler'' (defined more precisely later) than
the original program $\PP_N$.  The optimizations enabled by the
affected variable analysis and the simplification of the difference
program are crucial to actualize this requirement.  Each of the above
steps, depicted in Fig. \ref{fig:flow-overview}, is elaborated in
Sects. \ref{sec:renaming} - \ref{sec:simp-diff}.

\input{rename}  

\input{peeling}  

\input{refine-pdg}  

\input{compute-affected}  

\input{difference}  

\input{simp-diff}  

\input{diff-pre}  

%% file: fig-flow-overview.tex
\begin{figure*}[!t]
\centering

\resizebox{0.80\textwidth}{!}{

\tikzset{every picture/.style={line width=0.75pt}} 

\begin{tikzpicture}[x=0.75pt,y=0.75pt,yscale=-1,xscale=1]

\draw  [fill={rgb, 255:red, 255; green, 255; blue, 255 }  ,fill opacity=1 ][line width=2.25]  (142,30.8) .. controls (142,20.42) and (150.42,12) .. (160.8,12) -- (260.53,12) .. controls (270.92,12) and (279.33,20.42) .. (279.33,30.8) -- (279.33,87.2) .. controls (279.33,97.58) and (270.92,106) .. (260.53,106) -- (160.8,106) .. controls (150.42,106) and (142,97.58) .. (142,87.2) -- cycle ;
\draw [line width=2.25]    (34.67,59) -- (135.67,59) ;
\draw [shift={(140.67,59)}, rotate = 180] [fill={rgb, 255:red, 0; green, 0; blue, 0 }  ][line width=0.08]  [draw opacity=0] (16.07,-7.72) -- (0,0) -- (16.07,7.72) -- (10.67,0) -- cycle    ;
\draw  [fill={rgb, 255:red, 255; green, 255; blue, 255 }  ,fill opacity=1 ][line width=2.25]  (388,30.8) .. controls (388,20.42) and (396.42,12) .. (406.8,12) -- (506.53,12) .. controls (516.92,12) and (525.33,20.42) .. (525.33,30.8) -- (525.33,87.2) .. controls (525.33,97.58) and (516.92,106) .. (506.53,106) -- (406.8,106) .. controls (396.42,106) and (388,97.58) .. (388,87.2) -- cycle ;
\draw [line width=2.25]    (280.67,59) -- (381.67,59) ;
\draw [shift={(386.67,59)}, rotate = 180] [fill={rgb, 255:red, 0; green, 0; blue, 0 }  ][line width=0.08]  [draw opacity=0] (16.07,-7.72) -- (0,0) -- (16.07,7.72) -- (10.67,0) -- cycle    ;
\draw  [fill={rgb, 255:red, 255; green, 255; blue, 255 }  ,fill opacity=1 ][line width=2.25]  (388,229.8) .. controls (388,219.42) and (396.42,211) .. (406.8,211) -- (506.53,211) .. controls (516.92,211) and (525.33,219.42) .. (525.33,229.8) -- (525.33,286.2) .. controls (525.33,296.58) and (516.92,305) .. (506.53,305) -- (406.8,305) .. controls (396.42,305) and (388,296.58) .. (388,286.2) -- cycle ;
\draw [line width=2.25]    (285.67,258) -- (386.67,258) ;
\draw [shift={(280.67,258)}, rotate = 0] [fill={rgb, 255:red, 0; green, 0; blue, 0 }  ][line width=0.08]  [draw opacity=0] (16.07,-7.72) -- (0,0) -- (16.07,7.72) -- (10.67,0) -- cycle    ;
\draw  [fill={rgb, 255:red, 255; green, 255; blue, 255 }  ,fill opacity=1 ][line width=2.25]  (634,229.8) .. controls (634,219.42) and (642.42,211) .. (652.8,211) -- (752.53,211) .. controls (762.92,211) and (771.33,219.42) .. (771.33,229.8) -- (771.33,286.2) .. controls (771.33,296.58) and (762.92,305) .. (752.53,305) -- (652.8,305) .. controls (642.42,305) and (634,296.58) .. (634,286.2) -- cycle ;
\draw [line width=2.25]    (531.67,258) -- (632.67,258) ;
\draw [shift={(526.67,258)}, rotate = 0] [fill={rgb, 255:red, 0; green, 0; blue, 0 }  ][line width=0.08]  [draw opacity=0] (16.07,-7.72) -- (0,0) -- (16.07,7.72) -- (10.67,0) -- cycle    ;
\draw [line width=2.25]    (703.67,106) -- (703.67,203.42) ;
\draw [shift={(703.67,208.42)}, rotate = 270] [fill={rgb, 255:red, 0; green, 0; blue, 0 }  ][line width=0.08]  [draw opacity=0] (16.07,-7.72) -- (0,0) -- (16.07,7.72) -- (10.67,0) -- cycle    ;
\draw [line width=2.25]    (525.67,59) -- (626.67,59) ;
\draw [shift={(631.67,59)}, rotate = 180] [fill={rgb, 255:red, 0; green, 0; blue, 0 }  ][line width=0.08]  [draw opacity=0] (16.07,-7.72) -- (0,0) -- (16.07,7.72) -- (10.67,0) -- cycle    ;
\draw  [fill={rgb, 255:red, 255; green, 255; blue, 255 }  ,fill opacity=1 ][line width=2.25]  (633,28.8) .. controls (633,18.42) and (641.42,10) .. (651.8,10) -- (751.53,10) .. controls (761.92,10) and (770.33,18.42) .. (770.33,28.8) -- (770.33,85.2) .. controls (770.33,95.58) and (761.92,104) .. (751.53,104) -- (651.8,104) .. controls (641.42,104) and (633,95.58) .. (633,85.2) -- cycle ;
\draw  [fill={rgb, 255:red, 255; green, 255; blue, 255 }  ,fill opacity=1 ][line width=2.25]  (142,229.8) .. controls (142,219.42) and (150.42,211) .. (160.8,211) -- (260.53,211) .. controls (270.92,211) and (279.33,219.42) .. (279.33,229.8) -- (279.33,286.2) .. controls (279.33,296.58) and (270.92,305) .. (260.53,305) -- (160.8,305) .. controls (150.42,305) and (142,296.58) .. (142,286.2) -- cycle ;
\draw [line width=2.25]    (39.67,258) -- (140.67,258) ;
\draw [shift={(34.67,258)}, rotate = 0] [fill={rgb, 255:red, 0; green, 0; blue, 0 }  ][line width=0.08]  [draw opacity=0] (16.07,-7.72) -- (0,0) -- (16.07,7.72) -- (10.67,0) -- cycle    ;

\draw (160,50) node [anchor=north west][inner sep=0.75pt]   [align=left] {\textbf{\textit{{\llarge Renaming}}}};
\draw (5,48) node [anchor=north west][inner sep=0.75pt]   [align=left] {\llarge $\PP_{N}$};
\draw (418,50) node [anchor=north west][inner sep=0.75pt]   [align=left] {\textbf{\textit{{\llarge Peeling}}}};
\draw (323,33) node [anchor=north west][inner sep=0.75pt]   [align=left] {\llarge $\PP_{N}^{r}$};
\draw (400,232) node [anchor=north west][inner sep=0.75pt]   [align=left] {\textbf{\textit{{\llarge Generating}}}\\\textbf{\textit{{\llarge Difference}}}\\\textbf{\textit{{\llarge Program}}}};
\draw (-97,249) node [anchor=north west][inner sep=0.75pt]   [align=left] {\llarge Simplified $\boldsymbol\partial \PP_{N}$};
\draw (650,232) node [anchor=north west][inner sep=0.75pt]   [align=left] {{\llarge \textbf{\textit{Identifying}}}\\{\llarge \textbf{\textit{Affected}}}\\{\llarge \textbf{\textit{Variables}}}};
\draw (571,230) node [anchor=north west][inner sep=0.75pt]   [align=left] {\llarge $\PP_{N}^{p}$};
\draw (546,266) node [anchor=north west][inner sep=0.75pt]  [font=\small] [align=left] {\textbf{AffectedVars}};
\draw (673,144) node [anchor=north west][inner sep=0.75pt]   [align=left] {\llarge $\PP_{N}^{p}$};
\draw (708,148) node [anchor=north west][inner sep=0.75pt]   [align=left] {\llarge $G_{D}$};
\draw (571,33) node [anchor=north west][inner sep=0.75pt]   [align=left] {\llarge $\PP_{N}^{p}$};
\draw (648,40) node [anchor=north west][inner sep=0.75pt]   [align=left] {{\llarge \textbf{\textit{Computing}}}\\{\llarge \textbf{\textit{DDG}}}};
\draw (155,232) node [anchor=north west][inner sep=0.75pt]   [align=left] {\textbf{\textit{{\llarge Simplifying}}}\\\textbf{\textit{{\llarge Difference}}}\\\textbf{\textit{{\llarge Program}}}};
\draw (324,234) node [anchor=north west][inner sep=0.75pt]   [align=left] {\llarge $\boldsymbol\partial \PP_{N}$};

\end{tikzpicture}

}
\caption{Sequence of steps for the computation of the difference
  program $\partial \PP_N$.}
\label{fig:flow-overview}
\end{figure*}
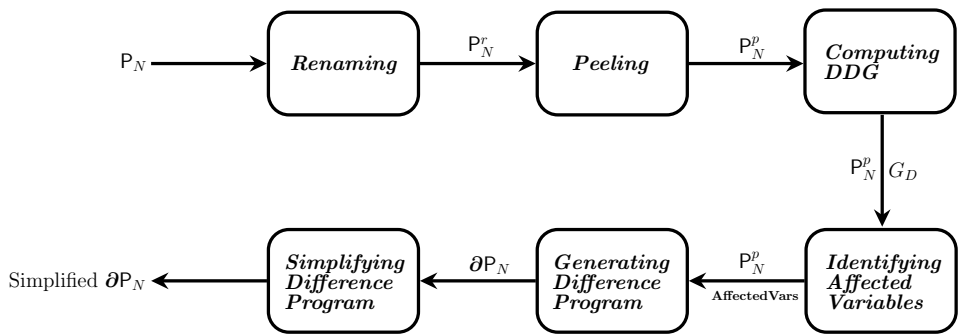

%% file: rename.tex
\subsection{Renaming Variables and Arrays}
\label{sec:renaming}

Recall that our proposed approach requires us to construct a
difference program $\partial \PP_N$ such that $\{\varphi(N)\}$
$\;{\PP_N}\;$ $\{\psi(N)\}$ holds iff $\{\varphi(N)\}$ $\;
{\PP_{N-1}};{\partial \PP_N}\;$ $\{\psi(N)\}$ holds (condition $1$ of
Theorem \ref{thm:full-prog-ind-sound}).  A natural (though not
necessary) way to do this is to construct $\partial \PP_N$ such that
both $\PP_N$ and $\PP_{N-1}; \partial \PP_N$ modify all relevant
scalar variables and arrays in exactly the same way.  Note, however,
that $\PP_N$ may update the same scalar variable or array in multiple
sequentially composed loops.  Therefore, when $\PP_{N-1}$ terminates
and $\partial \PP_N$ starts executing (in $\PP_{N-1}; \partial
\PP_N$), we may no longer have access to the values of scalar
variables and arrays that resulted after individual loops in
$\PP_{N-1}$ terminated.  In general, this makes it difficult to
construct $\partial \PP_N$ compositionally from the peels of
individual loops while ensuring that $\PP_{N-1}; \partial \PP_N$ has
the same effect as $\PP_N$ on all relevant scalar variables and
arrays.  To circumvent this problem, we propose to pre-process $\PP_N$
such that each loop in $\PP_N$ updates its own ``private'' copy of
scalar variables and arrays.  We add \emph{glue code} to copy the
values of these scalar variables and arrays after one loop ends and
before the next one begins.  We also rename the variables/arrays
referred in the post-condition $\psi(N)$ to their versions
corresponding to the last loop in the program.  As we show later, this
eases the construction of $\partial \PP_N$, and also helps in
inductive strengthening of the pre- and post-conditions.

It is important to note here that static single assignment (SSA)
\cite{ssa} is a well-known technique for renaming scalar variables
such that a variable is updated at most once in a program.  Similarly,
array SSA renaming has been studied earlier in the context of
compilers to achieve similar goals \cite{arrayssa}.  Unlike SSA
renaming, we do not have the stringent requirement of a single update
in the whole program.  For our method to function successfully, we
only require each loop to update its own copy of a scalar/array
variable.  We note that these well-studied techniques can be easily
adapted for our purposes.

In the following discussion, we define the \emph{collapsed CFG of a
  program} $\PP_N$ as the CFG obtained by collapsing all nodes and
edges in the body of each loop of $\PP_N$ into a single node
identified with the loop-head.  Given the syntactic restrictions on
the input programs as discussed in Sect. \ref{sec:prelim}, the
collapsed CFG is a finite directed acyclic graph (DAG).  This DAG has
finitely many paths, and along each such path, there is a total
ordering of all collapsed loops appearing along the path.  For
notational clarity, we henceforth use $vA$ (as opposed to $v$ for a
scalar variable and $A$ for an array) as a combined symbolic name to
refer to a scalar variable or array, depending on the context.  At
each node $n$ of the collapsed CFG, we rename each scalar/array $vA$
to $vA^{n}$.  Note that when $n$ is a loop-head, this amounts to
renaming all scalars/arrays in the body of the loop as well (due to
collapsing of nodes in the loop body).  To ensure the correct flow of
data values between nodes of the collapsed CFG, we create fresh nodes
called \emph{glue nodes} whenever required, and add program statements
in these glue nodes that effectively copy values of the appropriate
scalars/arrays from one node of the collapsed CFG to another.


We assume the availability of a function \textsc{Rename}, that
generates a program with the necessary renaming as described, while
ensuring correct data flow.  For notational convenience, we call the
renamed program corresponding to $\PP_N$ as $\PP_N^r$.  The interested
reader can find a detailed algorithm for \textsc{Rename} in
\cite{divyesh-phdthesis}.

\input{ex-ss-rename}

\begin{lemma} \label{lemma:rename:no-overwriting}
  Let $n$ be a node in the collapsed CFG of $\PP_N$.  In every
  execution of the renamed program $\PP_N^r$ in which control flows
  through $n$, the scalar variable/array $vA^n$ is not updated after
  the execution exits $n$.
\end{lemma}
\begin{proof}
  Since the collapsed CFG of $\PP_N^r$ is acyclic, once control flow
  exits node $n$, it cannot come back to either $n$ or to any node
  $n'$ that has a control flow path to $n$.  The proof now follows
  from the observation that renaming ensures that any scalar
  variable/array renamed $vA^n$ can only be updated in glue nodes
  immediately leading to node $n$ or in node $n$ itself.
\end{proof}

For convenience of exposition, we will henceforth refer to the
property formalized in Lemma \ref{lemma:rename:no-overwriting} as the
\emph{``no-overwriting''} property of renamed programs.  For a node
$n$ that corresponds to a collapsed loop in the collapsed CFG of
$\PP_N$, we will also use the notation $\mu(n)$ to denote the entire
loop represented by $n$ in the subsequent discussion.

\begin{lemma} \label{lemma:rename:sound}
  $\{\varphi(N)\}$ $\;\PP_N\;$ $\{\psi(N)\}$ holds iff
  $\{\varphi^r(N)\}$ $\;\PP^r_N\;$ $\{\psi^r(N)\}$ holds.
\end{lemma}
\begin{proof}\textit{(Sketch)}
Suppose the Hoare triple $\{\varphi(N)\}$ $\;\PP_N\;$ $\{\psi(N)\}$
holds.  Then, at every node $n$ in the collapsed CFG of $\PP_N$ there
exists a pre-condition invariant $\inv^{pre}_n$ and a post-condition
invariant $\inv^{post}_n$, such that (i) the Hoare triple
$\{\inv^{pre}_n\}$ $\;\mu(n)\;$ $\{\inv^{post}_n\}$ at node $n$ holds
(ii) the composition of these Hoare triples entails $\{\varphi(N)\}$
$\;\PP_N\;$ $\{\psi(N)\}$.  Note that once renaming is done, the Hoare
triple obtained by renaming variables/arrays in $\inv^{pre}_n$ and
$\inv^{post}_n$ and by replacing $\mu(n)$ with the corresponding
renamed program statement(s) holds iff $\{\inv^{pre}_n\}$ $\;\mu(n)\;$
$\{\inv^{post}_n\}$ holds.  Composing these renamed Hoare triples at
all the nodes in the collapsed CFG of $\PP^r_N$ proves the forward
direction of the lemma.  The case in the reverse direction is similar.
\end{proof}

We end this subsection with an illustration of the program
transformation achieved by applying the renaming strategy mentioned
above.  For convenience we replicate our running example from
Fig. \ref{fig:ss} in Fig. \ref{fig:ss-rename}(a).

\begin{example}
Consider the program shown in Fig. \ref{fig:ss-rename}(a).  This
program has multiple sequentially composed loops that update a scalar
{\tt S} and an array {\tt A}. The transformed program after renaming
the scalar and array variables using function \textsc{Rename} is shown
in Fig. \ref{fig:ss-rename}(b), where we have used simple names for
the renamed versions of {\tt S} and {\tt A} for ease of readability.
Notice that we rename the array {\tt A} in the second loop to {\tt
  A1}, and rename the variable {\tt S} in the third loop to {\tt S1}.
The statement at line 8 in Fig. \ref{fig:ss-rename}(b) is the glue
code to copy values from one version of the renamed scalar/array
(variable {\tt S} in our program) to another version of the same
scalar/array (version {\tt S1}).  We avoid creating new versions of
{\tt S} and {\tt A} for statements and loops that do not update them.
Values are read directly from the version of {\tt S} and {\tt A} that
reaches the access location.  This helps in reducing the glue code
required for renaming quite significantly.
\end{example}

%% file: ex-ss-rename.tex
\begin{figure*}[tb]
  \begin{minipage}{0.12\linewidth}
    ~
  \end{minipage}
\begin{minipage}{0.38\linewidth}
\begin{alltt}
// assume(\(\forall\)i\(\in\)[0,N) A[i] = 1)

1.  S = 0;
2.  for(i=0; i<N; i++) \{
3.    S = S + A[i];
4.  \}

5.  for(i=0; i<N; i++) \{
6.    A[i] = A[i] + S;
7.  \}

8.  for(i=0; i<N; i++) \{
9.    S = S + A[i];
10. \}

// assert(S = N \(\times\) (N+2))
\end{alltt}
\begin{center}
  (a)
\end{center}
\end{minipage}
\begin{minipage}{0.38\linewidth}
\begin{alltt}
// assume(\(\forall\)i\(\in\)[0,N) A[i] = 1)

1.  S = 0;
2.  for(i=0; i<N; i++) \{
3.    S = S + A[i];
4.  \}

5.  for(i=0; i<N; i++) \{
6.    A1[i] = A[i] + S;
7.  \}
{\color{blue}
8.  S1 = S;   // Glue code}
9.  for(i=0; i<N; i++) \{
10.   S1 = S1 + A1[i];
11. \}

// assert(S1 = N \(\times\) (N+2))
\end{alltt}
\begin{center}
  (b)
\end{center}
\end{minipage}
\begin{minipage}{0.1\linewidth}
  ~
\end{minipage}
\caption{(a) Running example and (b) Hoare triple with renamed program}
\label{fig:ss-rename}
\end{figure*}

%% file: peeling.tex
\subsection{Peeling the Loops}
\label{sec:peeling}

\input{algo-loop-peeling} 

Recall from Sect. \ref{sec:overview} that our induction strategy
requires us to use $\PP_{N-1}; \partial \PP_N$ in place of $\PP_N$
when proving the Hoare triple $\{\varphi(N)\}\;\PP_N\;\{\psi(N)\}$.
In general, the parameter $N$ may determine the number of times each
loop in $\PP_N$ iterates (see, for example,
Fig. \ref{fig:ss-rename}(b)).  Therefore, the count of iterations of a
loop in $\PP_{N-1}$ may differ from the corresponding count in
$\PP_N$.  Relating $\PP_N$ and $\PP_{N-1}$ requires taking into
account such differences of loop iterations.  Towards this end, we
transform $\PP_N$ by \emph{peeling} the last few iterations of each
loop as needed, so that corresponding loops in $\PP_{N-1}$ and the
transformed $\PP_N$ iterate the same number of times.  This is done by
function \textsc{PeelAllLoops} shown in Algorithm \ref{alg:peelloops}.
The algorithm first makes a copy, viz. $\PP^p_N$, of the non-collapsed
input CFG $\PP_N$.  Let $\textsc{Loops}(\PP^p_N)$ denote the set of
loops of $\PP^p_N$, and let $k_L(N)$ and $k_L(N-1)$ denote the number
of times loop $L$ iterates in $\PP^p_N$ and $\PP_{N-1}$ respectively.
The difference $k_L(N) - k_L(N-1)$, computed in line
\ref{peel:line:count}, gives the extra iteration count of loop $L$ in
$\PP^p_N$.  If this difference is not a constant, we currently report
a failure of our technique (line \ref{peel:line:fail}).  For example,
consider a loop in $\PP_N$ with the counter $i$ initialized to $0$ and
the loop termination condition ``$i < N^2$''.  The corresponding loop
in $\PP_{N-1}$ has the same initialization but the termination
condition is ``$i < (N-1)^2$''.  Thus, $k_L(N) = N^2$ and $k_L(N-1) =
(N-1)^2$ and the difference of these iteration counts is ``$2 \times N
+1$''.  Our technique is unable to handle such cases currently.  Note
that such cases cannot arise if the upper bounds of all loops in
$\PP_N$ are linear functions of $N$.

The routine \textsc{PeelSingleLoop} transforms loop $L$ of $\PP^p_N$
as follows: it replaces the termination condition $(\ell < k_L(N))$ of
$L$ by $(\ell < k_L(N-1))$.  It also peels the last $(k_L(N) -
k_L(N-1))$ iterations of $L$ and adds control flow edges such that the
peeled iterations are executed immediately after the loop body is
iterated $k_L(N-1)$ times.  Effectively, \textsc{PeelSingleLoop} peels
the last $(k_L(N)-k_L(N-1))$ iterations of loop $L$ in $\PP^p_N$. The
transformed CFG is returned as the updated $\PP^p_N$ in line
\ref{peel:line:peelloop}.  In addition, \textsc{PeelSingleLoop} also
returns the set $Locs'$ of all CFG nodes newly added while peeling the
loop $L$.  We accumulate these newly added nodes for loops in the set
$\peelnodes$ in line \ref{peel:line:collect-peel-nodes}.  Henceforth,
we call all nodes in $\peelnodes$ as \emph{peeled nodes}, all other
nodes in the CFG as \emph{non-peeled nodes}, and the CFG resulting
from the invocation of \textsc{PeelAllLoops} as a \emph{peeled program}.
This function \textsc{PeelAllLoops} returns the peeled program $\PP^p_N$
and the set of peeled nodes $\peelnodes$ in line
\ref{peel:line:return}.  

We now state several useful properties of peeled programs.

\begin{lemma}\label{lemma:branches-n-peels}
  Let $n \in \mathit{Locs}^p$ be a node in the peel of loop $L$, and
  let $n_h$ be the loop-head of loop $L$.  For every $n' \in
  \mathit{Locs}^p$ that is not in the peel, if there is a control flow
  path in $\PP_N^p$ from $n'$ to $n$, the path necessarily passes
  through $n_h$.
\end{lemma}
\begin{proof}
  The proof follows from the observation that the peel of a loop $L$
  must necessarily execute after the loop $L$ has itself executed
  $k_L(N-1)$ times. Hence, the sole predecessor of the first node in
  the peel must be the loop-head node $n_h$. It follows that every
  control flow path from $n'$ to $n$ where $n'$ is not in the peel
  must pass through $n_h$.
\end{proof}

Peeling of loops can destroy the no-overwriting property (as mentioned
in Lemma~\ref{lemma:rename:no-overwriting}), since the same
variable/array can get updated in a loop $L$ and also in its peel.
However, a weaker variant of the no-overwriting property continues to
hold, as described below.

\begin{lemma}\label{lemma:peel:no-overwriting}
  Let $n$ be a node in the collapsed CFG of $\PP_N$.  In every
  execution of the renamed and peeled program $\PP^p_N$ in which
  control flows through $n$, the following hold.
  \begin{enumerate}
  \item If $n$ is not a collapsed node, the no-overwriting property as
    in Lemma~\ref{lemma:rename:no-overwriting} holds for all scalar
    variables/arrays $vA^n$.
  \item If $n$ is a node representing a collapsed loop $L$, the scalar
    variable/array $vA^n$ is not updated at any subsequent node along
    the execution, except possibly in the peel of $L$.
  \end{enumerate}
\end{lemma}
\begin{proof}
  Follows from the same reasoning as used in the proof of
  Lemma~\ref{lemma:rename:no-overwriting}.
\end{proof}

We will henceforth refer to the property formalized in Lemma
\ref{lemma:peel:no-overwriting} as the \emph{``no-overwriting''}
property of the renamed and peeled program $\PP_N^p$.

\begin{lemma}
  \label{lemma:alg-close-cond-branch}
  In the peeled program $\PP_N^p$, each conditional branch node in a
  peel of a loop has an immediate post-dominator in the same peel.
\end{lemma}
\begin{proof}
  The syntactic restrictions on the input program, imposed by the
  grammar shown in Fig.~\ref{fig:grammar}, do not admit \emph{break,
    continue, goto, exit} and \emph{return} statements.  Since a loop
  body is also syntactically a complete program, conditional branch
  nodes in the body of the loop, if any, always have an immediate
  post-dominator node within the body of the same loop.  The peel of a
  loop is obtained by creating a copy of the loop body (using the
  function \textsc{PeelSingleLoop} invoked on line
  \ref{peel:line:peelloop} of routine \textsc{PeelAllLoops} in
  Algorithm \ref{alg:peelloops}).  Thus, the conditional branch nodes,
  if any, in the peel have an immediate post-dominator node within the
  same peel.
\end{proof}

Finally, the following lemma asserts that peeling does not change the
Hoare semantics of programs.

\begin{lemma}
\label{lemma:alg-peelloops}
$\{\varphi_N\}\;\PP_N\;\{\psi_N\}$ holds iff
$\{\varphi_N\}\;\PP^p_N\;\{\psi_N\}$ holds.
\end{lemma}
\begin{proof}
  Follows immediately from the observation that peeling each loop
  preserves the semantics of the program.
\end{proof}

\input{ex-ss-peel}  

\begin{example}
  We execute function \textsc{PeelAllLoops} on the renamed version of
  our running example, shown in Fig. \ref{fig:ss-rename}(b).  The
  resulting program is shown in Fig. \ref{fig:ss-peel}.  The algorithm
  first computes the number of iterations to be peeled from a loop in
  the program, given by $\peelcount$.  The upper bound expression of
  each loop in the program is $N$.  Hence, the number of iterations to
  be peeled is $N - (N-1)$ $=$ $1$.  In other words, only the last
  iteration is to be peeled from each loop.  The function appends the
  statements in the peeled iteration after each loop and updates the
  upper bound expressions of each loop in the resulting program, as
  shown in Fig. \ref{fig:ss-peel}.  The algorithm also returns the set
  of peeled nodes, i.e. CFG nodes corresponding to the statements at
  lines $5$, $9$, and $14$.
\end{example}

%% file: algo-loop-peeling.tex
\begin{algorithm*}[!t]
  \caption{\textsc{PeelAllLoops}$\left((\mathit{Locs}, \mathit{CE}, \mu): \mbox{ program } \PP_N \right)$}
  \label{alg:peelloops}
  \begin{algorithmic}[1]
    \State $\PP^p_N := (\mathit{Locs}^p, \mathit{CE}^p, \mu^p)$, where $\mathit{Locs}^p = \mathit{Locs}$, $\mathit{CE}^p = \mathit{CE}$, $\mu^p = \mu$; \Comment{$\PP^p_N$ is a copy of $\PP_N$}
    \State $\peelnodes$ := $\varnothing$;
    \For{each loop $L \in \textsc{Loops}( \PP^p_N )$}
      \label{alg:line:canon-peel-loop}
      \State Let $k_L({N})$ be the expression for iteration count of $L$ in $\PP^p_N$;
      \State $\peelcount := \textsc{Simplify}(k_L({N})-k_L({N-1}))$;  \label{peel:line:count}
      \If{$\peelcount$ is not a constant} \label{peel:line:check}
        \State {\bf throw} ``Failed to peel non-constant number of iterations''; \label{peel:line:fail}
      \EndIf
      \State $\langle\PP^p_N, \mathit{Locs}'\rangle := \textsc{PeelSingleLoop}(\PP^p_N, L, k_L({N-1}), \peelcount)$; \label{peel:line:peelloop}
      \infocomment{We assume availability of function \textsc{PeelSingleLoop}, for example, from a compiler framework like LLVM.}
      \infocomment{It transforms loop $L$ so that last $\peelcount$ iterations of $L$ are peeled.}
      \infocomment{Updated CFG and newly created CFG nodes for the peeled iterations are returned.}
      \State $\peelnodes$ := $\peelnodes$ $\union$ $\mathit{Locs}'$;  \label{peel:line:collect-peel-nodes}
    \EndFor
    \State \Return $\langle\PP^p_N, \peelnodes\rangle$;  \label{peel:line:return}
  \end{algorithmic}
\end{algorithm*}

%% file: ex-ss-peel.tex

\begin{figure}[h]
\begin{alltt}
// assume(\(\forall\)i\(\in\)[0,N) A[i] = 1)

1.  S = 0;
2.  for(i=0; i<N-1; i++) \{
3.    S = S + A[i];
4.  \}
5.  S = S + A[N-1];

6.  for(i=0; i<N-1; i++) \{
7.    A1[i] = A[i] + S;
8.  \}
9.  A1[N-1] = A[N-1] + S;

10. S1 = S;
11. for(i=0; i<N-1; i++) \{
12.   S1 = S1 + A1[i];
13. \}
14. S1 = S1 + A1[N-1];

// assert(S1 = N \(\times\) (N+2))
\end{alltt}
\caption{Hoare triple containing a program with peeled loops}
\label{fig:ss-peel}
\end{figure}


%% file: refine-pdg.tex
\subsection{Tracking Data Dependencies}
\label{sec:refine-pdg}


As discussed in Sect.~\ref{sec:cfg}, the flow of control in a
program can be conveniently represented by a CFG.  A CFG, however,
does not immediately provide information about data dependencies
between program statements.  We use a separate \emph{data dependence
  graph}, or DDG, to summarize data dependencies between relevant
statements in a program.  Our primary purpose in constructing such a
DDG is to understand the dependencies of and from statements that are
executed in $\PP_N$ but not in $\PP_{N-1}$.  These are related to the
peeled statements described in Sect.~\ref{sec:peeling}, and
determine what must eventually go into the difference program
$\partial \PP_N$, so that $\PP_N$ and $\PP_{N-1};\partial \PP_N$ have
the same effect on arrays and scalar variables.

While there are several notions of data dependence used in the
literature (see \cite{ddg1,ddg2} for details), we use a fairly simple
notion that best serves our purpose.  We say that there is a
read-after-write data dependence from $n_1$ to $n_2$ if the statement
at $n_2$ uses a data value that is potentially generated by the
statement at $n_1$.  There is another kind of data dependence that is
peculiar to our approach that also needs special handling.  It may so
happen that the glue code inserted between two nodes during renaming
has a loop, say $L_1$, that updates an array $A$ that is also
subsequently updated in another loop, say $L_2$ in non-glue code.  If
the peel of $L_1$ potentially updates an element of $A$ that is also
updated in the non-peeled part of $L_2$, then we have a
write-after-write dependence between a statement in the peel of a
(glue) loop and subsequent statement in the non-peeled part of another
(non-glue) loop.  We call such a dependence
\emph{non-peeled-write-after-peeled-write} dependence.  Since we
intend to move peels to the end of the program to construct a
difference program, this kind of dependence poses a problem.
Therefore, we explicitly identify such
\emph{non-peeled-write-after-peeled-write} dependencies below.  Given
the way our renaming operates, it is easy to see that such a
dependence can only arise for arrays and not for scalars.

Note that in the above case when we have a glue loop followed by a
non-glue loop updating the same array, there may also be
write-after-write dependencies between the non-peeled (resp. peeled)
part of the glue loop and the non-peeled (resp. peeled) part of the
non-glue loop.  However, such dependencies are preserved if we move
peels of all loops to the end of the program to construct a difference
program.  Therefore, such write-after-write dependencies do not pose
any problem for our purposes, and hence we do not keep track of them.
Furthermore, due to the way our renaming operates, it can be seen that
write-after-read dependencies can never arise between nodes of the
collapsed CFG.

Formally, a DDG is a directed graph $G_D = (\mathit{Locs},
\mathit{DE}, \mu)$, where $\mathit{Locs}$ and $\mu$ are exactly as in
the definition of a CFG, and $\mathit{DE} \subseteq \mathit{Locs}
\times \mathit{Locs}$ represents \emph{read-after-write} and
\emph{non-peeled-write-after-peeled-write} dependencies between
statements in the program.  Since our primary interest is in using
data dependencies to and from peeled statements for purposes of
constructing difference programs, and since loops have a very specific
form in our programs of interest (the grammar in
Fig.~\ref{fig:grammar} allows only loop counter to be updated in a
loop-head node), it suffices to restrict our attention to data
dependencies between distinct non-loop-head nodes in the peeled
program.

Several existing compilers generate \emph{program dependence graph},
or PDG from a given input program, and a DDG can be extracted from
such a PDG~\cite{pdg}.  Standard dataflow analysis techniques are
usually used to identify data dependencies when constructing a
PDG~\cite{pdg,pdg2}.  One needs to be particularly careful when
identifying dependence between statements updating and accessing array
elements, since it is not only the same array name that must be
involved in the update and access, but also the same element in the
array.  The problem is further compounded by the fact that array
indices can be arbitrary expressions in general.  While vectorizing
compilers can compute precise dependencies with array index expressions
using sophisticated {\em dependence tests}~\cite{kennedy-book}, it is
not always the case that these are implemented in non-vectorizing
compilers.  A conservative generation of DDG may contain spurious data
dependence edges, which, in our context, can lead to the construction
of a difference program that is more complex than what is needed.

Let $n$ and $n'$ be two CFG (hence also DDG) nodes.  A conservative
way of generating DDG edges is to add the edge $(n, n')$ to
$\mathit{DE}$ if there is a control flow path $\pi: (n=n_1, n_2,
\ldots, n_{t-1}, n_t = n')$ in the CFG such that one of the following
conditions hold.
\begin{enumerate}
  \item [D1:] $\big(\mathit{def}(n) \cap \mathit{uses}(n')\big)
    \setminus \bigcup_{i=2}^{t-1} \mathit{def}(n_i)$ contains a scalar
    variable $v$, or
  \item [D2:] $\mathit{def}(n)$ contains an array $A$ such that
    \begin{enumerate}
    \item Either of the following conditions hold:
      \begin{enumerate}
      \item $A \in \mathit{uses}(n')$ and there is a common value that
        the index expression $\defindex(A,n)$ and some index
        expression in $\useindex(A,n')$ can have.
        \item $n$ $\in$ $\peelnodes$ and $n'$ $\not \in$ $\peelnodes$
          and $A \in \mathit{def}(n')$ and there is a common value
          that both the index expressions $\defindex(A,n)$ and
          $\defindex(A,n')$ can have.
      \end{enumerate}
      \item Some elements of $A$ are potentially not updated along the
        path $\pi$.
    \end{enumerate}
\end{enumerate}

\begin{lemma}\label{lem:DDG-conditions}
  For $n, n' \in \mathit{Locs}$ such that $n'$ is reachable from $n$
  in the CFG, if neither condition D1 nor condition D2 holds, then
  there is no read-after-write or non-peeled-write-after-peeled-write
  dependence from $n$ to $n'$.
\end{lemma}
\begin{proof}
  We prove the lemma by contradiction.  Suppose, if possible, neither
  D1 nor D2 holds and yet there is a \emph{read-after-write}
  dependence due to the data value generated at $n$ being potentially
  used at $n'$. There are two cases to consider.
  \begin{itemize}
  \item If the data value pertains to a scalar variable $v$ that is
    updated at $n$ and accessed at $n'$, then there must be a control
    flow path $\pi$ from $n$ to $n'$ along which $v$ is not updated at
    any intermediate node.  This implies condition D1 is satisfied --
    a contradiction!
  \item Suppose the data value pertains to an element of array $A$
    that is updated at $n$ and accessed at $n'$.  Let the index
    expression of the array element updated at $n$ be $e$, and let the
    corresponding index expression of the same element accessed at
    $n'$ be $e'$.  Clearly, both $e$ and $e'$ can assume the same
    value (the concrete index of the element of $A$ under
    consideration), and there is a control flow path $\pi$ from $n$ to
    $n'$ along which this specific array element has not been updated.
    This implies that both the conditions D2(a)(i) and D2(b) are
    satisfied, and hence condition D2 is satisfied -- a contradiction!
  \end{itemize}

  Suppose, if possible, neither D1 nor D2 holds and yet there is a
  \emph{non-peeled-write-after-peeled-write} dependence from $n$ to
  $n'$.  Suppose the data value pertains to an element of array $A$
  that is updated at nodes $n$ and $n'$ where $n$ $\in$ $\peelnodes$
  and $n'$ $\not \in$ $\peelnodes$.  Let the index expression of the
  array element updated at $n$ be $e$, and let the corresponding index
  expression of the same element updated at $n'$ be $e'$.  Clearly,
  both $e$ and $e'$ can assume the same value (the concrete index of
  the element of $A$ under consideration), and there is a control flow
  path $\pi$ from $n$ to $n'$ along which this specific array element
  has not been updated.  This implies that both the conditions
  D2(a)(ii) and D2(b) are satisfied, and hence, condition D2 is
  satisfied -- a contradiction!
\end{proof}

Condition D2(b) above is not easy to check in general.  However, for
programs generated by the grammar in Fig.~\ref{fig:grammar}, it is
possible to detect that condition D2(b) is violated in special cases.
As an example, if there is a loop that updates all elements of array
$A$ in every control flow path from $n$ to $n'$, then indeed condition
D2(b) is violated.  For purposes of this paper, we use this special
case as a sufficient condition to detect violation of condition D2(b),
and conservatively assume that the condition is potentially satisfied
in all other cases.  Needless to say, a more precise analysis can be
done with additional computational effort to reduce the degree of
conservativeness in the above check for condition D2(b).  We defer
such an improved analysis to future work.  We now look at how
condition D2(a) is checked.  Recall from Sect.~\ref{sec:prelim} that
the array indices in our programs can only be expressions in terms of
constants, scalar variables, the loop counter variables and the
parameter $N$.  Furthermore, our programs do not have nested loops.
Therefore, at most one loop counter variable can appear in an array
index expression.  Specifically, if $e$ is the index expression
$\defindex(A, n)$, and if node $n$ is part of a loop $L$ with loop
counter $\ell$, then $e$ depends in general on $\ell$, $N$ and a set
of scalar constants.  Otherwise, i.e. if node $n$ is not part of a
loop, $e$ depends on $N$ and a set of scalar constants.  A similar
reasoning applies for array index expression(s) in $\useindex(A, n')$
as well.  Condition D2(a) is satisfied if the constraint $(e = e')$
has a model, i.e.  is satisfiable, for some index expression $e' \in
\useindex(A,n')$, subject to the following conditions:
\begin{itemize}
  \item Loop counters $\ell$ and $\ell'$ must have values within their
    respective loop bounds.
  \item If both $n$ and $n'$ are part of the same loop, then $\ell \le
    \ell'$ (update at $n$ cannot happen in an iteration after access
    at $n'$).
  \item Every scalar variable $v$ that appears in both $e$ and $e'$
    and is updated along some control flow path from $n$ to $n'$ is
    renamed in $e'$ to a fresh variable (since the values of $v$ in
    $e$ and $e'$ may be different).
\end{itemize}



\input{algo-compute-pdg1}  

Function \textsc{ComputeDDG}, shown in Algorithm
\ref{alg:compute-ddg}, constructs the DDG for an input program $\PP_N$
represented using its CFG $(\mathit{Locs}, \mathit{CE}, \mu)$.  We use
the notation $n \stackrel{X}{\rightsquigarrow} n'$ to denote that
there is a control flow path from $n$ to $n'$ in the CFG that passes
through intermediate nodes in $X \subseteq \mathit{Locs}$.
\textsc{ComputeDDG} proceeds by initializing the set of data
dependence edges $DE$ to $\emptyset$, and by checking for every pair
of distinct non-loop-head nodes $(n, n')$ such that $n
\stackrel{\mathit{Locs}}{\rightsquigarrow} n'$, whether condition D1
or D2 referred to above is satisfied.  If either one of the conditions
is satisfied, it adds $(n, n')$ to $DE$
(lines~\ref{alg:line:d1d2check}--\ref{alg:line:d1d2checkend}).

The set $S$ of scalar variables and arrays that potentially introduce
data dependence from $n$ to $n'$ is initialized to $\mathit{def}(n)
\cap \mathit{uses}(n')$ in line~\ref{alg:line:initS}.  If $n$ is a
peeled node and $n'$ is a non-peeled node, then we append $S$ with
$\mathit{def}(n) \cap \mathit{def}(n')$ in
line~\ref{alg:line:appendS}.  Subsequently, the check for D1 is done
in the loop in
lines~\ref{alg:line:checkford1start}--\ref{alg:line:checkford1end}.
In each iteration of this loop, we choose a scalar variable $v$ from
the set $S$ and check whether there exists a control flow path from
$n$ to $n'$ such that no intermediate node along the path updates $v$.
The latter check is implemented by first collecting all nodes $n''$
(other than $n$ and $n'$) that doesn't update $v$ in the set
$\NonDefV$ (line~\ref{alg:line:nondefvdefn}).  If there is a control
flow path from $n$ to $n'$ that passes through intermediate nodes in
$\NonDefV$, the value of $v$ updated at $n$ can reach the use of $v$
at $n'$.  In this case, there is a potential data dependence of $n'$
on $n$ through $v$ and condition D1 is satisfied.  We therefore set
$\mathit{D1Sat}$ to $\true$ and abort the search over additional
scalar variables $v$ in $S$ (line~\ref{alg:line:foundd1sat}).
Otherwise, there is no dependence of $n'$ on $n$ through $v$.

If the flag $\mathit{D1Sat}$ is not set to $\true$ even after
iterating over all scalar variables in $S$, we turn to checking if
condition D2 can be satisfied. Towards this end, we iterate over all
array names $A$ remaining in $S$, and formulate a constraint to check
if condition D2(a) is satisfied
(lines~\ref{alg:line:d2aconstraintstart}--\ref{alg:line:d2aconstraintend}).
This condition effectively checks if it is possible for the index
expression $\defindex(A,n)$ to have the same value as any index
expression $e' \in \useindex(A, n')$ or if it is possible for the
index expression $\defindex(A,n)$ to have the same value as the index
expression $\defindex(A, n')$ when $n$ is a peeled node and $n'$ is a
non-peeled node.  If not, the update/read of array $A$ at $n$ and
$n'$, cannot be for the same element, and hence, there is no data
dependence $(n, n')$ through $A$.  As discussed above, to check if
condition D2(a) is satisfied, we must conjoin loop bound constraints
for loop counter variables in case $n$ or $n'$ is present in a loop
(lines~\ref{alg:line:d2aconstraintstart}--\ref{alg:line:range-end}),
and rename every scalar variable $v$ in expressions $e' \in
\useindex(A, n')$ that also appears in the index expression
$\defindex(A, n)$, if $v$ is potentially re-defined in a control flow
path from $n$ to $n'$.  The renaming of scalar variables, if needed,
is done in lines~\ref{alg:line:renamestart}--\ref{alg:line:renameend}.
The call to $\textsc{IsSat}$ in line~\ref{alg:line:checkoverlap} is an
invocation of an SMT solver that tells us whether the constraint fed
to it as argument is satisfiable, i.e. has a model.  If not, condition
D2(a), and hence D2, is violated.  Otherwise, we check in
lines~\ref{alg:line:d2bcheck1} and \ref{alg:line:d2bcheck2} if there
exists a loop $L''$ not containing $n$ and $n'$ that necessarily
executes as control flows from $n$ to $n'$ ($n
\stackrel{\mathit{Locs}\setminus\{n''\}}{\not\rightsquigarrow} n'$
checks this), and in which all elements of the array $A$ are updated.
Recall from the grammar in Fig.~\ref{fig:grammar} that all loops in
our programs are \textbf{for} loops with a loop counter that
increments by $1$ in each operation, and cannot be updated in the body
of the loop.  For such programs, it is sometimes easy to identify if a
loop $L''$ is indeed updating all elements of an array $A$.  If we
cannot determine whether $L''$ necessarily updates all elements of
$A$, we conservatively assume that it does not and the check in
line~\ref{alg:line:d2bcheck2} fails.  If both the checks in
lines~\ref{alg:line:d2bcheck1} and \ref{alg:line:d2bcheck2} succeed,
we conclude that condition D2(b), and hence D2, has been violated.  In
all other cases, we conservatively assume that D2 is satisfied, and
set $\mathit{D2Sat}$ to $\true$ in line~\ref{alg:line:d2sattrue}.  In
such cases, we also abort the search over additional array variables
$A$ in $S$.
\begin{lemma}\label{lem:compute-ddg-correct}
  Given a program represented as $(\mathit{Locs}, CE, \mu)$, let
  $(\mathit{Locs}, DE, \mu)$ be the DDG computed by
  \textsc{ComputeDDG}.  For every pair of distinct non-loop-head nodes
  $n, n' \in \mathit{Locs}$, if $(n, n') \not\in DE$, there is no
  read-after-write or non-peeled-write-after-peeled-write data
  dependence from $n$ to $n'$.
\end{lemma}
\begin{proof}
  Since function \textsc{ComputeDDG} implements the checks for
  conditions D1, D2(a) and D2(b) in a straightforward manner, the
  proof follows from Lemma~\ref{lem:DDG-conditions}.
\end{proof}

We conclude this subsection with an illustration of DDG edges computed
by \textsc{ComputeDDG} for our running example.

\input{draw-ddg}

\begin{example}
  Our running example with peeled loops is shown in
  Fig.~\ref{fig:ss-peel}.  The CFG for this program is shown using
  solid edges in Fig.~\ref{fig:draw-ddg}. For convenience of
  exposition, we have named nodes such that node $n_i$ in the CFG of
  this program corresponds to the statement at line $i$ of the peeled
  program, with two special nodes $n_{start}$ and $n_{end}$, as usual.
  If we execute function \textsc{ComputeDDG} on this CFG, we obtain
  the data dependence edges shown using dashed edges in
  Fig.~\ref{fig:draw-ddg}.  For ease of understanding, each DDG edge
  $(n, n')$ is also labeled by a scalar variable/array that is
  responsible for the data dependence of $n'$ on $n$.  Thus, DDG edges
  $(n_1, n_3), (n_1, n_5), (n_3, n_5), (n_5, n_7), (n_5, n_9)$ and
  $(n_5, n_{10})$ represent data dependence through the scalar
  variable {\tt S} and edges $(n_{10}, n_{12}), (n_{10}, n_{14}),
  (n_{12}, n_{14})$ represent data dependence through the scalar
  variable {\tt S1}.  In all these cases, condition $D1$ holds.
  Similarly, DDG edges $(n_7, n_{12}), (n_{9}, n_{14})$ represent data
  dependence through the array {\tt A1}, since conditions D2(a) and
  D2(b) hold in these cases.  Note that edge $(n_7, n_{14})$
  (resp. $(n_9, n_{12})$) is not added although ${\mathtt A1} \in
  \mathit{def}(n_7) \cap \mathit{uses}(n_{14})$ (resp. $\in
  \mathit{def}(n_9) \cap \mathit{uses}(n_{12})$) because condition
  D2(a) fails in this case.  To see why D2(a) fails, notice that
  $\phi_7 := 0 \leq \mathtt{i} < \mathtt{N}-1$ and $\phi_{14} :=
  \true$.  The expressions used to define and access array
  $\mathtt{A}$ are $\defindex(\mathtt{A},7) := \mathtt{i}$ and
  $\useindex(\mathtt{A}, 14) := \{\mathtt{N}-1\}$.  The constraint $0
  \leq \mathtt{i} < \mathtt{N}-1 \wedge \mathtt{i} = \mathtt{N}-1$,
  computed in line~\ref{alg:line:checkoverlap} of
  Algorithm~\ref{alg:compute-ddg}, is unsatisfiable.  This violates
  D2(a)(i).  Note that, in this example there are no
  non-peeled-write-after-peeled-write dependencies.
\end{example}

%% file: algo-compute-pdg1.tex
\begin{algorithm*}[!t]
  \caption{\textsc{ComputeDDG}($(\mathit{Locs}, \mathit{CE}, \mu)$: program $\PP_N$, $\peelnodes$: peeled statements)}
  \label{alg:compute-ddg}
  \begin{algorithmic}[1]
    \State $\mathit{DE}$ := $\emptyset$;  
    \For{each $n, n' \in \mathit{Locs}\setminus \mathit{LoopHeads}$ s.t. $n \neq n'$ and
      $n \stackrel{\mathit{Locs}}{\rightsquigarrow} n'$}\Comment{$n'$ reachable from $n$ in CFG}
      \State $\mathit{D1Sat}$ := $\false$; $\mathit{D2Sat}$ := $\false$; \Comment{Initializing flags to indicate if condition D1/D2 holds}
      \State $S$ := $\mathit{def}(n) \cap \mathit{uses}(n')$;\Comment{\parbox[t]{0.6\textwidth}{Set of scalar variables/arrays that potentially introduce data dependence between $n$ and $n'$}} \label{alg:line:initS}
      \If{$n \in \peelnodes \wedge n' \not\in \peelnodes$} $S := S \cup \big( \mathit{def}(n) \cap \mathit{def}(n')\big)$; \label{alg:line:appendS} \EndIf \Comment{For non-peeled-write-after-peeled-write dependence}
      \infocomment{Check for condition D1}
         \For{each scalar variable $v \in S$} \label{alg:line:checkford1start}
            \State $\NonDefV$ := $\{n'' \mid n'' \neq n, n'' \neq n',  v \not\in \mathit{def}(n'')\}$;\Comment{Nodes other than $n$, $n'$ that don't define $v$} \label{alg:line:nondefvdefn}

          \If{$n \stackrel{\NonDefV}{\rightsquigarrow} n'$}\Comment{Potential data dependence $(n, n')$ due to variable $v$}
            \State $\mathit{D1Sat}$ := $\true$; \textbf{break}; \label{alg:line:foundd1sat}       
          \EndIf
        
        \EndFor \label{alg:line:checkford1end}
      \infocomment{Check for condition D2 if D1 isn't already satisfied}
      \If{(\textbf{not} $\mathit{D1Sat}$)}  
        \For{each array $A \in S$}
          \State $\phi_n$ := $\true$; $\phi_{n'}$ := $\true$;\Comment{Initializing loop bound constraints for $n$ and $n'$} \label{alg:line:d2aconstraintstart}
          \If{$n$ is part of loop $L$ with loop counter $\ell$} $\phi_n$ := $(0 \leq \ell < k_{L})$; \label{alg:line:range-start}
          \EndIf
          \If{$n'$ is part of a loop $L'$ with loop counter $\ell'$}
            $\phi_{n'}$ := $(0 \leq \ell' < k_{L'})$;
            \If{$L'$ same as $L$}
               $\phi_{n'}$ := $\phi_{n'} \wedge (\ell \leq \ell')$; \Comment{$n$ and $n'$ in the same loop}
            \EndIf
          
          \EndIf \label{alg:line:range-end}

          \State $e$ := $\defindex(A,n)$; \Comment{Index expression used to update element of $A$ at $n$} \label{alg:line:renamestart}
          \State $U := \useindex(A, n')$;
          \If{$n \in \peelnodes \wedge n' \not\in \peelnodes$} $U := U \cup \{ \defindex(A, n') \}$; \EndIf
          \For{each scalar variable $v$ that appears in both $e$ and some $e' \in U$}
            \State $\DefV$ := $\{n'' \mid n'' \neq n, n'' \neq n', v \in \mathit{def}(n'')\}$; \Comment{Nodes other than $n$, $n'$ that define $v$}
            \If{$\exists n'' \in \DefV \wedge n \rightsquigarrow n'' \wedge n'' \rightsquigarrow n'$} \Comment{$v$ modified along a path from $n$ to $n'$}
             \State $e'$ := $e'[v/v_{\mathit{fresh}}]$; \Comment{Rename variable $v$ in $e'$ with fresh variable $v_{\mathit{fresh}}$}
            \EndIf
          \EndFor   \label{alg:line:renameend}

          \If{\textbf{not} \textsc{IsSat}$\big(\phi_n \wedge \phi_{n'} \wedge \bigvee_{e' \in U}(e = e')\big)$} \label{alg:line:checkoverlap}
             {\bf continue};  \Comment{D2(a) violated for array $A$}\label{alg:line:removeedge}
          \EndIf  \label{alg:line:d2aconstraintend}
          \If{$\exists$ loop $L''$ with loop-head $n''$ s.t. ($n, n'$ not in $L''$) $~\wedge$ ($n \stackrel{\mathit{Locs}\setminus \{n''\}}{\not\rightsquigarrow} n'$)} \label{alg:line:d2bcheck1}
            \If{$L''$ necessarily updates all elements of $A$} \label{alg:line:d2bcheck2}
              \State {\bf continue};   \Comment{Condition D2(b) violated for array $A$} \label{alg:line:d2bviolation}
            \EndIf
            \EndIf

          \State $\mathit{D2Sat}$ := $\true$; \textbf{break}; \label{alg:line:d2sattrue} \Comment{Potential data dependence $(n, n')$ due to array $A$}

        \EndFor
      \EndIf    
      \If{$\big(\mathit{D1Sat} \vee \mathit{D2Sat}\big)$} \label{alg:line:d1d2check}
          $\mathit{DE}$ := $\mathit{DE} \cup \{(n, n')\}$;
       \EndIf \label{alg:line:d1d2checkend}
    \EndFor \label{alg:line:loopedges-end}
    \State \Return $(\mathit{Locs}, DE, \mu)$;
  \end{algorithmic}
\end{algorithm*}

%% file: draw-ddg.tex
\begin{figure}[h]
  \begin{center}
    \scalebox{0.6}{
    \begin{tikzpicture}
      [
        ->,
        >=stealth,
        node distance=1cm,
        noname/.style={
          ellipse,
          very thick,
          fill=white,
          minimum width=1em,
          minimum height=1em,
          draw
        }
      ]

      \node[noname] (s)               {Start};
      \node[noname] (1)  [right=of s]  {1};
      \node[noname] (2)  [below=of 1]  {2};
      \node[noname] (3)  [right=of 2]  {3};
      \node[noname] (4)  [right=of 3]  {4};
      \node[noname] (5)  [below=of 2]  {5};
      \node[noname] (6)  [below=of 5]  {6};
      \node[noname] (7)  [right=of 6]  {7};
      \node[noname] (8)  [right=of 7]  {8};
      \node[noname] (9)  [below=of 6]  {9};
      \node[noname] (10) [below=of 9]  {10};
      \node[noname] (11) [below=of 10] {11};
      \node[noname] (12) [right=of 11] {12};
      \node[noname] (13) [right=of 12] {13};
      \node[noname] (14) [below=of 11] {14};
      \node[noname] (e)  [left=of 14] {End};

      \begin{scope}[very thick,rounded corners,-latex]
        \path
        (s)  edge (1)
        (1)  edge (2)
        (2)  edge node [above, xshift=-2mm] {$~~~\ltrue$} (3)
        (3)  edge (4)
        (4)  edge [bend right=55pt] node {} (2)
        (2)  edge node [above, xshift=2mm,yshift=-2mm] {$\lfalse$} (5)
        (5)  edge (6)      
        (6)  edge node [above, xshift=-2mm] {$~~~\ltrue$} (7)
        (7)  edge (8)
        (8)  edge [bend right=55pt] node {} (6)        
        (6)  edge node [above, xshift=2mm,yshift=-2mm] {$\lfalse$} (9)
        (9)  edge (10)
        (10) edge (11)
        (11) edge node [above, xshift=-2mm] {$~~~\ltrue$} (12)
        (12) edge (13)
        (13) edge [bend right=55pt] node {} (11)
        (11) edge node [above, xshift=2mm,yshift=-2mm] {$\lfalse$} (14)
        (14) edge (e)
        ;
      \end{scope}

      \begin{scope}[dashed,very thick,rounded corners,-latex]
        \path
        (1.east) edge [bend left=30] node [xshift=1mm, yshift=2mm] {\tt S} (3)
        (1) edge [bend right=40] node [xshift=-2mm] {\tt S} (5)
        (3) edge [bend left=30] node [xshift=3mm] {\tt S} (5.30)
        (5) edge [bend left=30] node [xshift=1mm, yshift=2mm] {\tt S} (7)
        (5) edge [bend right=40] node [xshift=-2mm] {\tt S} (9)
        (5) edge [bend right=90] node [xshift=-2mm] {\tt S} (10.west)
        (7) edge node [xshift=3mm] {\tt A1} (12)
        (9) edge [bend right=90] node [xshift=-3mm] {\tt A1} (14.150)
        (10.east) edge [bend left=30] node [xshift=1mm, yshift=2mm] {\tt S1} (12.140)
        (10) edge [bend right=40] node [xshift=-3mm] {\tt S1} (14)
        (12) edge [bend left=30] node [xshift=3mm] {\tt S1} (14)
        ;
      \end{scope}
      
    \end{tikzpicture}
    }
  \end{center}
  \caption{CFG (solid edges) and DDG (dashed edges) of peeled program in Fig.~\ref{fig:ss-peel}}
  \label{fig:draw-ddg}
\end{figure}
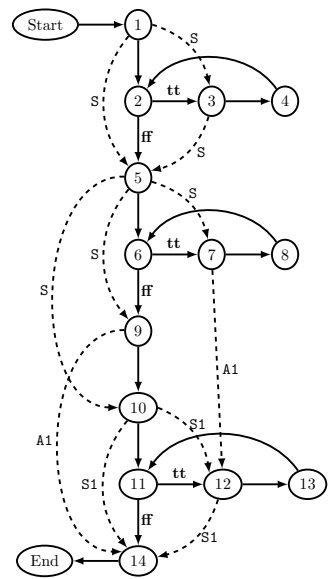

%% file: compute-affected.tex
\subsection{Identifying ``Affected'' Variables}
\label{sec:affected}

Recall that every loop $L$ originally present in $\PP_N$ iterates
$k_L(N-1)$ times in $\PP_{N-1}$ and $k_L(N)$ times in $\PP_N$.  The
iterations missed by $\PP_{N-1}$ are represented by the peeled
statements computed by function \textsc{PeelAllLoops}.  It is natural
to expect the difference program $\partial \PP_N$ to contain the
peeled statements, perhaps with some adaptations, if $\PP_{N-1};
\partial \PP_N$ is to have the same effect as $\PP_N$ on all relevant
scalar variables and arrays.  However, statements that are present in
both $\PP_N$ and $\PP_{N-1}$ may also differ in their semantics, and
therefore require ``rectification'' in $\partial \PP_N$.  For example,
statements like {\tt x = N;} and {\tt if (x > N)} in $\PP_N$ become
{\tt x = N-1;} and {\tt if (x > N-1)} respectively, in $\PP_{N-1}$.
Clearly, the corresponding statements in $\PP_N$ and $\PP_{N-1}$ in
the above examples have different semantics.  We say that such
statements, though present in both $\PP_{N-1}$ and $\PP_N$, are
``affected'' by the parameter $N$, and potentially need to be
``rectified'' in $\partial \PP_N$.  Our goal in this subsection is to
identify all relevant scalar variables/arrays that are potentially
affected in this sense, i.e. they are updated by versions of the same
statement in $\PP_N$ and $\PP_{N-1}$ but can potentially result in
different values being assigned due to the change in the parameter
$N$.  We use the data dependence information computed in the previous
subsection to identify such variables and arrays, which we also call
\emph{affected variables/arrays}.  Once these variables/arrays are
identified, we can proceed to generate the difference program
$\partial \PP_N$ for effecting any rectification that may be needed.

\input{algo-compute-affected1}  

Function \textsc{ComputeAffected}, shown in
Algorithm~\ref{alg:pdg-affected}, computes the set of affected scalar
variables/arrays of a peeled program $\PP_N^p$, represented by its CFG
$(\mathit{Locs}^p, CE^p, \mu^p)$.  Besides the CFG, the function also
takes as input the set of CFG nodes corresponding to peeled
statements, denoted $\peelnodes$.  Recall that such a set is obtained
when function \textsc{PeelAllLoops} is invoked.
\textsc{ComputeAffected} starts by constructing the data dependence
graph using function \textsc{ComputeDDG}
(line~\ref{affected:line:pdg}).  The set of data dependence edges thus
obtained is represented by $DE^p$.  We use $\affectedvars$ to denote
the set of affected variables and arrays of $\PP_N$, and initialize it
to the empty set in line~\ref{affected:line:init}.  We also maintain a
list of nodes $n$ in $\worklist$ such that the semantics of the
program statement in $\mu(n)$ is potentially affected (directly or
indirectly) by $N$.  This worklist is initialized in
line~\ref{affected:line:init-wl} with all non-peeled nodes $n$ that
either (i) have $N$ in $\mathit{uses}(n)$, or (ii) are potentially
data dependent on a peeled node, i.e.  $\exists n'.\; n' \in
\peelnodes$ and $(n', n) \in DE^p$.  The exclusion of peeled nodes
from the worklist is justified by the observation that these
statements are present in $\PP^p_N$ but not in $\PP_{N-1}$.
Therefore, these must necessarily appear (possibly with modifications)
in $\partial \PP_N$, and no additional analysis is needed to identify
these statements or variables/arrays updated by them. We also keep
track of all non-peeled nodes that have been processed so far in the
set $\processednodes$, initialized in
line~\ref{affected:line:processednodesinit}.

The loop in
lines~\ref{affected:line:start-while}--\ref{affected:line:end-while}
iterates over the worklist, processing one node at a time to identify
affected scalar variables and arrays.  We remove the node $n$ at the
head of the worklist and add it to $\processednodes$ in
lines~\ref{affected:line:removen}--\ref{affected:line:addntoprocessed}.
If $\mu^p(n)$ is an assignment statement, we conservatively consider
the scalar variable or array updated at $n$ to be potentially affected
(marked in line~\ref{affected:line:add-defs}).  We also add all as-yet
unprocessed nodes $n'$ that have a data dependence on $n$ to the
worklist in line~\ref{affected:line:updateworklist1}.  This accounts
for nodes that are potentially affected because they use a value that
is generated at node $n$.  If node $n$ corresponds to a conditional
branch statement, we conservatively consider all non-peeled nodes $n'$
that are reachable from $n$ in the CFG of $\PP_N^p$ as potentially
affected by $N$.  If such a node $n'$ hasn't been processed yet, we
add it to the worklist in line~\ref{affected:line:updateworklist2}.
Finally, if $n$ corresponds to a loop-head, we skip the identification
of affected variables from $n$
(line~\ref{affected:line:skip-loop-heads}).  This is justified since
the special form of loops allowed by the grammar in
Fig.~\ref{fig:grammar} permits only loop counter variables to be
updated in a loop-head, and loop counter variables are not relevant
for the post-conditions we wish to prove.  The overall set of affected
scalar variables/arrays is iteratively computed until there are no
nodes left in the worklist to process.  Since the CFG of $\PP_N^p$ has
only a finite number of nodes, and since no node is processed more
than once (thanks to the book-keeping done using the set
$\processednodes$), the loop in
lines~\ref{affected:line:start-while}--\ref{affected:line:end-while}
is guaranteed to terminate.

\begin{lemma}\label{lemma:add-worklist}
  Let $\PP_N^p$ be a peeled program fed as input to the function
  \textsc{ComputeAffected}.  Let $n$ be a node in $\PP_N^p$ such that
  some scalar variable/array in $\mathit{uses}(n)$ is transitively
  data/control dependent on $N$ or on a peeled node in $\PP_N^p$.
  Then $n$ is added to $\worklist$ during the execution of function
  \textsc{ComputeAffected}.
\end{lemma}
\begin{proof}
  A transitive data/control dependence as referred to in the lemma can
  be represented by a sequence of nodes $n_{i_1}, n_{i_2}, \ldots,
  n_{i_k} (= n)$, where either (i) $k=1$ and $N \in \mathit{uses}(n)$
  or (ii) $k>1$ and $\mathit{uses}(n_{i_j})$ is data/control dependent
  on $n_{i_{j-1}}$, for $2 \leq j \leq k$.  By Lemma
  \ref{lem:compute-ddg-correct}, in case (ii), there is an edge from
  $n_{i_j-1}$ to $n_{i_j}$, for $2 \leq j \leq k$, in the data
  dependence graph computed at line \ref{affected:line:pdg} of
  function \textsc{ComputeAffected}.  We now prove the claim by
  induction on $k$.

  We consider two base cases of the induction.  If $k=1$, we know that
  $N \in \mathit{uses}(n)$.  Hence, $n$ is added to $\worklist$ in
  line \ref{affected:line:init-wl} of the function
  \textsc{ComputeAffected}.  If $k=2$, either $n_{i_1}$ is a peeled
  node or $N \in \mathit{uses}(n_{i_1})$.  In the former case,
  $n=n_{i_2}$ is added to $\worklist$ in line
  \ref{affected:line:init-wl} of \textsc{ComputeAffected}.  Otherwise,
  $n_{i_1}$ is added to $\worklist$ in line
  \ref{affected:line:init-wl} and must be removed from $\worklist$ in
  a later iteration before function \textsc{ComputeAffected}
  terminates.  In the iteration in which $n_{i_1}$ is removed from
  $\worklist$, the node $n=n_{i_2}$ is added to $\worklist$ either due
  to data dependence (line \ref{affected:line:updateworklist1} of
  \textsc{ComputeAffected}) or control dependence (line
  \ref{affected:line:updateworklist2} of \textsc{ComputeAffected}) of
  $\mathit{uses}(n)$ on $n_{i_1}$.

  We next hypothesize that, for every transitive data/control
  dependence on $N$ or on a peeled node, represented by the sequence
  of nodes $(n_{i_1}, n_{i_2}, \ldots, n_{i_j})$, where $2 \leq j \leq
  k-1$, the node $n_{i_j}$ is added to $\worklist$ during the
  execution of \textsc{ComputedAffected}.

  For the inductive step, consider a transitive data/control
  dependence on $N$ or on a peeled node, represented by the sequence
  of nodes $(n_{i_1}, n_{i_2}, \ldots, n_{i_k})$, where $n_{i_k} = n$.
  This implies that $\mathit{uses}(n_{i_{k-1}})$ is also data/control
  dependence on $N$ or on a peeled node.  Now by the inductive
  hypothesis, $n_{i_{k-1}}$ must be added to $\worklist$ in some
  iteration of the loop in lines \ref{affected:line:start-while} --
  \ref{affected:line:end-while} of \textsc{ComputedAffected}.  In the
  iteration in which $n_{i_{k-1}}$ is removed from $\worklist$, the
  node $n=n_{i_k}$ is added to $\worklist$ either due to data
  dependence (line \ref{affected:line:updateworklist1} of
  \textsc{ComputeAffected}) or control dependence (line
  \ref{affected:line:updateworklist2} of \textsc{ComputeAffected}) of
  $\mathit{uses}(n_{i_k})$ on $n_{i_{k-1}}$.
\end{proof}

In order to study additional properties of function
\textsc{ComputeAffected}, we need to introduce some additional
notation.  Given a peeled program $\PP_N^p$ generated by
\textsc{PeelAllLoops}, let $\PPP{N}{N-1}$ denote the program obtained
by removing the peels of all loops from $\PP_N^p$.  Clearly,
$\PPP{N}{N-1}$ is identical to the un-peeled program $\PP_N$ (fed as
input to \textsc{PeelAllLoops}), but with all instances of $N$ in
upper bound expressions of loops replaced by $N-1$.  The program
$\PPP{N}{N-1}$ is also closely related, but not identical, to
$\PP_{N-1}$.  Indeed, $\PP_{N-1}$ has all occurrences of $N$ (not just
those appearing in upper bound expressions of loops) in $\PP_N$
replaced by $N-1$, whereas only upper bound expressions of loops are
modified to obtain $\PPP{N}{N-1}$.  As an example, the loop in
Fig. \ref{fig:ex-loop}(a) in the program $\PP_N$ transforms to the
loop in Fig. \ref{fig:ex-loop}(b) in the program $\PPP{N}{N-1}$ and to
the loop in Fig. \ref{fig:ex-loop}(c) in the program $\PP_{N-1}$.

\begin{figure}[h]
\begin{alltt}
  for (l=0; l<N; l=l+1) \{
    if (x < N) \{ x = x + N; \}
  \}
\end{alltt}
\begin{center}(a)\end{center}
\begin{alltt}
  for (l=0; l<N-1; l=l+1) \{
    if (x < N) \{ x = x + N; \}
  \}
\end{alltt}
\begin{center}(b)\end{center}
\begin{alltt}
  for (l=0; l<N-1; l=l+1) \{
    if (x < N-1) \{ x = x + N-1; \}
  \}
\end{alltt}
\begin{center}(c)\end{center}
\caption{A loop in (a) $\PP_N$, (b) $\PPP{N}{N-1}$, and (c) $\PP_{N-1}$}
\label{fig:ex-loop}
\end{figure}

Since there is a bijection between the nodes in the CFGs of $\PP_N$
and $\PP_{N-1}$, and similarly between the nodes in the CFGs of
$\PP_{N}$ and $\PPP{N}{N-1}$, there exists a bijection between the
nodes in the CFGs of $\PP_{N-1}$ and $\PPP{N}{N-1}$ as well.  It is
also easy to see that since programs generated by the grammar in
Fig.~\ref{fig:grammar} only allow loop bound expressions that depend
on constants and $N$, if $L$ and $L'$ are corresponding loops in
$\PP_{N-1}$ and $\PPP{N}{N-1}$ respectively, then both $L$ and $L'$
iterate exactly the same number, i.e. $k_L(N-1)$, times.

Let $n_0, n_1, n_2, \ldots, n_t$ denote nodes in the CFG of
$\PP_{N-1}$, and let $n_0', n_1', n_2', \ldots, n_t'$ denote the
corresponding nodes (per the bijection) in the CFG of $\PPP{N}{N-1}$.
For notational convenience, we let $n_0$ and $n_t$ be the start and
end nodes respectively of the CFG of $\PP_{N-1}$, and similarly for
$n_0'$ and $n_t'$.  Let $\sigma$ denote an arbitrary initial state,
i.e. valuation of all scalar variables and arrays, from which we wish
to start executing $\PP_{N-1}$ and $\PPP{N}{N-1}$.  Since programs
generated by the grammar in Fig.~\ref{fig:grammar} are deterministic,
there is exactly one control flow path, say $\pi$, in the CFG of
$\PP_{N-1}$ that corresponds to the execution of $\PP_{N-1}$ starting
from $\sigma$. A similar argument holds for $\PPP{N}{N-1}$, and let
$\pi'$ be the corresponding path in its CFG.  In the following
discussion, we use $n_{i_j}$ to denote the $j^{th}$ node starting from
$n_{i_0}$ in $\pi$, where $i_0 = 0$ and $0 \le j < |\pi|$.  The
interpretation of $n_{i_j}'$ in the context of $\pi'$ is analogous.

\begin{lemma}\label{lem:helper1}
  Let $\pi$ (resp. $\pi'$) be the path in the CFG of $\PP_{N-1}$
  (resp.  $\PPP{N}{N-1}$) that corresponds to the execution of
  $\PP_{N-1}$ (resp.  $\PPP{N}{N-1}$) starting from the state
  $\sigma$.  Let $\widehat{\pi}: (n_{i_0}, n_{i_1}, \ldots, n_{i_j})$
  be a prefix of the path $\pi$.  Suppose, upon termination of
  \textsc{ComputeAffected}, no conditional branch node in
  $\widehat{\pi}$ is present in $\processednodes$. Then, $(n_{i_0}',
  n_{i_1}', \ldots, n_{i_j}')$ must be a prefix of $\pi'$.
\end{lemma}
\begin{proof}
Recall that the CFGs of $\PP_{N-1}$ and $\PPP{N}{N-1}$ are identical
(including all loop bounds) except possibly for the usage of $N$ or an
expression involving $N$ in some conditional branches and/or
assignment statements.  Note that by definition, none of these CFGs
contain any peeled nodes.  Since $\pi$ and $\pi'$ start at
corresponding nodes $n_0$ and $n_0'$ in the respective CFGs, every
subsequent node $n_{i_j}$ in $\pi$ must be matched by the
corresponding node $n_{i_j}'$ in $\pi'$, until a branch node is
encountered along one of the paths and the branch condition
potentially depends on $N$.  To prove the lemma, it therefore suffices
to show that no conditional branch node $n_{i_k}$ in $\widehat{\pi}$
has any transitive data/control dependence on $N$.

Let $n_{i_k}$ be a conditional branch node in $\widehat{\pi}$.  By
Lemma \ref{lemma:add-worklist}, if any scalar variable/array in
$\mathit{uses}(n_{i_k})$ is transitively data/control dependent on
$N$, then $n_{i_k}$ must be added to $\worklist$ at some point during
the execution of function \textsc{ComputeAffected}.  Consequently,
$n_{i_k}$ must also be removed from $\worklist$
(line~\ref{affected:line:removen}) and added to $\processednodes$
(line~\ref{affected:line:addntoprocessed}) before
\textsc{ComputeAffected} terminates. However, this violates the
premise of the claim, i.e. $n_{i_k}$ is not present in
$\processednodes$ on termination of \textsc{ComputeAffected}.
Therefore, no scalar variable/array in $\mathit{uses}(n_{i_k})$ can be
transitively data/control dependent on $N$.  This completes the proof
of the lemma.
\end{proof}

\begin{lemma}\label{lemma:affected-correct}
  Let $\PP_N^p$ be a peeled program fed as input to the function
  \textsc{ComputeAffected}.  Let $vA$ be a scalar variable/array that
  is absent in $\affectedvars$ when \textsc{ComputeAffected}
  terminates.  If $\PP_{N-1}$ and $\PPP{N}{N-1}$ are executed starting
  from the same state $\sigma$, then $vA$ has the same value on
  termination of both programs.
\end{lemma}
\begin{proof}
  We prove the lemma by contradiction.  If possible, let $\sigma$ be a
  state (i.e, valuation of variables and arrays) such that $vA$ has
  different values on termination of $\PP_{N-1}$ and $\PPP{N}{N-1}$,
  when both programs are executed starting from $\sigma$.  As before,
  we use $\pi$ and $\pi'$ to denote the paths in the CFGs of
  $\PP_{N-1}$ and $\PPP{N}{N-1}$ respectively, that correspond to the
  execution of the respective programs starting from $\sigma$.  Note
  that by definition, none of these CFGs contain any peeled nodes.  We
  consider the following cases.
  \begin{itemize}
  \item If none of $\pi$ and $\pi'$ updates $vA$, the value of $vA$ at
    the end of execution of the two programs is the same as the value
    it had in $\sigma$.  Clearly, the lemma holds in this case.
  \item Suppose node $n_{i_j}$ in $\pi$ updates $vA$.  We define
    $\Branch{\pi}{n_{i_j}}$ to be the set of all nodes $n_{i_k}$ in
    the prefix of $\pi$ ending at $n_{i_j}$ such that $n_{i_k}$
    corresponds to a conditional branch statement.  Similarly,
    $\Dep{\pi}{n_{i_j}}$ is defined to be the set of all nodes
    $n_{i_k}$ in the same prefix of $\pi$ such that there is a path
    through data dependency edges in $DE^p$ from node $n_{i_k}$ to
    either $n_{i_j}$ or to one of the nodes in
    $\Branch{\pi}{n_{i_j}}$.

    If possible, let $n_{i_k}$ be a node in $\Branch{\pi}{n_{i_j}}$
    such that the branch condition in $\mu^p(n_{i_k})$ has a (possibly
    transitive) data dependence on $N$.  By Lemma
    \ref{lemma:add-worklist}, $n_{i_k}$ must be added to $\worklist$
    sometime during the execution of \textsc{ComputeAffected}.  Since
    $n_{i_j}$ is reachable from $n_{i_k}$ along $\pi$, this further
    implies that $n_{i_j}$ must be added to $\worklist$, and
    subsequently $vA \in \mathit{def}(n_{i_j})$ must be added to
    $\affectedvars$ during the execution of \textsc{ComputeAffected}.
    However, we know that $vA$ is not present in $\affectedvars$ on
    termination of \textsc{ComputeAffected}.  Therefore, no branch
    condition in any node $n_{i_k}$ in $\Branch{\pi}{n_{i_j}}$ can be
    transitively data dependent on $N$.  It follows that no node in
    $\Branch{\pi}{n_{i_j}}$ can be added to $\worklist$ during the
    execution of \textsc{ComputeAffected}.  Hence, none of them can be
    present in $\processednodes$ on termination of
    \textsc{ComputeAffected}. By Lemma~\ref{lem:helper1}, it now
    follows that if $(n_{i_0}, n_{i_1}, \ldots, n_{i_j})$ is a prefix
    of $\pi$, then $(n_{i_0}', n_{i_1}', \ldots, n_{i_j}')$ must be a
    prefix of $\pi'$.

    Since the scalar variable/array $vA$ is updated at $n_{i_j}$ in
    $\pi$, the statement at $n_{i_j}'$ in $\pi'$ must also update
    $vA$.  Therefore, (i) $vA$ is updated at $n_{i_j}'$ in $\pi'$, and
    (ii) for every node $n_{i_k}$ in $\Dep{\pi}{n_{i_j}} \cup
    \Branch{\pi}{n_{i_j}}$, the corresponding node $n_{i_k}'$ is
    present in $\Dep{\pi'}{n_{i_j}'} \cup \Branch{\pi'}{n_{i_j}'}$.

    Finally, we argue that no node in $\Dep{\pi}{n_{i_j}}$ can be
    transitively data dependent on $N$.  Indeed, if this was not the
    case, by Lemma \ref{lemma:add-worklist}, $n_{i_j}$ would be added
    to $\worklist$ during the execution of \textsc{ComputeAffected},
    and hence $vA$ would be added to $\affectedvars$.  However, this
    violates the premise that $vA$ is absent in $\affectedvars$.
    Combining this with the result obtained above, we find that as far
    as path $\pi$ is concerned, no node in $\Dep{\pi}{n_{i_j}} \cup
    \Branch{\pi}{n_{i_j}}$ is transitively data dependent on $N$.
    Since modifications, if any, in statements at corresponding nodes
    of $\PP_{N-1}$ and $\PPP{N}{N-1}$ only involve replacing $N-1$ by
    $N$, such modifications preserve the dependence of every node on
    $N$.  Therefore, no node in $\Dep{\pi'}{n_{i_j}'} \cup
    \Branch{\pi'}{n_{i_j}'}$ transitively depends on $N$.  This
    implies that the statements labeling nodes in $\Dep{\pi}{n_{i_j}}
    \cup \Branch{\pi}{n_{i_j}}$ in $\pi$ are identical to the
    statements labeling the corresponding nodes in
    $\Dep{\pi'}{n_{i_j}'} \cup \Branch{\pi'}{n_{i_j}'}$.

    Since both $\PP_{N-1}$ and $\PPP{N}{N-1}$ start from the same
    state $\sigma$, the values of $vA$ computed by $\PP_{N-1}$ after
    executing the sequence of statements corresponding to $(n_{i_0},
    n_{i_1}, \ldots, n_{i_j})$ must therefore be identical to that
    computed by $\PPP{N}{N-1}$ after executing the sequence of
    statements corresponding to $(n_{i_0}', n_{i_1}', \ldots,
    n_{i_j}')$. This proves the lemma.
  \item The case when node $n_{i_j}'$ in $\pi'$ updates $vA$ is
    analogous to the above case.
  \end{itemize}
\end{proof}


For a variable/array $vA$ that is not identified as affected, $vA$
cannot be in the $\mathit{def}$ set of a non-peeled node that either
(i) has a transitive data dependence on $N$ or on a peeled node, or
(ii) has a control flow path from a branch node that, in turn, has a
transitive data dependence on $N$ or on a peeled node.  The following
lemma formalizes this property.

\begin{lemma}\label{lemma:not-affected}
Let $\PP_N^p$ be a peeled program fed as input to the function
\textsc{ComputeAffected}.  Let $vA$ be a scalar variable/array that is
absent in $\affectedvars$ after \textsc{ComputeAffected} terminates.
Then, $vA$ $\not \in$ $\mathit{def}(n)$ for every non-peeled node $n$
in $\PP_N^p$ such that some scalar variable/array in
$\mathit{uses}(n)$ is transitively data/control dependent on $N$ or on
a peeled node in $\PP_N^p$.
\end{lemma}
\begin{proof}
Consider an arbitrary non-peeled node $n$ in $\PP_N^p$ such that some
scalar variable/array in $\mathit{uses}(n)$ has a transitive
data/control dependence on $N$ or on a peeled node.  Then, by Lemma
\ref{lemma:add-worklist}, $n$ must be added to $\worklist$ sometime
during the execution of \textsc{ComputeAffected}.  Consequently, $n$
must also be removed from $\worklist$ (line
\ref{affected:line:removen} of \textsc{ComputeAffected}).  If $n$ is
an assignment node, then $\mathit{def}(n)$ is added to the set
$\affectedvars$ (line \ref{affected:line:add-defs} of
\textsc{ComputeAffected}).  But since $vA$ is absent in
$\affectedvars$, it follows that $vA$ $\not \in$ $\mathit{def}(n)$.
\end{proof}

For clarity of exposition, we will henceforth refer to the property
formalized in Lemma \ref{lemma:not-affected} as the
\emph{``not-affected''} property in our arguments.

\begin{example}
Consider the peeled program in Fig.~\ref{fig:ss-peel} along with its
DDG in Fig.~\ref{fig:draw-ddg}.  Recall that nodes in the DDG are
named such that node $n_i$ corresponds to the statement at line $i$ of
the peeled program.  The value of variable {\tt S} computed in the
peeled node (line $5$ in Fig.~\ref{fig:ss-peel}) of the first loop is
used to define array {\tt A1} in the body of the second loop (line
$7$).  Furthermore, the value of variable {\tt S} computed in line $5$
is used to initialize the value of {\tt S1} in line $10$.  Thus, the
algorithm initializes the worklist with the non-peeled nodes $n_7$ and
$n_{10}$ that are data dependent on the peeled node $n_{5}$.  Array
{\tt A1} and variable {\tt S1} updated at $n_{7}$ and $n_{10}$
respectively are marked as affected in line
\ref{affected:line:add-defs} of \textsc{ComputeAffected}.  We then add
non-peeled nodes that have a data dependence on $n_7$ and $n_{10}$ to
the worklist in line~\ref{affected:line:updateworklist1} of the
algorithm.  Since array {\tt A1} is used in node $n_{12}$ in the third
loop to define the variable {\tt S1}, $n_{12}$ is added to the
worklist. Subsequently, the variable {\tt S1} updated at $n_{12}$ is
marked as affected.  No further non-peeled nodes have any data
dependence on $n_{12}$ and the worklist therefore becomes empty.
Function \textsc{ComputeAffected} therefore terminates with $\{${\tt
  A1}, {\tt S1}$\}$ as the set of potentially affected
variables/arrays.  Note that variable {\tt S} updated in the first
loop (line 3 of Fig.~\ref{fig:ss-peel}) is not marked as affected
since its value doesn't transitively depend on $N$ or on any
variable/array updated in peeled nodes.
\end{example}

%% file: algo-compute-affected1.tex
\begin{algorithm*}[!t]
  \caption{\textsc{ComputeAffected}\big($(\mathit{Locs}^p, CE^p, \mu^p)$: peeled program $\PP_N^p$, $\peelnodes$: peeled statements\big)}
  \label{alg:pdg-affected}
  \begin{algorithmic}[1]
    \State $(\mathit{Locs}^p, DE^p, \mu^p)$ := \textsc{ComputeDDG}($(\mathit{Locs}^p, CE^p, \mu^p)$, $\peelnodes$);  \label{affected:line:pdg}
    \State $\affectedvars$ := $\emptyset$;  \Comment{Initialize $\affectedvars$}  \label{affected:line:init}
    \infocomment{Initialize $\worklist$ with non-peeled nodes of CFG that either use $N$ directly or have data dependence on peeled nodes} 
    \State $\worklist$ := $\{n \mid n \in \mathit{Locs}^p\setminus \peelnodes, N \in \mathit{uses}(n)$ or $\exists n'.\;n' \in \peelnodes \wedge (n',n) \in DE^p\}$; \label{affected:line:init-wl}
    \State $\processednodes$ := $\emptyset$;\label{affected:line:processednodesinit}
    \While{$\worklist$ is not empty}  \label{affected:line:start-while}
      \State Remove a node $n$ from the head of $\worklist$;\label{affected:line:removen}
      \State $\processednodes$ := $\processednodes \cup \{n\}$;\label{affected:line:addntoprocessed}
      \If{$\mu(n)$ is an assignment statement}\label{affected:line:start-case-analysis}
        \State $\affectedvars$ := $\affectedvars \cup \mathit{def}(n)$; \Comment{$\mathit{def}(n)$ potentially affected by $N$}\label{affected:line:add-defs}
        \For{all nodes $n'$ s.t. $n' \in \mathit{Locs}^p\setminus\peelnodes$, $n' \not\in \processednodes$ and $(n,n') \in DE^p$}
          \State $\worklist$ := \textsc{AppendToList}($\worklist$, $n'$);\label{affected:line:updateworklist1}
        \EndFor    
      \ElsIf{$\mu(n)$ is a branch condition}
        \For{all nodes $n'$ s.t. $n' \in \mathit{Locs}^p\setminus\peelnodes$, $n' \not\in \processednodes$ and $n \stackrel{\mathit{Locs}^p}{\rightsquigarrow} n'$}
          \State $\worklist$ := \textsc{AppendToList}($\worklist$, $n'$);\label{affected:line:updateworklist2}
        \EndFor    
      \Else ~~\textbf{continue}; \Comment{$n$ is a loop-head; don't do anything for loop-heads} \label{affected:line:skip-loop-heads}
      \EndIf \label{affected:line:end-case-analysis}
    \EndWhile \Comment{end of while loop processing elements of $\worklist$} \label{affected:line:end-while}
   \State \Return $\affectedvars$;  \label{affected:line:return}
  \end{algorithmic}
\end{algorithm*}
 

%% file: difference.tex
\subsection{Generating the Difference Program $\mathbf{\partial \PP_N}$}
\label{sec:diff-prog}
We now have most of the ingredients to generate $\partial \PP_N$ from
$\PP_N$ such that $\PP_N$ and $\PP_{N-1}; \partial \PP_N$ have the
same effect on scalar variables and arrays of interest.  For
notational convenience, in the remainder of this subsection, we use
$\PP_N$ to denote the renamed version of a given program, and
$\PP_N^p$ to denote the peeled version of the renamed program.

Generating the difference program ${\partial \PP_N}$ from the renamed
and peeled program using the set of affected variables computed
previously is still a daunting task.  In order to help the reader
better visualize the computation of difference program as well as to
simplify the proof of correctness, we present steps involved in the
computation of $\partial \PP_N$ as a sequence of simple program
transformations.  Fig. \ref{fig:trans-seq} presents a high level
overview of this sequence of transformations.  We start with a peeled
program $\PP_N^p$.  We first canonicalize it to a program $\TRPP_N^p$
that consists of a sequence of statements of a specific form
(explained in Sect. \ref{sec:transform}).  The statements in
$\TRPP_N^p$ corresponding to the peels of loops in $\PP_N^p$ are then
moved to the end of $\TRPP_N^p$ to obtain the program $\TRPP_N^o$.
The resulting program $\TRPP_N^o$ can be viewed as the program
$\PPP{N}{N-1}$ followed by the peels of all loops in $\PP_N^p$.  We
call the block of statements corresponding to the peels of all loops
as $\Peel{\PP_N}$.  Finally, if variables/arrays of interest are not
those identified as affected by \textsc{ComputeAffected},
$\PPP{N}{N-1}$ can be replaced with $\PP_{N-1}$.  This allows us to
obtain the difference program as $\Peel{\PP_N}$.  In the subsequent
sections, we present each transformation in detail, describe the
programs generated by them and prove that they preserve the overall
semantics of the program as far as the variables/arrays of interest
are concerned.  This allows us to show that for a large class of
programs wherein the variables/arrays of interest have specific
properties, it is possible to use just the peels of loops in $\PP_N^p$
as the difference program.  This simplifies the computation of
$\partial \PP_N$ significantly.


\input{fig-trans-seq}

We continue to use $vA$ to denote a scalar variable or an array
depending on the context.  If $vA$ is an array, the discussion below
applies to every individual element $vA[j]$, where $j$ is an index in
the allowed range of indices of array $vA$.  However, for notational
convenience, we use $vA$ (and not $vA[j]$) to refer to such an array
element in the lemmas below.  Note that this implies that the proof,
once completed, applies to an arbitrary element of the array $vA$, and
hence to the whole of $vA$.

\input{transform-pn}  

\input{reorder-trans-pn}  

\input{reverse-trans-pn}  

\subsubsection{Peels of Loops as the Difference Program}
\label{sec:diff-prog:peel}

Recall from Sect. \ref{sec:affected} that $\PPP{N}{N-1}$ is
effectively $\PP_N$ with the peels removed.  This is exactly what we
get by de-canonicalizing the part of $\TRPP_N^o$ that precedes the
guarded statements corresponding to peels.  If we call the
de-canonicalized version of the guarded statements corresponding to
peels as $\Peel{\PP_N^p}$ then $\PP_N^o$ can be written as
$\PPP{N}{N-1};\Peel{\PP_N^p}$.

It turns out that $\Peel{\PP_N^p}$ can be constructed directly from
$\PP_N$ without having to go through canonicalization, reordering and
de-canonicalization.  We now describe how to do this.  Recall from
Sect. \ref{sec:peeling} that $\peelnodes$ denotes the set of peeled
nodes in $\PP_N^p$.  Let $\condnodes$ be the set of all non-peeled
conditional branch nodes $b$ such that there is a peeled node $n$
within the scope of the branch $b$.  In other words, if $d$ denotes
the immediate post-dominator of $b$, there is a path from $b$ to $d$
that passes through $n$.  We define $\mathit{Locs}$ to be the set
$\peelnodes \cup \condnodes$.  Only these nodes in the CFG of
$\PP_N^p$ are relevant for the construction of $\Peel{\PP_N^p}$.
Therefore, we construct $\Peel{\PP_N^p}$ by replacing the labels of
all other nodes in the CFG of $\PP_N^p$ by {\Skip}.  Recall that
{\Skip} is a syntactic shorthand for \texttt{x = x;} for a variable
\texttt{x}, as discussed before.  Since a sequence of {\Skip}
statements can be collapsed without changing the program semantics, we
use $\Peel{\PP_N^p}$ to denote the program obtained after this
optimization.

\begin{example}
Consider the peeled program $\PP_N^p$ shown in
Tab. \ref{tab:ex-trans-pn}(a).  The program has a non-peeled
conditional branch statement on line $2$.  The peeled statements on
lines $6$ and $13$ are within the scope of the conditional branch
statement on line $2$.  Thus, $\condnodes$ = \{ $2$ \} and
$\mathit{Locs}$ = \{ $2b$, $5$, $6b$, $9$, $10b$, $12$ \}.  The
program $\Peel{\PP_N^p}$ is the program fragment consisting of the
nodes in the set $\mathit{Locs}$ in Tab. \ref{tab:ex-trans-pn}(d).
Notice that the non-peeled conditional branch node in $\PP_N^p$ (on
line $2$) that has the peeled nodes within its scope is retained in
$\Peel{\PP_N^p}$ along with the peels of loops.
\end{example}

\begin{lemma} \label{lemma:diff-program-without-affected-vars}
  Let $\PP_N^p$ be a peeled program and let $vA$ be a scalar
  variable/array in $\PP_N^p$ that is absent in $\affectedvars$ when
  \textsc{ComputeAffected} is executed on $\PP_N^p$.  If $\PP_N^o$ and
  $\PP_{N-1};\Peel{\PP_N^p}$ are executed from the same state
  $\sigma$, then $vA$ has the same value on termination of both
  programs.
\end{lemma}
\begin{proof}
  We break the proof in two parts.  We first show that if $\PP_N^o$
  and $\PPP{N}{N-1};\Peel{\PP_N^p}$ are executed starting from the
  same state $\sigma$, then $vA$ has the same value on termination of
  both programs.  This follows easily from Lemmas
  \ref{lemma:transform-pn-correct},
  \ref{lemma:sound-reordering-wo-affected} and \ref{lemma:reverse}.

  Next we show that if $\PPP{N}{N-1};\Peel{\PP_N^p}$ and
  $\PP_{N-1};\Peel{\PP_N^p}$ are executed from the same state
  $\sigma$, then $vA$ has the same value on termination of both
  programs.  We prove this part by case analysis.

  Suppose the last update to $vA$ in $\PPP{N}{N-1};\Peel{\PP_N^p}$
  happens in a non-peeled node in $\PPP{N}{N-1}$. Then, the proof
  follows immediately from Lemma \ref{lemma:affected-correct}.

  Suppose the last update to $vA$ in $\PPP{N}{N-1};\Peel{\PP_N^p}$
  happens in a peeled node $n$ in $\Peel{\PP_N^p}$.  Let $S$ denote
  the set of variables/arrays $vA'$ such that the updated value of
  $vA$ at node $n$ depends on the values of each $vA' \in S$ on
  termination of $\PPP{N}{N-1}$.  There are two sub-cases to consider.

  If no $vA' \in S$ is identified as affected by
  \textsc{ComputeAffected}, then by Lemma \ref{lemma:affected-correct}
  the value of every such $vA'$ is the same after termination of
  $\PPP{N}{N-1}$ and $\PP_{N-1}$.  This implies that the value of $vA$
  is also same after termination of $\PPP{N}{N-1};\Peel{\PP_N^p}$ and
  $\PP_{N-1};\Peel{\PP_N^p}$.

  Now consider the case where some $vA' \in S$ is identified as
  affected by \textsc{ComputeAffected}.  Let $L$ be a loop in
  $\PP_N^p$ from which the node $n$ is peeled.  From the construction
  of peeled nodes, we know that for every node $n$ in the peel of $L$
  there is a corresponding node $n'$ in the ``uncollapsed'' body of
  loop $L$ such that the $\mathit{def}$ and $\mathit{uses}$ sets of
  the two nodes $n$ and $n'$ coincide.  Since the update to $vA$ at
  node $n$ depends on $vA'$ that is identified as affected, the update
  to $vA$ at node $n'$ in loop $L$ must also depend on the affected
  variable/array $vA'$.  However, this would cause
  \textsc{ComputeAffected} to identify $vA$ as an affected variable.
  This leads to a contradiction since we know $vA$ is not affected.
  This completes the proof.
\end{proof}


Lemma \ref{lemma:diff-program-without-affected-vars} allows us to use
$\Peel{\PP_N^p}$ as the difference program $\partial \PP_N$ if none of
the scalar variables and arrays of interest are identified as affected
by \textsc{ComputeAffected}.  This holds true in the case where the
post-condition $\psi^r(N)$ does not refer to any affected
variable/array.  Note that using $\Peel{\PP_N^p}$ as the difference
program $\partial \PP_N$ works even if there are other
variables/arrays (not of interest) that are affected.  However, if
some of our variables/arrays of interest are indeed identified as
affected by \textsc{ComputeAffected}, we must include additional code
in the difference program that effectively ``rectifies'' the values of
affected variables as computed by $\PP_{N-1}$.  We elaborate on this
in the next subsection.

\subsubsection{Generalized Computation of Difference Programs}
\label{sec:diff-prog:affected}

\input{ex-ss-incorrect-diff}

Recall that $\Peel{\PP_N^p}$ was constructed by replacing some of the
nodes in the collapsed CFG of $\PP_N^p$ with {\Skip} and by
simplifying the resulting CFG.  As seen above, this suffices to serve
as the difference program $\partial \PP_N$ if none of the
variables/arrays of interest are identified as affected by function
\textsc{ComputeAffected}.  If, however, some variables/arrays of
interest are identified as affected, $\Peel{\PP_N^p}$ as computed
above may no longer serve as a correct difference program.  To see an
example of this, consider the peeled program in
Fig. \ref{fig:ss-peel}.  If we were to compute $\Peel{\PP_N^p}$ for
this program we would get the program fragment in lines $12 - 14$
shown in Fig. \ref{fig:ss-incorrect-diff}.  In this program, array
{\tt A1} and scalar variable {\tt S1} are identified as affected.
Notice that, the program $\PP_{N-1};\Peel{\PP_N}$ in
Fig. \ref{fig:ss-incorrect-diff} computes incorrect values of {\tt A1}
at line $6$ and {\tt S1} at line $14$.

Interestingly, even in cases like the above example, a correct
difference program can often be constructed by modifying the way in
which $\Peel{\PP_N^p}$ is constructed.  To prevent confusion, we do
not call the program resulting from this modified construction as
$\Peel{\PP_N^p}$.  Instead we call it $\partial \PP_N$.  The
modification referred to above concerns which statements are replaced
by {\Skip} and which are retained but possibly with a change.
Specifically, all assignment statements that update an affected
variable/array but were earlier (while constructing $\Peel{\PP_N^p}$)
replaced by {\Skip} are retained with a possibly changed expression in
the right hand side of the assignment.  Since no new nodes are added
to the CFG of $\PP_N$ in this way of constructing $\partial \PP_N$,
there is a natural injective mapping, say $\beta$, from the nodes in
the CFG of $\partial \PP_N$ to those in the CFG of $\PP_N$.



To understand how the right hand side expressions of assignments may
need to be changed when constructing $\partial \PP_N$, consider an
execution of each of $\PP_N$ and $\PP_{N-1};\partial \PP_N$ starting
from the same initial state $\sigma$.
\begin{definition}
For every node $n$ in the CFG of $\partial \PP_N$ and for every
variable/array element $vA$ we say that $vA$ has a \emph{rectified
  value} at $n$ if its value at $n$ matches the value of $vA$ at
$\beta(n)$.  Otherwise, we say that $vA$ has an \emph{unrectified
  value} at $n$.
\end{definition}

For every node $n$ in $\partial \PP_N$ that updates an affected
variable/array $vA$ of interest, we modify the right hand side of the
assignment (if necessary) such that the right hand side expression
evaluates to the rectified value of $vA$ at $n$.  This expression is
constructed in such a manner that it uses the unrectified value of
$vA$ at $n$ (if the assignment statement was replaced by a {\Skip}) in
its computation of the rectified value.  This construction allows us
to establish an important property of the resulting program $\partial
\PP_N$: every variable/array $vA$ of interest has its rectified value
at every node $n$ in $\partial \PP_N$.

We now elaborate on how we construct the modified right hand side
expression of an assignment statement at node $n$ in $\partial \PP_N$
that updates the affected variable/array $vA$.  We assume that we have
access to the rectified and unrectified values of all variables/arrays
$vA'$ used in the right hand side expression of the assignment
statement at node $\beta(n)$ in $\PP_N$.  The easiest way to do this
would be to construct the right hand side expression exclusively in
terms of the rectified values of $vA'$ at node $n$.  Note that, this
results in a difference program that is as complex as the original
program $\PP_N$.  This defeats our purpose, since full-program
induction can succeed only if $\partial \PP_N$ is ``simpler'' than
$\PP_N$.  Therefore, we do not use this naive method and present an
operator algebra to compute of the rectified value of $vA$ updated in
the assignment statement in terms of its unrectified value and the
``difference'' between the rectified and unrectified values of other
variables/arrays that have a data dependence to $vA$.

Let $\circ$ be a binary operator on a set $\smem$ that denotes the
domain of values of variables/arrays in $\PP_N$.  We say that $e$ is
the right identity element of $\circ$ if $v$ $\circ$ $e$ = $v$ and $e$
is the left identity element if $e$ $\circ$ $v$ = $v$ for each $v \in
\smem$.  We call $v^{-\circ}$ the right inverse element of $v$ under
$\circ$ if $v$ $\circ$ $v^{-\circ}$ = $e$ and we call it the left
inverse element if $v^{-\circ}$ $\circ$ $v$ = $e$ for each $v \in
\smem$.  We say that $\circ$ is an associative operator if $(u \circ
v) \circ w$ = $u \circ (v \circ w)$.  We say that $\circ$ is a
commutative operator if $u \circ v$ = $v \circ u$.  When the operator
$\circ$ is associative, $(u \circ v)^{-\circ}$ = $(v^{-\circ} \circ
u^{-\circ})$.


For the following lemmas, we assume that $\circ$ is an associative
operator, there exists a left identity element under $\circ$ in
$\smem$ and each element in $\smem$ has a right inverse under $\circ$.

\begin{lemma}
\label{lemma:diffcomp1}
Let $n$ be a node in $\partial \PP_N$ such that the statement at
$\beta(n)$ in $\PP_N$ is $w$ $:=$ $u$ $\circ$ $v$.  Suppose $w$ is an
affected variable/array of interest.  Let $w_{N-1}$, $u_{N-1}$,
$v_{N-1}$ denote the values of $w$, $u$ and $v$ at the end of
execution of $\PP_{N-1}$.  Let $w_N$, $u_N$, $v_N$ be the rectified
values of $w$, $u$ and $v$ at node $n$.  Then the rectified value of
$w$ is computed as $w_N$ $:=$ $w_{N-1}$ $\circ$ $((v_{N-1})^{-\circ}$
$\circ$ $((u_{N-1})^{-\circ}$ $\circ$ $u_N)$ $\circ$ $v_N)$.
\end{lemma}
\begin{proof}
We proceed as follows:
  \begin{enumerate}
    \item $w_N$ $:=$ $e$ $\circ$ $w_N$ 
    \item $w_N$ $:=$ $(w_{N-1}$ $\circ$ $(w_{N-1})^{-\circ})$ $\circ$ $w_N$ 
    \item $w_N$ $:=$ $w_{N-1}$ $\circ$ $((w_{N-1})^{-\circ}$ $\circ$ $w_N)$ 
    \item $w_N$ $:=$ $w_{N-1}$ $\circ$ $((u_{N-1}$ $\circ$ $v_{N-1})^{-\circ}$ $\circ$ $(u_N$ $\circ$ $v_N))$
    \item $w_N$ $:=$ $w_{N-1}$ $\circ$ $((v_{N-1})^{-\circ}$ $\circ$ $(u_{N-1})^{-\circ}$ $\circ$ $(u_N$ $\circ$ $v_N))$
    \item $w_N$ $:=$ $w_{N-1}$ $\circ$ $((v_{N-1})^{-\circ}$ $\circ$ $((u_{N-1})^{-\circ}$ $\circ$ $u_N)$ $\circ$ $v_N)$
\end{enumerate}
\end{proof}

Suppose $\circ$ is additionally a commutative operator.  Then the
following lemmas hold.

\begin{lemma}
\label{lemma:diffcomp2}
Under the assumptions stated in Lemma \ref{lemma:diffcomp1}, the
rectified value of $w$ is computed as $w_N$ := $w_{N-1}$ $\circ$
$(u_N$ $\circ$ $(u_{N-1})^{-\circ})$ $\circ$ $(v_N$ $\circ$
$(v_{N-1})^{-\circ})$.
\end{lemma}
\begin{proof}
We continue from the proof of Lemma~\ref{lemma:diffcomp1} and proceed as follows:
  \begin{enumerate}
  \setcounter{enumi}{5}
    \item $w_N$ $:=$ $w_{N-1}$ $\circ$ $((v_{N-1})^{-\circ}$ $\circ$ $((u_{N-1})^{-\circ}$ $\circ$ $u_N)$ $\circ$ $v_N)$
    \item $w_N$ $:=$ $w_{N-1}$ $\circ$ $(((u_{N-1})^{-\circ}$ $\circ$ $u_N)$ $\circ$ $(v_{N-1})^{-\circ}$ $\circ$ $v_N)$
    \item $w_N$ $:=$ $w_{N-1}$ $\circ$ $((u_{N-1})^{-\circ}$ $\circ$ $u_N)$ $\circ$ $((v_{N-1})^{-\circ}$ $\circ$ $v_N)$
    \item $w_N$ $:=$ $w_{N-1}$ $\circ$ $(u_N$ $\circ$ $(u_{N-1})^{-\circ})$ $\circ$ $(v_N$ $\circ$ $(v_{N-1})^{-\circ})$
\end{enumerate}
\end{proof}

To use the equation from Lemma~\ref{lemma:diffcomp2} for statements
with non-commutative operators such as $\{-, \div\}$ often used in
practice, we perform a simple transformation that allows us to use
commutative operators inplace of non-commutative ones.  As an example
of this transformation, consider the expressions $u - v$ and $u / v$.
We transform them into the expressions $u + (-v)$ and $u \times (1/v)$
respectively.  This allows us to use the equation in
Lemma~\ref{lemma:diffcomp2} when for every element $v \in \smem$, the
elements $-v$ and $1/v$ are also in $\smem$.

\begin{lemma}
\label{lemma:diffcomp3}
Let $n$ be a node in $\partial \PP_N$ such that the statement at
$\beta(n)$ in $\PP_N$ is $w$ $:=$ $w$ $\circ$ $v$.  Then the rectified
value of $w$ is computed as $w_N$ := $w_N$ $\circ$ $(v_N$ $\circ$
$(v_{N-1})^{-\circ})$ along with the presumption $w_N = w_{N-1}$.
\end{lemma}
\begin{proof}
Using the given presumption $w_N = w_{N-1}$, we have the definition
of the identity element as: $e$ $=$ $w_N \circ (w_N)^{-\circ}$ $=$
$w_N \circ (w_{N-1})^{-\circ}$. \\

\noindent
As the first step, we use the result from Lemma~\ref{lemma:diffcomp2}
and proceed as follows:
\begin{enumerate}
  \item $w_N$ $:=$ $w_{N-1}$ $\circ$ $(w_N$ $\circ$
  $(w_{N-1})^{-\circ})$ $\circ$ $(v_N$ $\circ$ $(v_{N-1})^{-\circ})$
  \item $w_N$ $:=$ $w_{N-1}$ $\circ$ $e$ $\circ$ $(v_N$ $\circ$ $(v_{N-1})^{-\circ})$
  \item $w_N$ $:=$ $w_{N-1}$ $\circ$ $(v_N$ $\circ$ $(v_{N-1})^{-\circ})$
  \item $w_N$ $:=$ $w_N$ $\circ$ $(v_N$ $\circ$ $(v_{N-1})^{-\circ})$
\end{enumerate}
\end{proof}

It is worth noting that the rectification described in Lemmas
\ref{lemma:diffcomp1}, \ref{lemma:diffcomp2} and \ref{lemma:diffcomp3}
applies not only when the set $\smem$ is integers, i.e. integers are
stored as array elements but even when the set consists of {\em
  matrices, vectors, and polynomials}.  When matrices are stored as
array elements, such arrays are called {\em tensors}.  These are
extensively used in machine learning algorithms.  Further, it applies
to interesting operators such as $\boldsymbol{\oplus}$,
$\boldsymbol{+}$ $mod$ $x$, $\boldsymbol{\times}$ $mod$ $y$ as well as
to other interesting algebraic structures.
It is also worth mentioning that for a restricted class of programs
our technique extends to computing the differences of programs that
manipulate heaps.

\input{algo-prog-diff.tex}  

The routine \textsc{ProgramDiff} presented in
Algorithm~\ref{alg:diff-generation} shows how the difference program
is computed.  In line~\ref{line:diff:peelloops}, we peel each loop in
the program $\PP_N$ and collect the list of peeled nodes using
Algorithm~\ref{alg:peelloops}.  We then compute the set of affected
variables using Algorithm~\ref{alg:pdg-affected}
(line~\ref{line:diff:affected}).  The difference program $\partial
\PP_N$ inherits the skeletal structure of the peeled program $\PP^p_N$
after peeling each loop (line~\ref{line:diff:skeletonpn}).  Next, we
collapse all nodes and edges in the body of each loop into a single
node identified with the loop-head in the CFG of $\PP^p_N$ using the
function \textsc{CollapseLoopBody} in line \ref{line:diff:collapse}.
The collapsed CFG of the resulting program $\partial \PP_N$ is a DAG
with finitely many paths.  We then initialize a worklist of CFG nodes
with $n_{start}$ in line \ref{line:diff:init-wl}.

The while loop in
lines~\ref{line:diff:begin-wl}--\ref{line:diff:end-wl} performs a
breadth-first top-down traversal over the DAG of $\partial \PP_N$
starting from the node $n_{start}$ and processes one node at a
time. We first remove a node $n$ from the worklist in line
\ref{line:diff:removeonenode}.  We store the nodes that are already
processed by our algorithm in $\processednodes$ (that is
initialized in line \ref{line:diff:processednodesinit}).  We add the
node $n$ removed from the worklist to $\processednodes$ in line
\ref{line:diff:addntoprocessed}.  Next, the loop in
lines~\ref{line:diff:copyst}--\ref{line:diff:copyend} appends each
successor $n'$ of $n$ to the worklist that is not already processed.
We use the routine \textsc{Succ} to obtain the list of successors of
node $n$.

If $n$ is a peeled node, then we retain it as is in the difference
program (line \ref{line:diff:skip-peel}).  Otherwise, we check if any
scalar variable/array used at node $n$ is affected at
line~\ref{line:diff:check-affected}.  We have defined the sub-routine
\textsc{HasAffectedVars} that checks if the scalar variable/array
defined at node $n$ is affected.

For nodes $n$ that refer to an affected variable/array, we do the
following.  We check if a node $n$ is a glue node that refers to an
affected variable/array in line \ref{line:diff:retain-affected-glue}
and retain such nodes as is in the difference program.  Otherwise, we
check if the node $n$ corresponds to a loop-head in line
\ref{line:diff:loop-head}.  We uncollapse the nodes corresponding to a
loop-head that represent the entire loop in line
\ref{line:diff:uncollapse-l}.  We assume that the sub-routine
\textsc{Nodes}($L$) returns the set of CFG nodes in loop $L$.  Next,
the loop in lines
\ref{line:diff:loopnodes-start}--\ref{line:diff:loopnodes-end}
iterates over all nodes $n'$ in the body of $L$ and process one node
at a time.  In line \ref{line:diff:check-affected2}, we check if the
variable/array updated at node $n'$ is affected using function
\textsc{HasAffectedVars}, and compute its rectified value in line
\ref{line:diff:stdiff-loop}, using the function \textsc{NodeDiff}.  If
the variable/array defined at $n'$ is not identified as affected, then
we remove from ${\partial \PP_N}$ nodes $n'$ that do not update an
affected variable/array using the routine \textsc{RemoveNode} in line
\ref{line:diff:removen1}.  For a non-peeled node $n$ that does not
correspond to a loop-head, we compute the rectified value of an
affected variable/array defined at $n$ in line \ref{line:diff:stdiff},
using the function \textsc{NodeDiff}.

For nodes $n$ that are not peeled nodes and do not update an
affected variable/array, we do the following.  We compute the set
$\condnodes$ of conditional branch nodes that have at least one peeled
node within its scope in line \ref{line:diff:condnodes}.  In line
\ref{line:diff:removen2}, we remove from ${\partial \PP_N}$ nodes $n$
(including collapsed loop nodes) that do not update an affected
variable/array and are not in the set $\condnodes$, using the routine
\textsc{RemoveNode}, as they do not need any rectification.

The sub-routine \textsc{NodeDiff} computes the statements that rectify
values of variables/arrays updated at a node.  It determines the type
of statement (assignment, aggregation or branch condition) at the
given node and acts accordingly.  For assignment statements, we
compute the rectified value as shown in Lemma~\ref{lemma:diffcomp2}
and for aggregating statements, we compute the rectified value as
shown in Lemma~\ref{lemma:diffcomp3}.  For the nodes representing a
conditional branch in $\partial \PP_{N}$, we determine if its
conditional expression evaluates to the same value in $\PP_N$ and
$\PP_{N-1}$.  If so, the conditional branch is retained as is in
$\partial \PP_N$.  Otherwise, currently our technique cannot compute
${\partial \PP_N}$ and we report a failure using the {\bf throw}
statement.

To explain the intuition behind the steps of
Algorithm~\ref{alg:diff-generation}, we use the convention that all
variables and arrays of $\PP_{N-1}$ have the suffix {\tt \_Nm1} (for
N-minus-1), while those of $\PP_N$ have the suffix {\tt \_N}.  This
allows us to express variables/arrays of $\PP_N$ in terms of the
corresponding variables/arrays of $\PP_{N-1}$ in a systematic way in
${\partial \PP_N}$, given that the intended composition is
$\PP_{N-1};{\partial \PP_N}$.

For assignment statements, we compute the rectified values as follows.
For every assignment statement of the form {\tt v = E;} in $L$, a
corresponding statement is generated in ${\partial \PP_N}$ that
expresses {\tt v\_N} in terms of {\tt v\_Nm1} and the difference (or
ratio) between versions of variables/arrays that appear as
sub-expressions in {\tt E} in $\PP_{N-1}$ and $\PP_N$.

While the implementation is currently restricted to simple arithmetic
operators ($+, -, \times, \div$), specifically for the ease of
implementation and its use in practice, as previously stated, our
rectification method is general and applies to several operators
beyond the ones mentioned here.  The following example shows the
computation of rectified values of variables/arrays updated in simple
program statements.

\begin{example}
The statement {\tt A\_N[i] = B\_N[i] + v\_N;} in $\PP_{N}$ gives rise
to the statement {\tt A\_N[i] = A\_Nm1[i] + (B\_N[i] + (- B\_Nm1[i]))
  + (v\_N + (- v\_Nm1));} in ${\partial \PP_N}$ that rectifies the
value of {\tt A\_N[i]}.  Similarly, the statement {\tt A\_N[i] =
  B\_N[i] * v\_N;} in $\PP_{N}$ gives rise to the statement {\tt
  A\_N[i] = A\_Nm1[i] * (B\_N[i] * (1/B\_Nm1[i])) * (v\_N *
  (1/v\_Nm1));} under the assumption {\tt B\_Nm1[i] * v\_Nm1} $\neq
0$.
\end{example}

The program ${\PP_N}$ may have statements that aggregate/accumulate
values in scalars.  This kind of statement requires special processing
when generating the difference program ${\partial \PP_N}$.  The next
example shows the computation of rectified values of variables/arrays
in statements that accumulate values in scalar variables.

\begin{example}
Consider the loop {\tt for(i=0; i<N; i++) \{ sum\_N = sum\_N +
  A\_N[i]; \} } in program $\PP_N$. The difference {\tt A\_N[i] + (-
  A\_Nm1[i])} is aggregated over all indices from $0$ through $N-2$.
In this case, the loop in ${\partial \PP_N}$ that rectifies the value
of {\tt sum\_N} has the following form: {\tt sum\_N = sum\_Nm1; for
  (i=0; i<N-1; i++) \{ sum\_N = sum\_N + (A\_N[i] + (- A\_Nm1[i]));
  \}}.  A similar aggregation for multiplicative ratios can also be
shown.
\end{example}



Conditional branch statements pose a considerable challenge to the
computation of difference programs.  A branch condition may evaluate
to different outcomes in $\PP_N$ and $\PP_{N-1}$, for the same value
of $N$.  When this happens, programs $\PP_N$ and $\PP_{N-1}$ execute
totally unrelated blocks of statements.  In such situations, it is
immensely difficult to rectify the values of variables/arrays computed
along the unrelated branches, and hence, our algorithm avoids doing
so.  Only when we can determine that the condition evaluates to the
same value in $\PP_N$ and $\PP_{N-1}$, we rectify values of
variables/arrays computed along the corresponding branches.  Next we
present examples with branch conditions to highlight this.

\begin{example}
Consider the conditional branch statement {\tt if(t3 == 0) } in line
$10$ of Fig. \ref{fig:ex}.  The branch condition evaluates to the same
value in $\PP_N$ and $\PP_{N-1}$ because the condition has no
dependence on $N$.  Therefore, the branch statement is used as is
during the computation of the difference program.  However, recall
that since the arrays accessed in the program are not affected, none
of the loops are retained in the difference program shown in
Fig. \ref{fig:ex-ind}.

Consider another conditional branch statement {\tt if(A[i] == N)} in
$\PP_N$.  The corresponding statement in $\PP_{N-1}$ is {\tt if(A[i]
  == N-1)}.  Clearly, the conditions in these statements do not
evaluate to the same value in $\PP_N$ and $\PP_{N-1}$.  Thus, our
algorithm flags a failure to compute the difference program and
terminates.
\end{example}

\begin{figure}[h]
\begin{alltt}
x = N;
y = N;
for(i=0; i<N; i++) \{
  if(x == y)
    A[i] = i;
\}
\end{alltt}
\begin{center}(a)\end{center}
\begin{alltt}
x = N-1;
y = N-1;
for(i=0; i<N-1; i++) \{
  if(x == y)
    A[i] = i;
\}
\end{alltt}\begin{center}(b)\end{center}
\caption{Example program (a) $\PP_N$ and (b) $\PP_{N-1}$}
\label{fig:ex-br-diff}
\end{figure}

There are programs where conditional branch statements with dependence
on $N$ evaluate to the same value. For example, consider the program
$\PP_N$ in Fig. \ref{fig:ex-br-diff}(a).  The program $\PP_{N-1}$ is
shown in Fig. \ref{fig:ex-br-diff}(b).  While the branch condition
(indirectly) depends on the value of $N$, it evaluates to the same
value in $\PP_N$ and $\PP_{N-1}$, since the amount of change in the
value of variables {\tt x} and {\tt y} used in the branch condition is
same.  Our algorithm can successfully compute the difference program
in such cases.

The restriction on branch conditions that use affected
variables/arrays can be further relaxed by handling the case when the
condition evaluates to $\true$ in $\PP_{N-1}$ and to $\false$ in
$\PP_N$ by restoring the values of variables/arrays to their values at
the predecessor of the branch node.  However, when a branch condition
evaluates to $\false$ in $\PP_{N-1}$ and to $\true$ in $\PP_N$, the
entire computation within the branch has to be performed again in
$\partial\PP_N$ instead of just executing the rectification code.  For
example, consider the branch statement, {\tt if(i < N) Loop;}.  If the
branch condition {\tt i $<$ N-1} in $\PP_{N-1}$ evaluates to $\false$,
then the condition {\tt i $<$ N} in $\PP_N$ definitely evaluates to
$\true$.  This will require the difference program to execute the
entire computation performed by the code fragment {\tt Loop;} and not
just the difference program corresponding to {\tt Loop;}.  This will
hamper the progress guarantees on the class of programs that our
technique can verify.  Hence, we currently avoid handling these cases
in the algorithms and consider them as a part of future work.

We now prove the soundness of the routine \textsc{ProgramDiff} from
Algorithm \ref{alg:diff-generation}.  For the following lemma, we
assume that $\partial \PP_N$ is the difference program computed when
function \textsc{ProgramDiff} is invoked on the renamed program
$\PP_N$.  Suppose both $\PP_N^p$ and $\PP_{N-1}; \partial \PP_N$ are
executed from the same initial state $\sigma$.  We assume $\pi:
(n_{0}, n_{1}, \ldots, n_{k})$ to be the path in the CFG of $\partial
\PP_N$ corresponding to the execution of the difference program from
the state obtained after $\PP_{N-1}$ has terminated (in $\PP_{N-1};
\partial \PP_N$).  We assume $\pi': (n_{0}', n_{1}', \ldots, n_{k}')$
to be the corresponding path in the CFG of the peeled program
$\PP^p_N$.

\begin{lemma} \label{lemma:rectified-use-available}
For every node $n_j$ in $\pi$, the rectified values of all
variables/array elements used at $n_j$ during the execution of
$\partial \PP_N$ are identical to the values of the same
variables/array elements at the corresponding node $n'_j$ during the
execution of $\PP^p_N$.
\end{lemma}

\begin{proof}
  If $vA$ is not identified as an affected variable/array by function
  \textsc{ComputeAffected}, the result follows from the proof of
  Lemma~\ref{lemma:diff-program-without-affected-vars} and the
  no-overwriting property of renaming.

  If $vA$ is identified as an affected variable, we induct on the
  length of $\pi$.  The only difference in this case is that we also
  need to consider the assignment statements modified by function
  \textsc{NodeDiff} at lines \ref{line:diff:stdiff-loop} and
  \ref{line:diff:stdiff} of function \textsc{ProgramDiff}.  Lemmas
  \ref{lemma:diffcomp1}, \ref{lemma:diffcomp2} and
  \ref{lemma:diffcomp3} guarantee the correctness of the rectified
  value of $vA$ computed by these additional statements, given the
  unrectified value of $vA$ and rectified and unrectified values of
  all variables and array elements used in the right hand side of the
  assignment.  By the inductive hypothesis, the rectified values of
  the latter set of variables and array elements are indeed available.
  By the no-overwriting property, the unrectified values of $vA$ and
  all other variables and arrays used in the right hand side of the
  assignment are also available.  Therefore, the correct rectified
  value of $vA$ is computed at each node in $\pi$.

  Finally, note that once a rectified value is generated at a non-glue
  node in the difference program, renaming ensures that it is not
  re-defined by subsequent statements in the difference program.
  Therefore, rectified values, once computed in the difference
  program, are available for use at subsequent nodes in the execution
  path.  Putting the above parts together completes the proof.
\end{proof}

\begin{theorem}
\label{lemma:diff-gen-sound}
$\partial \PP_N$ generated by \textsc{ProgramDiff} is such that, for
all $N > 1$, $\{\varphi(N)\} \;\PP_{N-1};\partial \PP_N \;
\{\psi(N)\}$ holds iff $\{\varphi(N)\} \;\PP_{N} \; \{\psi(N)\}$
holds.
\end{theorem}
\begin{proof}
  Lemma~\ref{lemma:alg-peelloops} guarantees that $\{\varphi(N)\}
  \;\PP_{N} \; \{\psi(N)\}$ holds iff $\{\varphi(N)\} \;\PP_{N}^p \;
  \{\psi(N)\}$ holds.  Furthermore,
  Lemma~\ref{lemma:rectified-use-available} ensures that for every
  state $\sigma$ satisfying $\varphi(N)$, if we execute $\PP_N^p$ and
  $\PP_{N-1}; \partial \PP_N$ starting from $\sigma$, then if the
  final state after termination of $\PP_N^p$ satisfies $\psi(N)$, so
  does the final state after termination of $\PP_{N-1}; \partial
  \PP_N$.  This proves the theorem.
\end{proof}

\begin{example}
We illustrate the difference computation performed by the routine
\textsc{ProgramDiff} in Algorithm~\ref{alg:diff-generation} on our
running example.  Fig.~\ref{fig:ss-diff} shows the difference program
obtained after executing the algorithm on the program in
Fig.~\ref{fig:ss-peel}.  Notice that while some program statements in
Fig.~\ref{fig:ss-diff} are syntactically similar to corresponding
statements in Fig.~\ref{fig:ss-peel}, the additional statements
(e.g. at lines 3, 6 and 8) in Fig.~\ref{fig:ss-diff} have no syntactic
counterpart in Fig.~\ref{fig:ss-peel}.  For the first loop, since
variable {\tt S} is not affected, only the peeled iteration is
retained.  Since {\tt A1} and {\tt S1} both are affected, the
statements in the second and the third loop are replaced with
statements that rectify their values, along with inserting the peeled
statements for both loops.  Initialization of variable {\tt S1} is
also replaced with the statement that rectifies its value since it
depends on the value of {\tt S} computed in a peel.
\end{example}

\input{ex-ss-diff}  

%% file: fig-trans-seq.tex
\begin{figure}[!t]
\centering

\resizebox{0.50\textwidth}{!}{

\tikzset{every picture/.style={line width=0.75pt}} 

\begin{tikzpicture}[x=0.75pt,y=0.75pt,yscale=-1,xscale=1]

\draw  [fill={rgb, 255:red, 255; green, 255; blue, 255 }  ,fill opacity=1 ][line width=3]  (86,153) -- (193,153) -- (193,222) -- (86,222) -- cycle ;
\draw  [fill={rgb, 255:red, 255; green, 255; blue, 255 }  ,fill opacity=1 ][line width=3]  (86,225) -- (193,225) -- (193,279) -- (86,279) -- cycle ;
\draw  [fill={rgb, 255:red, 255; green, 255; blue, 255 }  ,fill opacity=1 ][line width=3]  (10,10) -- (117,10) -- (117,79) -- (10,79) -- cycle ;
\draw  [fill={rgb, 255:red, 255; green, 255; blue, 255 }  ,fill opacity=1 ][line width=3]  (286,152) -- (393,152) -- (393,221) -- (286,221) -- cycle ;
\draw  [fill={rgb, 255:red, 255; green, 255; blue, 255 }  ,fill opacity=1 ][line width=3]  (286,225) -- (393,225) -- (393,279) -- (286,279) -- cycle ;
\draw  [fill={rgb, 255:red, 255; green, 255; blue, 255 }  ,fill opacity=1 ][line width=3]  (193,10) -- (300,10) -- (300,79) -- (193,79) -- cycle ;
\draw  [fill={rgb, 255:red, 255; green, 255; blue, 255 }  ,fill opacity=1 ][line width=3]  (375,10) -- (482,10) -- (482,79) -- (375,79) -- cycle ;
\draw  [fill={rgb, 255:red, 0; green, 0; blue, 0 }  ,fill opacity=1 ] (134,35) -- (162.4,35) -- (162.4,30) -- (181.33,40) -- (162.4,50) -- (162.4,45) -- (134,45) -- cycle ;
\draw  [fill={rgb, 255:red, 0; green, 0; blue, 0 }  ,fill opacity=1 ] (315,35) -- (343.4,35) -- (343.4,30) -- (362.33,40) -- (343.4,50) -- (343.4,45) -- (315,45) -- cycle ;
\draw  [fill={rgb, 255:red, 0; green, 0; blue, 0 }  ,fill opacity=1 ] (261.33,196) -- (232.93,196) -- (232.93,201) -- (214,191) -- (232.93,181) -- (232.93,186) -- (261.33,186) -- cycle ;
\draw  [fill={rgb, 255:red, 0; green, 0; blue, 0 }  ,fill opacity=1 ] (425.94,101.8) -- (405.86,121.88) -- (409.39,125.42) -- (388.93,131.73) -- (395.25,111.28) -- (398.78,114.81) -- (418.87,94.73) -- cycle ;
\draw   (75.19,94.75) -- (91.48,118.01) -- (95.57,115.14) -- (98.24,136.39) -- (79.19,126.61) -- (83.29,123.75) -- (67,100.48) -- cycle ;
\draw  [line width=2.25]  (405.33,282) .. controls (410,282) and (412.33,279.67) .. (412.33,275) -- (412.33,226) .. controls (412.33,219.33) and (414.66,216) .. (419.33,216) .. controls (414.66,216) and (412.33,212.67) .. (412.33,206)(412.33,209) -- (412.33,157) .. controls (412.33,152.33) and (410,150) .. (405.33,150) ;

\draw (140.98,186.69) node  [font=\Large]  {$\PP_{N-1}$};
\draw (140.98,250.69) node  [font=\Large]  {$\Peel{\PP_N^p}$};
\draw (64.97,43.69) node  [font=\Large]  {$\PP_N^p$};
\draw (340.98,185.69) node  [font=\Large]  {$\PPP{N}{N-1}$};
\draw (340.98,250.69) node  [font=\Large]  {$\Peel{\PP_N^p}$};
\draw (247.98,43.69) node  [font=\Large]  {$\TRPP_N^p$};
\draw (429.98,43.69) node  [font=\Large]  {$\TRPP_N^o$};
\draw (442.5,218.5) node  [font=\Large]  {$\PP_N^o$};

\end{tikzpicture}

}
\caption{Sequence of program transformations to decompose $\PP_N^p$
  into $\PP_{N-1};\Peel{\PP_N^p}$}
\label{fig:trans-seq}
\end{figure}
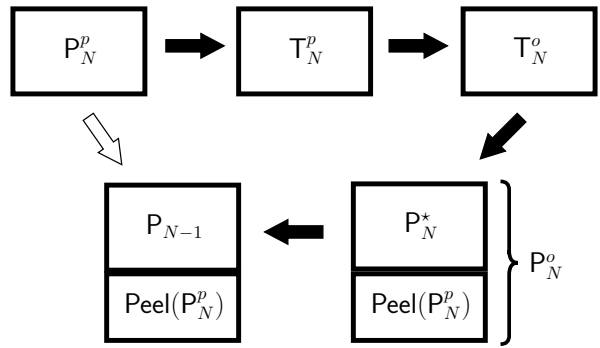

%% file: transform-pn.tex
\subsubsection{Canonicalizing the Program}
\label{sec:transform}

In this section, we describe a simple transformation of the program
$\PP_N^p$ that allows us to view the program as a \emph{linear
sequence of statements} of a specific form.  The transformation is
only meant for purposes of simplifying the proofs of lemmas in the
subsequent subsections and making them more approachable.


For every program $\PP_N$ that can be generated by the grammar shown
in Sect.~\ref{sec:prelim}, we rewrite the corresponding peeled program
$\PP_N^p$ as a linear sequence of statements of the form:

$${\iif}(\BoolCond) ~ {\tthen} ~ {\ProgFrag} ~ {\eelse} ~ {\Skip},$$

\noindent
where {\Skip} is shorthand for the assignment statement {\tt x = x;}
for an arbitrary scalar variable {\tt x} in $\PP_N^p$.  The program
fragment {\ProgFrag} is either (i) a loop, or (ii) a peel of a loop,
or (iii) an assignment statement outside loops and peels in $\PP_N^p$.
The conditional expression {\BoolCond} is a conjunction of boolean
expressions along the $\true$ (resp. $\false$) branches of all the
conditional branch nodes $b$ within the scope of which the program
fragment {\ProgFrag} occurs in the program $\PP_N^p$.  Since {\Skip}
does not change values of any variables or arrays, we omit the else
part in our subsequent discussion for notational clarity.  Henceforth,
we refer to statements in this form as guarded statements.

We now describe how to construct a program $\TRPP_N^p$ consisting of a
sequence of guarded statements starting from a given peeled program
$\PP_N^p$.  We have already seen earlier how to construct the
collapsed CFG of $\PP_N^p$.  For each loop that has been peeled, we
additionally collapse all nodes in the peel of the loop to a single
node to obtain an even more collapsed CFG $G_C$.  By the restrictions
imposed by our grammar, $G_C$ is necessarily a directed acyclic graph.
We first assign a topological index to all nodes of $G_C$ such that
the index of a node $n$ is strictly larger than the indices of all
nodes that have edges to $n$ in $G_C$.  For each node $n$ in order of
its topological index, we now construct a guarded statement for it as
follows.  If $n$ does not lie within the scope of any branch
statements, we construct the guarded statement `{\iif}($\true$) ~
{\tthen} ~ {\ProgFrag}' where {\ProgFrag} corresponds to the loop,
peel or assignment statement at node $n$.  Otherwise, we conjoin the
conditional expressions along the $\true$ (resp. $\false$) branches of
all the conditional branch nodes $b$ within the scope of which node
$n$ lies in $G_C$.  If this conjoined expression is {\BoolCond}, then
we generate the guarded statement `{\iif}({\BoolCond) ~ {\tthen} ~
{\ProgFrag}' where {\ProgFrag} is as previously defined.  We
illustrate the above construction on an example program with loops and
branches.

\input{ex-transform-pn}

\begin{example} \label{example:canonical}
Consider the peeled program $\PP_N^p$ shown in
Tab. \ref{tab:ex-trans-pn}(a).  The variables and arrays of the input
program are renamed (as described in Sect. \ref{sec:renaming})
ensuring that along each path of the program the value of a variable
or an array element is accessible till the end of the path and all the
loops in the program are peeled (as described in
Sect. \ref{sec:peeling}).  We number the lines in the program such
that the statements in loops (and in peels) have the same line number
but with an alphabet appended when multiple statements are present.
This numbering allows us to refer to the program statements with a
consistent line number even after collapsing loops and peels.
$\PP_N^p$ has three peeled loops, {\tt L1} (lines $4a$ and $4b$), {\tt
L2} (lines $8a$ and $8b$) and {\tt L3} (lines $11a$ and $11b$).  Loops
{\tt L1} and {\tt L2} along with their peels are within the scope of
the conditional branch statement in line $2$.  Loop {\tt L3} and its
peel are not within the scope of any branch statement.


Tab. \ref{tab:ex-trans-pn}(b) shows the program $\TRPP_N^p$ generated
by our transformation.  To distinguish the statements in $\PP_N^p$
from those in $\TRPP_N^p$, the line numbers of statements in
$\TRPP_N^p$ are suffixed with a prime symbol.  As can be seen,
$\TRPP_N^p$ is a linear sequence of guarded statements.  The line
numbers of each guarded statement in $\TRPP_N^p$ matches the line
number of the corresponding statement in $\PP_N^p$.  Notice that every
assignment statement in $\PP_N^p$ appears exactly at one unique
location in $\TRPP_N^p$.
\end{example}

Note that if a variable/array element is used in the conditional
expression at a branch node $b$, then it must have been updated (if at
all) in an assignment node that has a control flow path to $b$.  From
the no-overwriting property as stated in
Lemma \ref{lemma:peel:no-overwriting}, it now follows that the same
variable/array element cannot be updated in any node that has a
control flow path from $b$.  This includes all nodes within the scope
of branch $b$.  This interesting property allows us to prove the
following lemma.

\begin{lemma}\label{lemma:transform-pn-correct}
Let $\TRPP_N^p$ be the canonicalized version of $\PP_N^p$.  Let $n$ be
a node corresponding an assignment statement in the CFG of $\PP_N^p$
(and hence $\TRPP_N^p$).  Let $\pi$ (resp. $\pi'$) be a path in the
collapsed CFG of $\PP_N^p$ (resp. $\TRPP_N^p$) starting from the state
$\sigma$.  Then the following hold.
\begin{enumerate}
\item $n$ is reached along $\pi$ iff $n$ is also reached along $\pi'$.
\item The program state $\sigma_n$ is computed at node $n$ along $\pi$
iff the program state $\sigma_n$ is computed at node $n$ along $\pi'$.
\end{enumerate}
\end{lemma}
\begin{proof}\textit{(Sketch)}
We consider a path $\pi$ in the collapsed CFG of $\PP_N^p$
corresponding to an execution starting from $\sigma$.  Let $n$ be a
node corresponding to an assignment statement along path $\pi$. Let
$\widehat{\pi}$ be the prefix of $\pi$ that ends at $n$.  We prove by
induction on the length of $\widehat{\pi}$ that $n$ is also reached
along $\pi'$ and the program state computed at $n$ along $\pi$ is the
same as the state at $n$ along $\pi'$.  The proof crucially uses the
fact that the guards of all statements in $\TRPP_N^p$ that do not
correspond to statements in nodes along $\pi$ evaluate to $\false$.
The reasons for this are (i) every such guard has a conjunct that is the
negation of some branch condition $b$ that evaluates to $\true$ along
$\pi$, (ii) the consequence of the no-overwriting property stated
above, and (iii) sequencing of statements in topological index order
in $\TRPP_N^p$.  In particular, (ii) and (iii) above ensure that the
values of all variables/array elements used in the branch node $b$
along $\pi$ are the same as the corresponding values used in $b$ along
$\pi'$.

The converse direction of the proof is similar.
\end{proof}

%% file: ex-transform-pn.tex
\fboxrule=1pt
\fboxsep=1pt

\begin{table*}[h!]
\caption{(a) Program $\PP_N^p$, (b) Transformed Program $\TRPP_N^p$,
  (c) Transformed \& Reordered Program $\TRPP_N^o$ and (d) Program
  $\PP_N^o$ ~/~ Program $\PPP{N}{N-1};\Peel{\PP_N^p}$}
\label{tab:ex-trans-pn}
\begin{tabular} { p{7.6cm}  p{6.6cm} }
\begin{alltt}
1.   S = 10;
2.   if(S > 5) \{
3.     S1 = 1;
\fcolorbox{red}{white}{\begin{minipage}{7.5cm}4a.    for(i=0; i<N-1; i++)  //Loop L1
4b.      A1[i] = A[i] + 1;\end{minipage}}
5.     A1[N-1] = A[N-1] + 1; //Peel of L1
6.   \} else \{
7.     S1 = 20;
\fcolorbox{green}{white}{\begin{minipage}{7.5cm}8a.    for(i=0; i<N-1; i++)  //Loop L2
8b.      A1[i] = A[i];\end{minipage}}
9.     A1[N-1] = A[N-1];     //Peel of L2
10.  \}
\fcolorbox{blue}{white}{\begin{minipage}{7.5cm}11a. for(i=0; i<N-1; i++)    //Loop L3
11b.   A2[i] = A1[i] + S1;\end{minipage}}
12.  A2[N-1] = A1[N-1] + S1; //Peel of L3
\end{alltt}
  &
\begin{alltt}
1'.  if(\(\true\)) S = 10;
3'.  if(S > 5) S1 = 1;
4'.  if(S > 5) Loop L1;
5'.  if(S > 5) Peel of L1;
7'.  if(!(S > 5)) S1 = 20;
8'.  if(!(S > 5)) Loop L2; 
9'.  if(!(S > 5)) Peel of L2;
11'. if(\(\true\)) Loop L3;
12'. if(\(\true\)) Peel of L3;
\end{alltt}
\\
\begin{center}(a)\end{center}
&
\begin{center}(b)\end{center}
\\
\begin{alltt}
1'.  if(\(\true\)) S = 10;
3'.  if(S > 5) S1 = 1;
4'.  if(S > 5) Loop L1;
7'.  if(!(S > 5)) S1 = 20;
8'.  if(!(S > 5)) Loop L2; 
11'. if(\(\true\)) Loop L3;

5'.  if(S > 5) Peel of L1;
9'.  if(!(S > 5)) Peel of L2;
12'. if(\(\true\)) Peel of L3;
\end{alltt}
&
\begin{alltt}
1.   S = 10;
2a.  if(S > 5) \{
3.     S1 = 1;
4a.    for(i=0; i<N-1; i++)  //Loop L1
4b.      A1[i] = A[i] + 1;
6a.  \} else \{
7.     S1 = 20;
8a.    for(i=0; i<N-1; i++)  //Loop L2
8b.      A1[i] = A[i];
10a. \}
11a. for(i=0; i<N-1; i++)    //Loop L3
11b.   A2[i] = A1[i] + S1;

2b.  if(S > 5) \{
5.     A1[N-1] = A[N-1] + 1; //Peel of L1
6b.  \} else \{
9.     A1[N-1] = A[N-1];     //Peel of L2
10b. \}
12.  A2[N-1] = A1[N-1] + S1; //Peel of L3
\end{alltt}
\\
\begin{center}(c)\end{center}
&
\begin{center}(d)\end{center}
\\
\end{tabular}
\end{table*}

%% file: reorder-trans-pn.tex
\subsubsection{Reordering the Peels}
\label{sec:reorder}

We now reorder the statements in the program $\TRPP_N^p$ such that all
guarded statements corresponding to peels of loops are executed after
all other guarded statements.  However, the relative ordering among
the guarded statements corresponding to peels is preserved.  We use
$\TRPP_N^o$ to denote the program obtained after this reordering.

\begin{example}
Continuing with the canonicalized program $\TRPP_N^p$ shown in
Tab. \ref{tab:ex-trans-pn}(b), the reordered program $\TRPP_N^o$ is
shown in Tab. \ref{tab:ex-trans-pn}(c).  The line numbers follow the
same pattern described in Example \ref{example:canonical}.  Notice
that in $\TRPP_N^o$ the guarded statements corresponding to peels at
lines $5'$, $9'$ and $12'$ appear after all other statements in
the program.
\end{example}

Let $vA$ be a variable/array in $\PP_N^p$ that is not identified as
affected by function \textsc{ComputeAffected}.  The following lemmas
establish that if programs $\TRPP_N^p$ and $\TRPP_N^o$ are executed
from the same state $\sigma$, they always compute the same value of
$vA$.  Specifically, Lemma \ref{lemma:preserve-dependence} shows that
all data dependencies that potentially have a bearing on the value of
$vA$ are the same in $\TRPP_N^p$ and $\TRPP_N^o$.
Lemma \ref{lemma:sound-reordering-wo-affected} uses this to show that
the value of $vA$ computed by $\TRPP_N^p$ and $\TRPP_N^o$ are the
same.  For clarity of exposition in the following discussion, when we
say that there is a data dependence path from $n_i$ to $n_j$ in a
program, we mean that there is a path from $n_i$ to $n_j$ in the DDG
of the program.  Similarly, when we say that there is a control
dependence path from $n_i$ to $n_j$, we mean that there is a data
dependence path from $n_i$ to a conditional branch node within whose
scope $n_j$ lies.


\begin{lemma}\label{lemma:preserve-dependence}
Let $vA$ be a scalar variable/array that is absent in $\affectedvars$
when \textsc{ComputeAffected} is invoked on $\PP_N^p$.  Let $n$ be a
node in $\PP_N^p$ such that $vA \in \mathit{def}(n)$.  For every node
$n'$ in $\PP_N^p$, there is a data/control dependence path from $n'$
to $n$ in $\TRPP_N^p$ iff there exists such a path in $\TRPP_N^o$.
\end{lemma}
\begin{proof}
Consider nodes $n$ and $n'$ in $\PP_N^p$ (hence also in $\TRPP_N^p$
and $\TRPP_N^o$).  We consider two cases.
\begin{enumerate}
\item Suppose there is a data/control dependence path
from $n'$ to $n$ in $\TRPP_N^p$.  From our construction of
$\TRPP_N^p$, we know that there exists such a data/control dependence
path in $\PP_N^p$ as well.  We now show that such a data/control
dependence path also exists in $\TRPP_N^o$ by considering two
sub-cases.
\begin{enumerate}
\item Suppose $n$ is a non-peeled node in $\PP_N^p$.  Since $vA$ is not
identified as affected, by the not-affected property
(Lemma \ref{lemma:not-affected}), we have that $n'$ is not a peeled
node.  Therefore, both $n$ and $n'$ are non-peeled nodes.  Since the
relative ordering of all non-peeled nodes is preserved by our
reordering transformation, the data/control dependence between $n$ and
$n'$ continues to exist in $\TRPP_N^o$ as well.
\item Suppose $n$ is a peeled node in $\PP_N^p$.  If $n'$ is also a
peeled node, then since the relative ordering among the peeled nodes
is preserved by reordering, the data/control dependence between $n$
and $n'$ continues to exist in $\TRPP_N^o$.  On the other hand, if
$n'$ is a non-peeled node, since all non-peeled nodes precede all
peeled nodes after reordering, the data/control dependence exists in
$\TRPP_N^o$ in this case as well.
\end{enumerate}
\item  Suppose there is a data/control dependence path from $n'$ to $n$
in $\TRPP_N^o$.  We show that such a dependence path exists in
$\TRPP_N^p$ by considering the following sub-cases.
\begin{enumerate}
\item If both $n'$ and $n$ are non-peeled (resp. peeled) nodes, then
since reordering does not change the relative ordering of the
non-peeled (resp. peeled) nodes, the dependence is present in
$\TRPP_N^p$.
\item The case where $n$ is a non-peeled node and $n'$ is peeled node
cannot arise, since all non-peeled nodes appear before peeled nodes in
$\TRPP_N^o$.
\item If $n'$ is a non-peeled node and $n$ is a peeled node, there are
two further sub-cases.  If $n'$ is ordered before $n$ in $\TRPP_N^p$
then the dependence is present in $\TRPP_N^p$ as well.  Otherwise, we
ask if the dependence from $n'$ to $n$ in $\TRPP_N^o$ is a
read-after-write or write-after-write.  Since $n$ is ordered before
$n'$ in $\TRPP_N^p$, by the no-overwriting property
(Lemma \ref{lemma:peel:no-overwriting}), both $n$ and $n'$ cannot
update the same renamed variable.  Therefore, the dependence from $n'$
to $n$ in $\TRPP_N^o$ cannot be write-after-write, and hence must be
read-after-write.  This requires a variable/array element in
$\mathit{uses}(n)$ to also be present in $\mathit{def}(n')$.  Such a
variable/array element must have been updated (if at all) prior to its
use in node $n$.  Once again, by the no-overwriting property, this
variable/array element cannot be updated by $n'$, which appears after
$n$ in $\TRPP_N^p$.  This completes the proof.
\end{enumerate}
\end{enumerate}
\end{proof}

\begin{lemma} \label{lemma:sound-reordering-wo-affected}
  Let $vA$ be a scalar variable/array that is absent in
  $\affectedvars$ upon invocation of \textsc{ComputeAffected} on
  $\PP_N^p$.  If $\TRPP_N^p$ and $\TRPP_N^o$ are executed from the
  same state $\sigma$, then $vA$ has the same value on termination of
  both programs.
\end{lemma}
\begin{proof}
  Follows from Lemma \ref{lemma:preserve-dependence} and the fact that
  reordering does not change the individual guarded statements in the
  canonicalized program $\TRPP_N^p$.
\end{proof}

%% file: reverse-trans-pn.tex
\subsubsection{De-canonicalizing the Reordered Program}
\label{sec:reverse-trans}

Recall that our aim is to decompose the program $\PP_N^p$ into two
program fragments, the program $\PP_{N-1}$ and the difference program
$\partial \PP_N$.  We have seen above that the reordering step already
achieves the purpose of moving the guarded statements corresponding to
peels to the end of the program, providing a good candidate for the
difference program $\partial \PP_N$.  However, the part of the
reordered program that precedes the guarded statements corresponding
to peels may not have syntactic similarity with $\PP_{N-1}$ in
general.  In order to remedy this situation, we now ``undo'' the
canonicalization process (as described in Sect. \ref{sec:transform})
that allowed us to view the program $\PP_N^p$ as a linear sequence of
guarded statements.  Specifically, we transform the guarded statements
back to statements of the form that were present in $\PP_N^p$ to begin
with.  We do this separately for the guarded statements corresponding
to peels, and for the part of $\TRPP_N^o$ that precedes these, so that
we obtain a program fragment that is syntactically similar to
$\PP_{N-1}$ followed by a difference program.  In the subsequent
discussion, we call the resulting de-canonicalized program $\PP_N^o$.

\begin{example}
Consider the reordered program $\TRPP_N^o$ from the example shown in
Tab. \ref{tab:ex-trans-pn}(c).  The program $\PP_N^o$ shown in
Tab. \ref{tab:ex-trans-pn}(d) is obtained by de-canonicalization.  The
line numbers follow the pattern similar to the program $\PP_N^p$ as
described in Example \ref{example:canonical}.  It is worth noticing
that, the statements corresponding to peels of loops appear after all
other statements in $\PP_N^o$ and part of program $\PP_N^o$ that
precedes the peels is syntactically similar to $\PP_{N-1}$.
\end{example}

Notice that after de-canonicalization, there may be more conditional
branch nodes in the CFG of $\PP_N^o$ as compared to the CFG of
$\PP_N^p$ but fewer conditional branch nodes as compared to the CFG of
$\TRPP_N^o$.  

\begin{lemma} \label{lemma:reverse}
Let $\PP_N^o$ be the de-canonicalized version of $\TRPP_N^o$.  Let $n$
be a node corresponding an assignment statement in the CFG of
$\TRPP_N^o$ (and hence $\PP_N^o$).  Let $\pi$ (resp. $\pi'$) be a path
in the collapsed CFG of $\TRPP_N^o$ (resp. $\PP_N^o$) starting from
the state $\sigma$.  Then the following hold.
\begin{enumerate}
\item $n$ is reached along $\pi$ iff $n$ is also reached along $\pi'$.
\item The program state $\sigma_n$ is computed at node $n$ along $\pi$
iff the program state $\sigma_n$ is computed at node $n$ along $\pi'$.
\end{enumerate}
\end{lemma}
\begin{proof}
The proof is similar to that shown in
Lemma \ref{lemma:transform-pn-correct}.
\end{proof}

%% file: ex-ss-incorrect-diff.tex
\begin{figure}[h]
\begin{alltt}
// assume(\(\forall\)i\(\in\)[0,N) A[i] = 1)

1.  S = 0;
2.  for(i=0; i<N-1; i++) \{
3.    S = S + A[i];
4.  \}

5.  for(i=0; i<N-1; i++) \{
6.    A1[i] = A[i] + S;
7.  \}

8.  S1 = S;
9.  for(i=0; i<N-1; i++) \{
10.   S1 = S1 + A1[i];
11. \}

12. S = S + A[N-1];
13. A1[N-1] = A[N-1] + S;
14. S1 = S1 + A1[N-1];

// assert(S1 = N \(\times\) (N+2))
\end{alltt}
\caption{Hoare triple with the program $\PP_{N-1};\Peel{\PP_N}$}
\label{fig:ss-incorrect-diff}
\end{figure}

%% file: algo-prog-diff.tex
\begin{algorithm*}[!t]
  \caption{\textsc{ProgramDiff}($(\mathit{Locs}, CE, \mu)$: renamed program $\PP_N$, $\gluenodes$: set of glue nodes)}
  \label{alg:diff-generation}
  \begin{algorithmic}[1]
    \State $\langle(\mathit{Locs}^p, CE^p, \mu^p), \peelnodes\rangle$ := \textsc{PeelAllLoops}($(\mathit{Locs}, CE, \mu)$); \label{line:diff:peelloops}
    \State $\affectedvars$ := \textsc{ComputeAffected}$((\mathit{Locs}^p, CE^p, \mu^p), \peelnodes)$; \label{line:diff:affected}
    \State $ \partial \PP_N := (Locs', CE', \mu') $,
    where $Locs' := Locs^p$, $CE' := CE^p$, and $\mu' := \mu^p$; \label{line:diff:skeletonpn}
    \State $\partial \PP_N := \textsc{CollapseLoopBody}(\partial \PP_N)$;  \label{line:diff:collapse}  \Comment{Collapse nodes and edges of each loop into its loop-head}
    \State $\worklist$ := $(n_{start})$;  \Comment{Add the start node to the worklist} \label{line:diff:init-wl}
    \State $\processednodes$ := $\emptyset$;  \label{line:diff:processednodesinit}
    \While{$\worklist$ is not empty}  \label{line:diff:begin-wl}
      \State Remove a node $n$ from head of $\worklist$;  \label{line:diff:removeonenode}
      \State $\processednodes$ := $\processednodes \cup \{n\}$;  \label{line:diff:addntoprocessed}
      \For{each node $n' \in$ \textsc{Succ}$(n) \setminus \processednodes$} \label{line:diff:copyst}  
        \State $\worklist$ := \textsc{AppendToList}($\worklist$, $n'$);  \label{line:diff:append-wl}
      \EndFor \label{line:diff:copyend}
      \If{$n \in \peelnodes$}   \label{line:diff:check-peel}
        {\bf ~ continue};  \label{line:diff:skip-peel}  \Comment{Difference computation not required}
      \ElsIf{\textsc{HasAffectedVars}($n$, $\affectedvars$)} \label{line:diff:check-affected}
        \If{$n \in \gluenodes$} \label{line:diff:check-glue} {\bf ~ continue};  \label{line:diff:retain-affected-glue}  \Comment{Retain the glue loop}
        \ElsIf{$n$ is a loop-head} \label{line:diff:loop-head}
          \State $L := \textsc{UncollapseLoopBody}(n)$;  \label{line:diff:uncollapse-l}  \Comment{Uncollapse the loop-head}
          \For{each node $n' \in \textsc{Nodes}(L)$} \label{line:diff:loopnodes-start}
            \If{\textsc{HasAffectedVars}($n'$, $\affectedvars$)} \label{line:diff:check-affected2}
              \State $\mu'(n') := \textsc{NodeDiff}(n', \mu, \affectedvars)$;  \label{line:diff:stdiff-loop}
            \Else ~ 
               $\partial \PP_N$ := \textsc{RemoveNode}($n'$, $\partial \PP_N$);  \label{line:diff:removen1} \Comment{No affected variables at node $n'$} 
            \EndIf
          \EndFor \label{line:diff:loopnodes-end}
        \Else
          ~ $\mu'(n) := \textsc{NodeDiff}(n, \mu, \affectedvars)$;  \label{line:diff:stdiff}
        \EndIf
      \Else \Comment{No affected variables at node $n$}
        \State Let $\condnodes$ be the set of all branch nodes that have a peeled node within its scope;  \label{line:diff:condnodes}
        \If{$n \not \in \condnodes$ } \label{line:diff:skip-branch}
          $\partial \PP_N$ := \textsc{RemoveNode}($n$, $\partial \PP_N$);  \label{line:diff:removen2} 
        \EndIf
      \EndIf
    \EndWhile  \label{line:diff:end-wl}
    \State \Return ${\partial \PP_N}$;
  \end{algorithmic}

  \vspace{2ex}
  \textsc{HasAffectedVars}( $n$: node, $\affectedvars$: set of affected variables )
  \begin{algorithmic}[1]
    \If{$\exists vA$ such that $vA \in \mathit{def}(n)$ and $vA \in \affectedvars$}
      ~ \Return $\true$;
    \Else ~ \Return $\false$;
    \EndIf
  \end{algorithmic}

  \vspace{2ex}
  \textsc{NodeDiff}( $n$: node, $\mu$: node labelling function, $\affectedvars$: set of affected variables )
  \begin{algorithmic}[1]
    \If{$\mu(n)$ is of the form $w_{N} := r^1_{N}$ {\tt op} $r^2_{N}$}
      \State \Return $w_N$ $:=$ $w_{N-1}$ $\circ$ $(r^1_N$ $\circ$ $(r^1_{N-1})^{-\circ})$ $\circ$ $(r^2_N$ $\circ$ $(r^2_{N-1})^{-\circ})$; \Comment{Refer Lemma~\ref{lemma:diffcomp2}}
    \ElsIf{$\mu(n)$ is of the form $w_{N} := w_{N}$ {\tt op} $r^1_{N}$ wherein $w_{N}$ is a scalar}
      \State \Return $w_N$ $:=$ $w_{N}$ $\circ$ $(r^1_N$ $\circ$ $(r^1_{N-1})^{-\circ})$; \Comment{Refer Lemma~\ref{lemma:diffcomp3}}
    \Else \Comment{$\mu(n)$ is a conditional statement $C_N$}
      \If{($\exists v$ s.t. $v \in uses(n)$ and $v \in \affectedvars$) $\vee$ ($C_N \not= C_{N-1}$ is satisfiable)}
        \State {\bf throw} ``Branch conditions in $\PP_N$ and $\PP_{N-1}$ may not evaluate to same value'';  \label{line:diff:branch}
      \Else
        ~ \Return $\mu(n)$;
      \EndIf
    \EndIf
  \end{algorithmic}
\end{algorithm*}

%% file: ex-ss-diff.tex

\begin{figure}[h]
\begin{alltt}
// assume(\(\forall\)i\(\in\)[0,N) A[i] = 1)

1.  S = S_Nm1 + A[N-1];

2.  for(i=0; i<N-1; i++) \{
3.    A1[i] = A1_Nm1[i] + (S - S_Nm1);
4.  \}
5.  A1[N-1] = A[N-1] + S;

6.  S1 = S1_Nm1 + (S - S_Nm1);
7.  for(i=0; i<N-1; i++) \{
8.    S1 = S1 + (A1[i] - A1_Nm1[i]);
9.  \}
10. S1 = S1 + A1[N-1];

// assert(S1 = N \(\times\) (N+2))
\end{alltt}
\caption{Hoare triple with the difference program}
\label{fig:ss-diff}
\end{figure}


%% file: simp-diff.tex
\subsection{Simplifying the Difference Program}
\label{sec:simp-diff}

\input{algo-simp-diff.tex}  

While we have described a simple strategy to generate a difference
program ${\partial \PP_N}$ above, this may lead to unoptimized as well
as redundant statements in the naively generated difference program.
Our implementation aggressively optimizes $\partial \PP_N$ and removes
redundant code, renaming variables/arrays as needed.  The routine
\textsc{SimplifyDiff} in Algorithm~\ref{alg:simpdiff} simplifies
program statements that compute rectified values, removes redundant
loops from the difference program and substitutes loops with the
summarized statements computed using acceleration.  This helps in
${\partial \PP_N}$ having fewer and simpler loops in a lot of cases.
Below, we describe these optimizations and illustrate them using
examples.

Since the generation of statements that compute rectified values is
not fully optimized, these statements may have expressions that can be
further simplified using the values computed in other statements in
the generated difference program ${\partial \PP_N}$.  The function
\textsc{Simplify} performs this optimization aggressively and
simplifies the statements in the difference program
(lines~\ref{line:simp:simplify1}--\ref{line:simp:simplify2} in
Algorithm~\ref{alg:simpdiff}).  Let us take an example to illustrate
the effect of the \textsc{Simplify} function.

\begin{example}
Suppose the difference program ${\partial \PP_N}$ has statements of
the form {\tt B\_N[i] = B\_Nm1[i] + expr1;} and {\tt v\_N =
expr2*v\_Nm1;}.  If {\tt expr1} and {\tt expr2} are constants or
functions of $N$ and loop counters, then expressions such as {\tt
(B\_N[i] - B\_Nm1[i])} and {\tt (v\_N/v\_Nm1)} can often be simplified
from the statements in the difference program.  The expression {\tt
expr1} is substituted for {\tt B\_N[i] - B\_Nm1[i]} and {\tt expr2}
for {\tt v\_N/v\_Nm1} respectively.
\end{example}

The difference program ${\partial \PP_N}$ may contain loops that
perform redundant computation, for example, copying values across
versions of an array corresponding to $\PP_N$ and $\PP_{N-1}$, due to
the simplification of the statements that compute its rectified value.
We remove such loops from ${\partial \PP_N}$ in lines
\ref{line:simp:redundant-begin}--\ref{line:simp:redundant-end} of
Algorithm \ref{alg:simpdiff}.  Let us illustrate this with an example.

\begin{example}
Suppose the difference program $\partial \PP_N$ has the loop {\tt
  for(i=0; i<N-1; i++) \{ A\_N[i] = A\_Nm1[i]; \}} where {\tt A\_N} is
not used subsequently in the program.  Such loops can be removed from
the difference program, as these loops only copy values from the
version of array {\tt A} in $\PP_{N-1}$ to its version in $\PP_N$, and
hence, are redundant.
\end{example}

The difference program ${\partial \PP_N}$ may also contain loops that
compute values of variables that can be accelerated.  We perform this
optimization in lines
\ref{line:simp:check-form}--\ref{line:simp:add-remove-end} of
\textsc{SimplifyDiff} in Algorithm~\ref{alg:simpdiff}.  We first check
if the body of a loop $L$ is in the specific form eligible for this
optimization in line \ref{line:simp:check-form}.  If so, we create a
fresh node in line \ref{line:simp:fresh-node} to replace $L$.  Lines
\ref{line:simp:accelerate-begin}--\ref{line:simp:accelerate-end} of
Algorithm \ref{alg:simpdiff} label the fresh node with the accelerated
statement.  If we encounter operators that are not supported, then we
report a failure of our technique using the {\bf throw} statement in
line \ref{line:simp:throw}.  Next, we replace the loop with the fresh
node in lines
\ref{line:simp:add-remove-begin}--\ref{line:simp:add-remove-end}.  We
demonstrate this optimization with the following example.

\begin{example}
Suppose the difference program has the loop {\tt for(i=0; i<N-1; i++)
  \{ sum = sum + 1; \}}.  The semantics of the loop can be summarized
using the accelerated statement {\tt sum = sum + (N-1);}.
\textsc{SimplifyDiff} removes this loop from the program and
introduces the accelerated statement instead.
\end{example}

In the following lemma, we use $\partial \PP'_N$ to denote the program
generated by \textsc{SimplifyDiff}.

\begin{lemma}
\label{lemma:simp-sound}
$\{\varphi(N)\} \;\PP_{N-1};{\partial \PP'_N} \;
\{\psi(N)\}$ holds iff $\{\varphi(N)\} \;\PP_{N-1};{\partial \PP_N} \;
\{\psi(N)\}$ holds.
\end{lemma}
\begin{proof} Follows trivially from the fact that
\textsc{SimplifyDiff} in Algorithm~\ref{alg:simpdiff} optimizes
program statements, removes only redundant statements/loops, and
replaces loops using semantically equivalent accelerated statements.
\end{proof}

\input{ex-ss-simp}  

\begin{example}
We illustrate the application of the simplification routine
\textsc{SimplifyDiff} from Algorithm~\ref{alg:simpdiff} on our running
example.  The program in Fig.~\ref{fig:ss-diff-simp} is obtained after
simplification of $\partial \PP_N$ in Fig.~\ref{fig:ss-diff}.  The
difference terms are replaced with the simplified expressions from the
difference program itself.  Notice that the loop that rectifies the
value of {\tt S1} is accelerated and the statement obtained after this
optimization is shown in line 7.
\end{example}


%% file: algo-simp-diff.tex
\begin{algorithm*}[!t]
  \caption{\textsc{SimplifyDiff}($(\mathit{Locs}, CE, \mu)$: difference program $\partial \PP_N$)}
  \label{alg:simpdiff}
  \begin{algorithmic}[1]
    \State $ \partial \PP'_N := (Locs', CE', \mu')$,
    where $Locs' := Locs$, $CE' := CE$, and $\mu' := \mu$;
    \For{each loop $L \in \textsc{Loops}(\partial \PP'_N)$}
      \For{each node $m \in \textsc{Nodes}(L)$} \label{line:simp:simplify1}
        \State $\mu'(m)$ := \textsc{Simplify}($\mu'(m)$); \label{line:simp:simplify2}  \Comment{Simplify the statement}
      \EndFor
      \State $(n_1,n,c) := \textsc{IncomingEdge}(L)$;  \Comment{$c$ is the label of the edge from $n_1$ to $n$}
      \State$(n,n_2,\lfalse):=\textsc{ExitEdge}(L)$;
      \If{body of $L$ is of the form $w_{N} := w_{N}$ {\tt op} $expr$, wherein $w_{N}$ is a scalar variable}  \label{line:simp:check-form}
        \State $n_{acc}$ = $\textsc{FreshNode}()$;  \label{line:simp:fresh-node}
        \If{{\tt op} $\in \{ +, -\}$}     \label{line:simp:accelerate-begin}
          \State $\mu'(n_{acc})$ := ($w_{N} := w_{N}$ {\tt op} \textsc{Simplify}$(k_L(N-1) \times expr)$); \Comment{Accelerated statement}
        \ElsIf{{\tt op} $\in \{\times, \div \}$}
          \State $\mu'(n_{acc})$ := ($w_{N} := w_{N}$ {\tt op} \textsc{Simplify}$(expr^{k_L(N-1)})$);  \Comment{Accelerated statement}
        \Else {}
          \State {\bf throw} ``Specified operator not handled'';  \label{line:simp:throw}
        \EndIf    \label{line:simp:accelerate-end}
        \State $CE' := CE'$ $\union$ $\{(n_1, n_{acc},c), (n_{acc}, n_2, \Unlabeled)\}$ $\setminus$ $\{(n_1,n,c), (n, n_2, \lfalse)\}$;  \label{line:simp:add-remove-begin}
        \State $Locs'$ := $Locs'$ $\union$ $\{ n_{acc} \}$ $\setminus$ $\textsc{Nodes}(L)$;  \label{line:simp:add-remove-end}
      \EndIf
      \If{body of $L$ is of the form $w_{N} := w_{N-1}$ or $w_{N} := w_{N}$}  \label{line:simp:redundant-begin}  \Comment{Remove redundant loops}
        \State $CE' := CE'$ $\union$ $\{(n_1, n_2,c)\}$ $\setminus$ $\{(n_1,n,c), (n, n_2, \lfalse)\}$;
        \State $Locs' := Locs'$ $\setminus$ $\textsc{Nodes}(L)$;
      \EndIf  \label{line:simp:redundant-end}
    \EndFor
    \State \Return $\partial \PP'_N$;
  \end{algorithmic}
\end{algorithm*}

%% file: ex-ss-simp.tex

\begin{figure}[h]
\begin{alltt}
// assume(\(\forall\)i\(\in\)[0,N) A[i] = 1)

1.  S = S_Nm1 + A[N-1];

2.  for(i=0; i<N-1; i++) \{
3.    A1[i] = A1_Nm1[i] + 1;
4.  \}
5.  A1[N-1] = A[N-1] + S;

6.  S1 = S1_Nm1 + A[N-1];
7.  S1 = S1 + (N-1);
8.  S1 = S1 + A1[N-1];

// assert(S1 = N \(\times\) (N+2))
\end{alltt}
\caption{Simplification of the difference program}
\label{fig:ss-diff-simp}
\end{figure}


%% file: diff-pre.tex
\subsection{Generating the Difference Pre-condition $\mathbf{\partial \varphi(N)}$}
\label{sec:diff-pre}

\input{algo-syntactic-diff}  

We now present a syntactic routine, called
\textsc{SyntacticDiff}, in Algorithm \ref{alg:synt-diff} for
generation of the difference pre-condition ${\partial \varphi(N)}$.
Although this suffices for all our experiments, for the sake of
completeness, we present later a more sophisticated algorithm for
generating ${\partial \varphi(N)}$ simultaneously with $\ppre(N)$ in
Sect. \ref{sec:fpi-ext}.


Formally, given $\varphi(N)$, the function \textsc{SyntacticDiff} from
Algorithm \ref{alg:synt-diff} generates a formula
${\partial \varphi(N)}$ such that $\varphi(N)$ $\rightarrow$
$(\varphi(N-1) \odot {\partial \varphi(N)})$, where $\odot$ is
$\wedge$ when $\varphi(N)$ is a universally quantified formula and is
$\vee$ when $\varphi(N)$ is a existentially quantified formula.
Observe that if such a ${\partial \varphi(N)}$ exists for universally
quantified formulas $\varphi(N)$, then $\varphi(N)$ $\rightarrow$
$\varphi(N-1)$ must hold.  Similarly, if such a
${\partial \varphi(N)}$ exists for existentially quantified formulas
$\varphi(N)$, then $\varphi(N-1)$ $\rightarrow$ $\varphi(N)$ must
hold.  Therefore, we can use the validity of $\varphi(N)$
$\rightarrow$ $\varphi(N-1)$ and $\varphi(N-1)$ $\rightarrow$
$\varphi(N)$, as a test to decide the existence of
${\partial \varphi(N)}$ for universally and existentially quantified
formulas respectively.

Algorithm \ref{alg:synt-diff} incorporates the scenarios described
above and boolean combinations thereof.  When $\varphi(N)$ is of the
syntactic form $\forall i\in \{0 \ldots N\}\; \widehat{\varphi}(i)$,
we first check the validity of $\varphi(N)$ $\rightarrow$
$\varphi(N-1)$ in line \ref{dpre:line:testconj}.  If this test fails,
we report failure using the {\bf throw} statement in line
\ref{dpre:line:throw1}.  Otherwise, ${\partial \varphi(N)}$ is set to
$\widehat{\varphi}(N)$ in line \ref{dpre:line:base1}.  Similarly, when
$\varphi(N)$ is of the syntactic form $\exists i\in \{0 \ldots N\}\;
\widehat{\varphi}(i)$, then ${\partial \varphi(N)}$ is set to
$\widehat{\varphi}(N)$ in line \ref{dpre:line:base2}, after checking
the validity of the $\varphi(N-1)$ $\rightarrow$ $\varphi(N)$ (line
\ref{dpre:line:testdisj}).  If the test in line
\ref{dpre:line:testdisj} fails, again we report failure using the {\bf
  throw} statement in line \ref{dpre:line:throw2}.  When $\varphi(N)$
is of the syntactic form $\varphi^1(N)$ $\wedge$ $\cdots$ $\wedge$
$\varphi^k(N)$, ${\partial \varphi(N)}$ is computed by taking the
difference of each individual conjunct and {\em disjuncting} them as
${\partial \varphi^1(N)}$ $\vee$ $\cdots$ $\vee$ ${\partial
  \varphi^k(N)}$ (line \ref{dpre:line:conjunct}).  Note that this
operation results in an over-approximation of the difference
pre-condition.  When $\varphi(N)$ is of the form $\varphi^1(N)$ $\vee$
$\cdots$ $\vee$ $\varphi^k(N)$, ${\partial \varphi(N)}$ is computed by
taking the difference of each individual disjunct as ${\partial
  \varphi^1(N)}$ $\vee$ $\cdots$ $\vee$ ${\partial \varphi^k(N)}$
(line \ref{dpre:line:disjunct}).  Finally, if $\varphi(N)$ does not
belong to any of these syntactic forms (line \ref{dpre:line:default})
or if condition 2(a) of Theorem \ref{thm:full-prog-ind-sound} is
violated by the ${\partial \varphi(N)}$ computed in this manner (line
\ref{dpre:line:case2a}), then we over-approximate ${\partial
  \varphi_N}$ by $\true$ in lines \ref{dpre:line:dptrue1} and
\ref{dpre:line:dptrue2}.


\begin{lemma}\label{lemma:diff-pre-condition}
  The difference pre-condition $\partial\varphi(N)$ computed by
  \textsc{SyntacticDiff} is such that (i) $\varphi(N)$ $\rightarrow$
  $(\varphi(N-1) \odot {\partial \varphi(N)})$, where $\odot \in$
  \{$\wedge$, $\vee$\}, and (ii) $\PP_{N-1}$ does not modify
  variables/arrays in $\partial\varphi(N)$.
\end{lemma}
\begin{proof}
  Condition (i) follows from the checks implemented in lines
  \ref{dpre:line:testconj} and \ref{dpre:line:testdisj} of function
  \textsc{SyntacticDiff}.  The check in line \ref{dpre:line:case2a}
  ensures condition (ii).  This concludes the proof.
\end{proof}

\begin{example}
Consider the pre-condition $\varphi(N) := \forall i\in [0,N)\; A[i] =
1$ from our running example in Fig. \ref{fig:ss-diff-simp}.  The
difference pre-condition computed by function \textsc{SyntacticDiff}
in Algorithm \ref{alg:synt-diff} is $\partial \varphi(N) := A[N-1] = 1$
shown by the assume statement in Fig. \ref{fig:ss-wp}.
\end{example}

\begin{example}
Consider the pre-condition $\varphi(N) := \forall i\in [0,N)\; A[i] =
1 \vee \forall i\in [0,N)\; A[i] = 2$.  \textsc{SyntacticDiff} in
Algorithm \ref{alg:synt-diff} enters the recursive case in
line \ref{dpre:line:disjunct}.  The recursive invocations with
inputs $\varphi_1(N) := \forall i\in [0,N)\; A[i] = 1$ and
$\varphi_2(N) := \forall i\in [0,N)\; A[i] = 2$ compute the
difference pre-conditions $\partial \varphi_1(N) := A[N-1] = 1$ and
$\partial \varphi_2(N) := A[N-1] = 2$ respectively.  On returning
from the recursive case, the algorithm stores the formula
$A[N-1] = 1 \vee A[N-1] = 2$ in $\partial \varphi(N)$.
\end{example}

\begin{example}
Consider the pre-condition $\varphi(N) := \exists i\in [0,N)\;
A[i] \geq 100$ $\wedge$ $\exists j\in [0,N)\; A[j] \leq 10$.  The
difference pre-condition computed by 
function \textsc{SyntacticDiff} in Algorithm \ref{alg:synt-diff} is
$\partial \varphi(N) := A[N-1] \geq 100 \vee A[N-1] \leq 10$.  Notice
that the computed difference pre-condition is an over-approximation.
Had we computed the difference pre-condition as 
$\partial \varphi(N) := A[N-1] \geq 100 \wedge A[N-1] \leq 10$, then
it would have resulted in a contradiction.
\end{example}

\begin{example}
For pre-condition formulas $\varphi(N)$ := $\forall i\in [0,N)\;$
$A[i] = N$ and $\varphi(N)$ := $\exists i\in [0,N)\;$ $A[i] = N$ the
validity checks at lines \ref{dpre:line:testconj}
and \ref{dpre:line:testdisj} respectively in
Algorithm \ref{alg:synt-diff} fail.  Hence, the algorithm terminates
with out being able to compute an appropriate pre-condition.
\end{example}

%% file: algo-syntactic-diff.tex
\begin{algorithm*}[!t]
  \caption{\textsc{SyntacticDiff}({$\varphi(N)$}: pre-condition)}
  \label{alg:synt-diff}
  \begin{algorithmic}[1]
    \If{$\varphi(N)$ is of the form $\forall i \in \{0 \ldots N\}\; \widehat{\varphi}(i)$}
      \If{$\varphi(N) \rightarrow \varphi(N-1)$ is invalid} \label{dpre:line:testconj}
        \State {\bf throw} ``Unable to compute the difference pre-condition"; \label{dpre:line:throw1}
      \EndIf
      \State ${\partial \varphi(N)} := \widehat{\varphi}(N)$; \label{dpre:line:base1}
    \ElsIf{$\varphi(N)$ is of the form $\exists i \in \{0 \ldots N\}\; \widehat{\varphi}(i)$}
      \If{$\varphi(N-1) \rightarrow \varphi(N)$ is invalid} \label{dpre:line:testdisj}
        \State {\bf throw} ``Unable to compute the difference pre-condition"; \label{dpre:line:throw2}
      \EndIf
      \State ${\partial \varphi(N)} := \widehat{\varphi}(N)$; \label{dpre:line:base2}
    \ElsIf{$\varphi(N)$ is of the form $\varphi^1(N)$ $\wedge$ $\cdots$ $\wedge$ $\varphi^k(N)$}
      \State ${\partial \varphi(N)} := $\textsc{SyntacticDiff}({$\varphi^1(N)$}) $\boldsymbol{\vee}$ $\cdots$ $\boldsymbol{\vee}$ \textsc{SyntacticDiff}({$\varphi^k(N)$}); \label{dpre:line:conjunct}
    \ElsIf{$\varphi(N)$ is of the form $\varphi^1(N)$ $\vee$ $\cdots$ $\vee$ $\varphi^k(N)$}
      \State ${\partial \varphi(N)} := $\textsc{SyntacticDiff}({$\varphi^1(N)$}) $\vee$ $\cdots$ $\vee$ \textsc{SyntacticDiff}({$\varphi^k(N)$}); \label{dpre:line:disjunct}
    \Else \label{dpre:line:default}
      \State ${\partial \varphi(N)} := \true$; \label{dpre:line:dptrue1}
    \EndIf
    \If{$\PP_{N-1}$ updates scalars or array elements in $\partial \varphi(N)$} \label{dpre:line:case2a}
      \State ${\partial \varphi(N)} := \true$; \label{dpre:line:dptrue2}
    \EndIf
    \State \Return $\partial \varphi(N)$;
  \end{algorithmic}
\end{algorithm*}

%% file: algorithms.tex
In this section, we discuss the algorithms for {\em full-program
induction}.  The algorithm primarily focuses on generation of the
three crucial components: \emph{difference program}
${\partial \PP_N}$, \emph{difference pre-condition}
${\partial \varphi(N)}$, and the formula $\ppre(N)$ for strengthening
pre- and post-conditions.  We have already seen the computation of the
difference program and the difference pre-condition in
Sect. \ref{sec:diffcomp}.  Before describing the algorithm for
full-program induction, however, we present the strategy for computing
the formula $\ppre(N)$.

\input{weakest-pre}  

\input{fpi}  

\input{fpi-decompose}  

%% file: weakest-pre.tex
\subsection{Generating the Formula $\mathbf{\ppre(N-1)}$}
\label{sec:wp}

We use Dijkstra's weakest pre-condition computation to obtain
$\ppre(N-1)$ after the difference pre-condition
${\partial \varphi(N)}$ and the difference program
${\partial \PP_N}$ have been generated.  The weakest pre-condition can
always be computed using quantifier elimination engines in
state-of-the-art SMT solvers like Z3 if ${\partial \PP_N}$ is
loop-free.  In such cases, we use a set of heuristics to simplify the
calculation of the weakest pre-condition before harnessing the power
of the quantifier elimination engine.  If ${\partial \PP_N}$ contains
a loop, it may still be possible to obtain the weakest pre-condition
if the loop doesn't affect the post-condition.  Otherwise, we compute
as much of the weakest pre-condition as can be computed from the
non-loopy parts of ${\partial \PP_N}$, and then try to recursively
solve the problem by invoking full-program induction on
${\partial \PP_N}$ with appropriate pre- and post-conditions.

\input{ex-ss-wp}  

\begin{example}
We apply Dijkstra's weakest pre-condition computation on the Hoare
triple from our running example in Fig. \ref{fig:ss-diff-simp}.  The
Hoare triple in Fig. \ref{fig:ss-wp} shows the difference
pre-condition $\partial \varphi(N)$, post-condition $\psi(N)$ and the
formula $\psi(N-1)$ from the induction hypothesis as well as the
strengthened pre- and post-condition formulas.  The first application
of weakest pre-condition computation generates the pre-condition {\tt
  A1[N-1] = N+1} on array {\tt A1}.  This is lifted to the quantified
form $\forall i$ $\in [0,\mathtt{N})\;$ $\mathtt{A1}[i] = \mathtt{N} +
  1$ in a natural way and is used to strengthen the post-condition.
  We substitute $N$ with $N-1$ and rename the array to get the formula
  $\forall i$ $\in [0,\mathtt{N})\;$ $\mathtt{A1\_Nm1}[i] =
    \mathtt{N}$, which is used to strengthen the pre-condition.
    Re-applying weakest pre-condition computation generates the
    predicates on $\mathtt{S}$ and $\mathtt{S\_Nm1}$ that further
    strengthen the pre- and post-condition as shown in
    Fig. \ref{fig:ss-wp}.
\end{example}

%% file: ex-ss-wp.tex

\begin{figure}[h]
\begin{alltt}
// assume(A[N-1] = 1)            //\(\partial\varphi\)(N)
// assume(S1_Nm1 = (N-1)\(\times\)(N+1)) //\(\psi\)(N-1)
// assume(\(\forall\)i\(\in\)[0,N-1) A1_Nm1[i] = N)
// assume(S_Nm1 = N-1)

1.  S = S_Nm1 + A[N-1];

2.  for(i=0; i<N-1; i++) \{
3.    A1[i] = A1_Nm1[i] + 1;
4.  \}
5.  A1[N-1] = A[N-1] + S;

6.  S1 = S1_Nm1 + A[N-1];
7.  S1 = S1 + (N-1);
8.  S1 = S1 + A1[N-1];

// assert(S1 = N\(\times\)(N+2))         //\(\psi\)(N)
// assert(\(\forall\)i\(\in\)[0,N) A1[i] = N+1)
// assert(S = N)
\end{alltt}
\caption{Strengthening the pre- and post-conditions}
\label{fig:ss-wp}
\end{figure}


%% file: fpi.tex
\subsection{Verification by Full-program Induction}
\label{sec:fpi}

\input{algo-fpi.tex}  

The basic version of full-program induction algorithm is presented as
routine \textsc{FPIVerify} in Algorithm \ref{alg:fpi}.  The important
steps of Algorithm \ref{alg:fpi} include checking conditions 3(a),
3(b) and 3(c) of Theorem \ref{thm:full-prog-ind-sound} (lines
\ref{line:fpi:base1}, \ref{line:fpi:base2} and \ref{line:fpi:ind}
resp.), calculating the weakest pre-condition of the relevant part of
the post-condition (line \ref{line:fpi:wp}), recursively invoking our
routine \textsc{FPIVerify} with the strengthened pre- and
post-conditions (line \ref{line:fpi:ret-recursive}), and accumulating
the the weakest pre-condition predicates thus calculated for
strengthening the pre- and post-conditions (line \ref{line:fpi:cpre}).
We now discuss the algorithm in detail.

We first check the base case of the analysis in line
\ref{line:fpi:base1}.  The base case of our induction reduces to
checking the validity of a Hoare triple of a loop-free program.  This
is achieved by compiling the pre-condition, program and post-condition
into a first-order logic formula.  The validity of the formula can be
checked with an off-the-shelf back-end SMT solver like Z3.  If the
check fails, we have found a valid counter-example that is reported to
the user in line \ref{line:fpi:print-ce}, and the algorithm terminates
in line \ref{line:fpi:ret-base}.

Next, we rename the variables and arrays in the program $\PP_N$ as
well as the pre- and post-conditions (as described in Sect.
\ref{sec:renaming}) and collect the set of glue nodes (line
\ref{line:fpi:rename}).  Then, in line \ref{line:fpi:syntdiff}, we
compute the difference pre-condition $\partial \varphi(N)$ using
function \textsc{SyntacticDiff} (described in Sect.
\ref{sec:diff-pre}).  We then compute the difference program $\partial
\PP_N$, in line \ref{line:fpi:progdiff}, using function
\textsc{ProgramDiff} from Sect. \ref{sec:diff-prog}.  Note that this
function can compute the difference program when the scalar
variables/arrays of interest are identified as affected.  In line
\ref{line:fpi:progdiffsimp}, we simplify the statements in the
computed difference program, remove redundant statements and try to
accelerate loops, if any, using function \textsc{SimplifyDiff} from
Algorithm \ref{alg:simpdiff}.

The do-while loop in lines \ref{line:fpi:doloop}--\ref{line:fpi:base2}
iteratively checks if the given assertion can be proved.  Once the
base case succeeds, we check the inductive step in line
\ref{line:fpi:ind}.  If the loop terminates via the {\tt return}
statement in line \ref{line:fpi:ret-ind}, then the inductive claim has
been successfully proved.  Otherwise, in line \ref{line:fpi:wp}, we
compute Dijkstra's weakest pre-condition using the formula
$\ppre_i(N)$, over the difference program.  The formula $\ppre_i(N)$
is initialized to $\psi(N)$ in line \ref{line:fpi:wp-init}.  We denote
the computed weakest pre-condition as $\ppre_i(N-1)$.  Note that, the
formula $\ppre_i(N-1)$ strengthens the pre-condition and the same
formula $\ppre_i(N)$, but with $N$ substituted for $N-1$, strengthens
the post-condition.  The variable $c\_\ppre_i(N-1)$, initialized to
$\true$ in line \ref{line:fpi:cumu-wp}, accumulates weakest
pre-condition formulas from each loop iteration (line
\ref{line:fpi:cpre}).

In case no further weakest pre-conditions can be generated, checked in
line \ref{line:fpi:no-wp}, we recursively invoke \textsc{FPIVerify} on
${\partial \PP_N}$ in line \ref{line:fpi:ret-recursive}.  Prior to the
recursive invocation, we check if it will be beneficial using function
\textsc{CheckProgress}, in line \ref{line:fpi:check-progress}.
Discussion about \textsc{CheckProgress} is deferred to
Sect. \ref{sec:progress}.  The recursive invocation helps in
situations where the computed difference program ${\partial \PP_N}$
has loops.  To present an example of this scenario, we modify the
program in Fig. \ref{fig:ex} by having the statement {\tt C[t3] = N;}
(instead of {\tt C[t3] = 0;}) in line $10$.  In this case, ${\partial
  \PP_N}$ retains a loop that rectifies the value of {\tt C[t3]}
corresponding to its computation in the third loop in
Fig. \ref{fig:ex}.  The recursive invocation of full-program induction
on ${\partial \PP_N}$ as input for the example described here will
result in a loop-free difference program.  If the check in line
\ref{line:fpi:check-progress} reports that further application of
full-program induction will not yield any benefits then we report the
failure of our technique in line \ref{line:fpi:ret-prog-inconc}.

When weakest pre-condition computation succeeds, we conjoin the
computed strengthening predicate $\ppre_i(N)$ with the variable
$c\_\ppre_{i-1}(N)$ in line \ref{line:fpi:cpre}.  Since the weakest
pre-condition ($\ppre_i(N-1)$ in line \ref{line:fpi:wp}) computed in
every iteration of the loop is conjoined to strengthen the inductive
pre-condition ($c\_\ppre_i(N-1)$ in line \ref{line:fpi:cpre}), it
suffices to compute the weakest pre-condition of $\ppre_{i-1}(N)$
(instead of $c\_\ppre_i(N) \wedge \psi(N)$) in line \ref{line:fpi:wp}.
Possibly multiple iterations of the loop in lines
\ref{line:fpi:doloop}--\ref{line:fpi:base2} are required to strengthen
the pre- and post-conditions.  After each iteration, the base case is
checked again in line \ref{line:fpi:base2} with the strengthened pre-
and post-conditions.  If the loop terminates due to violation of the
base-case with the strengthened post-condition (line
\ref{line:fpi:base2}), we report the failure of our method by
returning $\false$ in line \ref{line:fpi:ret-inconc}.

\begin{lemma}  \label{lemma:fpi-alg}
  Upon successful termination, if function \textsc{FPIVerify} returns
  $\true$, then $\{\varphi_N\}$ $\;\PP_N\;$ $\{\psi_N\}$ holds for all
  $N \ge 1$.
\end{lemma}
\begin{proof}
  Verifying the given Hoare triple requires establishing the
  conditions mentioned in Theorem \ref{thm:full-prog-ind-sound}.  The
  functions \textsc{ProgramDiff} invoked in line
  \ref{line:fpi:progdiff} and \textsc{SimplifyDiff} invoked in line
  \ref{line:fpi:progdiffsimp} ensure condition $1$ of Theorem
  \ref{thm:full-prog-ind-sound} (refer Theorem
  \ref{lemma:diff-gen-sound} and Lemma \ref{lemma:simp-sound}).  The
  call to \textsc{SyntacticDiff} in line \ref{line:fpi:syntdiff} in
  \textsc{FPIVerify} computes the difference pre-conditions that
  satisfy conditions $2$(a) and $2$(b) (refer Lemma
  \ref{lemma:diff-pre-condition}).  The conditions $3$(a) and $3$(b)
  of Theorem \ref{thm:full-prog-ind-sound} are checked in lines
  \ref{line:fpi:base1} and \ref{line:fpi:base2} respectively.  The
  check in line \ref{line:fpi:ind} ensures that the return statement
  in line \ref{line:fpi:ret-ind} executes only when condition $3$(c)
  of Theorem \ref{thm:full-prog-ind-sound} is ensured.  Similarly, the
  statement in line \ref{line:fpi:ret-recursive} returns $\true$ only
  if the recursive call to \textsc{FPIVerify} proves all conditions in
  Theorem \ref{thm:full-prog-ind-sound}.  Hence, we conclude that
  $\{\varphi_N\}$ $\;\PP_N\;$ $\{\psi_N\}$ holds for all $N \ge 1$.
\end{proof}


%% file: algo-fpi.tex
\begin{algorithm*}[!t]
  \caption{\textsc{FPIVerify}({$\PP_N$}: program, {$\varphi(N)$}: pre-condition, {$\psi(N)$}: post-condition)}
  \label{alg:fpi}
  \begin{algorithmic}[1]
    \If{Base case check \{$\varphi(1)$\} $\PP_1$ \{$\psi(1)$\} fails}  \label{line:fpi:base1}
      \State {\bf print} ``Counterexample found!'';  \label{line:fpi:print-ce}
      \State \Return $\false$; \label{line:fpi:ret-base}
    \EndIf

    \State $\langle \PP_N, \varphi(N), \psi(N), \gluenodes \rangle$ := \textsc{Rename}($\PP_N$, $\varphi(N)$, $\psi(N)$);  \label{line:fpi:rename}  \Comment{Renaming as described in Sect. \ref{sec:renaming}}
    \State $\partial \varphi(N)$ := \textsc{SyntacticDiff}($\varphi(N)$); \label{line:fpi:syntdiff}

    \State $\partial \PP_N$ := \textsc{ProgramDiff}($\PP_N$, $\gluenodes$); \label{line:fpi:progdiff}
    \State $\partial \PP_N$ := \textsc{SimplifyDiff}($\partial \PP_N$); \label{line:fpi:progdiffsimp} \Comment{Simplify and Accelerate loops}

    \State $i := 0$;
    \State $\ppre_i(N) := \psi(N)$; \label{line:fpi:wp-init}
    \State $c\_\ppre_i(N) := \true$; \label{line:fpi:cumu-wp}	\Comment{Cumulative conjoined pre-condition}

    \Do
      \label{line:fpi:doloop}
      \If{ \{$c\_\ppre_i(N-1) \wedge \psi(N-1) \wedge \partial \varphi(N)$\} $\partial\PP_N$ \{$c\_\ppre_i(N) \wedge \psi(N)$\} }  \label{line:fpi:ind}
      \State \Return $\true$;  \label{line:fpi:ret-ind}  \Comment{Assertion verified}
      \EndIf
      \State $i := i+1$;
      \State $\ppre_i(N-1) := \textsc{LoopFreeWP}( \ppre_{i-1}(N), \partial\PP_N)$;  \label{line:fpi:wp}  \Comment{Dijkstra's $\mathsf{WP}$ sans $\mathsf{WP}$-for-loops}
      \If {no new $\ppre_i(N-1)$ obtained} \label{line:fpi:no-wp}  \Comment{Can happen if ${\partial \PP_N}$ has a loop}
        \If {\textsc{CheckProgress}($\PP_N$, $\partial \PP_N$)}  \label{line:fpi:check-progress}
          \State \Return \textsc{FPIVerify}(${\partial \PP_N}$, $c\_\ppre_{i-1}(N-1) \wedge \psi(N-1) \wedge \partial \varphi(N)$, $c\_\ppre_{i-1}(N) \wedge \psi(N)$);  \label{line:fpi:ret-recursive}
        \Else
          \State \Return $\false$; \Comment{Failed to prove by full-program induction} \label{line:fpi:ret-prog-inconc}
        \EndIf
      \Else
         \State $c\_\ppre_i(N) := c\_\ppre_{i-1}(N) \wedge \ppre_i(N)$;  \label{line:fpi:cpre}
      \EndIf
    \doWhile{Base case check \{$\varphi(1)$\} $\PP_1$ \{$c\_\ppre_i(1)$\} passes}; \label{line:fpi:base2}
    \State \Return $\false$; \Comment{Failed to prove by full-program induction} \label{line:fpi:ret-inconc}
  \end{algorithmic}
\end{algorithm*}

%% file: fpi-decompose.tex
\subsection{Generalized FPI Algorithm}
\label{sec:fpi-ext}

\input{algo-fpi-decompose.tex}  

While the algorithm \textsc{FPIVerify} suffices for all of our
experiments, it may not always be the case.  Specifically, even if
${\partial \PP_N}$ is loop-free, the analysis may exit the loop in
lines \ref{line:fpi:doloop}--\ref{line:fpi:base2} of
\textsc{FPIVerify} by violating the base case check in line
\ref{line:fpi:base2}.  To handle (at least partly) such cases, we
propose the following strategy. Whenever a (weakest) pre-condition
$\ppre_i(N-1)$ is generated, instead of using it directly to
strengthen the current pre- and post-conditions, we ``decompose'' it
into two formulas $\ppre_i'(N-1)$ and ${\partial \varphi_i'(N)}$ with
a two-fold intent: (a) potentially weaken $\ppre_i(N-1)$ to
$\ppre_i'(N-1)$, and (b) potentially strengthen the difference formula
${\partial \varphi(N)}$ to ${\partial \varphi_i'(N)} \wedge {\partial
  \varphi(N)}$.  The checks for these intended usages of
$\ppre_i'(N-1)$ and ${\partial \varphi_i'(N)}$ are implemented in
lines \ref{line:decomp:conda}, \ref{line:decomp:condb},
\ref{line:decomp:condc}, \ref{line:decomp:ind} and
\ref{line:decomp:base} of routine \textsc{FPIDecomposeVerify}, shown
as Algorithm \ref{alg:fpi-ext}.  This routine is meant to be invoked
as \textsc{FPIDecomposeVerify}$(i)$ after each iteration of the loop
in lines \ref{line:fpi:doloop}--\ref{line:fpi:base2} of routine
\textsc{FPIVerify} (so that $\ppre_i(N)$, $c\_\ppre_i(N)$ etc. are
initialized properly).  In general, several ``decompositions'' of
$\ppre_i(N)$ may be possible, and some of them may work better than
others.  \textsc{FPIDecompseVerify} permits multiple decompositions to
be tried through the use of the functions $\textsc{NextDecomposition}$
and $\textsc{HasNextDecomposition}$.  The meaning of both these functions
is intuitive from their names.  Lines
\ref{line:decomp:recursive1}--\ref{line:decomp:recursive2} of
\textsc{FPIDecomposeVerify} implement a simple back-tracking strategy,
allowing a search of the space of decompositions of $\ppre_i(N-1)$.
Observe that when we use \textsc{FPIDecomposeVerify}, we
simultaneously compute a difference formula $({\partial \varphi'_i(N)}
\wedge {\partial \varphi(N)})$ and an inductive pre-condition
$(c\_\ppre_{i-1}(N) \wedge \ppre_i'(N))$.

\begin{lemma} \label{lemma:fpi-dv-alg}
  Upon successful termination, if function \textsc{FPIDecomposeVerify}
  returns $\true$, then $\{\varphi_N\}$ $\;\PP_N\;$ $\{\psi_N\}$ holds
  for all $N \ge 1$.
\end{lemma}
\begin{proof}
  The conditions mentioned in Theorem \ref{thm:full-prog-ind-sound}
  are a pre-requisite to verifying the given Hoare triple.  Condition
  $1$ of Theorem \ref{thm:full-prog-ind-sound} is ensured by
  difference computation (functions \textsc{ProgramDiff} and
  \textsc{SimplifyDiff}) in \textsc{FPIVerify}.  Conditions $2$(a) and
  $2$(b) are established in \textsc{FPIVerify} (via the call to
  function \textsc{SyntacticDiff}) and the checks on lines
  \ref{line:decomp:conda}--\ref{line:decomp:condc} in
  \textsc{FPIDecomposeVerify} ensure that these conditions continue to
  hold.  Further, \textsc{FPIVerify} also ensures conditions $3$(a)
  and $3$(b) before it invokes \textsc{FPIDecomposeVerify}.  Now, the
  check in line \ref{line:decomp:ind} in \textsc{FPIDecomposeVerify}
  ensures condition $3$(c) of Theorem \ref{thm:full-prog-ind-sound}.
  Similarly, the statement in line \ref{line:decomp:ret-recursive} in
  \textsc{FPIDecomposeVerify} returns $\true$ only if the recursive
  call to \textsc{FPIDecomposeVerify} proves all the conditions in
  Theorem \ref{thm:full-prog-ind-sound}.  Hence, we conclude that
  $\{\varphi_N\}$ $\;\PP_N\;$ $\{\psi_N\}$ holds for all $N \ge 1$.
\end{proof}

%% file: algo-fpi-decompose.tex
\begin{algorithm*}[!t]
  \caption{\textsc{FPIDecomposeVerify}( i : integer )}
  \label{alg:fpi-ext}
  \begin{algorithmic}[1]
    \Do
      \State $\langle\ppre_i'(N-1), \partial \varphi_i'(N)\rangle$ := $\textsc{NextDecomposition}(\ppre_i(N-1))$;
      \State Check if (a) $\partial \varphi_i'(N) \wedge \ppre_i'(N-1)  \rightarrow  \ppre_i(N-1)$, \label{line:decomp:conda} \\
      \hspace*{0.75in}(b) $\varphi(N) \rightarrow \varphi(N-1) \odot \left(\partial \varphi'_i(N) \wedge \partial \varphi(N)\right)$,  \Comment{where $\odot \in \{\wedge, \vee\}$} \label{line:decomp:condb} \\
      \hspace*{0.75in}(c) $\PP_{N-1}$ does not update any variable or array element in $\partial \varphi_i'(N)$ \label{line:decomp:condc}
      \If{any check in lines \ref{line:decomp:conda} - \ref{line:decomp:condc} fails}
      \If{$\textsc{HasNextDecomposition}(\ppre_i(N-1))$}
      \State \textbf{continue};
      \Else
      \State \Return $\false$;
      \EndIf
      \EndIf
      
      \If{\{$c\_\ppre_{i-1}(N-1) \wedge \psi(N-1) \wedge \ppre_i(N-1) \wedge \partial \varphi(N)$\} $\partial\PP_N$ \{$c\_\ppre_{i-1}(N) \wedge \psi(N) \wedge \ppre_i'(N)$\}}  \label{line:decomp:ind}
        \State \Return $\true$; \Comment{Assertion verified}
      \Else
        \State $c\_\ppre_i(N) := c\_\ppre_{i-1}(N) \wedge \ppre_i'(N)$;
        \State $\ppre_{i+1}(N-1) := \textsc{LoopFreeWP}( \ppre_i'(N), \partial\PP_N)$; \Comment{Dijkstra's $\mathsf{WP}$ sans $\mathsf{WP}$-for-loops}

        \If {\{$\varphi(1)$\} $\PP_1$ \{$c\_\ppre_i(1) \wedge \ppre_{i+1}(1)$\} holds}  \label{line:decomp:base}
           \State $prev\_\partial \varphi(N)$ := $\partial \varphi(N)$;
           \State $\partial \varphi(N)$ := $\partial \varphi'_i(N) \wedge \partial \varphi(N)$;
           \If{\textsc{FPIDecomposeVerify}$(i+1)$} \Comment{Recursive invocation}  \label{line:decomp:recursive1}
             \State \Return $\true$; \Comment{Assertion verified}  \label{line:decomp:ret-recursive}
           \Else
             \State $\partial \varphi(N)$ := $prev\_\partial \varphi(N)$;  \label{line:decomp:recursive2}
          \EndIf
        \EndIf
        \State $i := i+1$;
      \EndIf
      \doWhile{$\textsc{HasNextDecomposition}(\ppre_i(N-1))$};

      \State \Return $\false$;  \Comment{Failed to prove by full-program induction}
    \end{algorithmic}
\end{algorithm*}

%% file: progress.tex
Recall from Sect.~\ref{sec:fpi} that given the parameterized Hoare
triple $\{\varphi(N)\}$ $\PP_{N}$ $\{\psi(N)\}$, our technique
recursively computes difference programs until the given
post-condition $\psi(N)$ is proved.  The difference computation must
eventually result in programs $\partial \PP_N$ that can be easily
verified without the need of further applying inductive reasoning or
indicate otherwise.  In this section, we define a progress measure
that can be used to check if the difference computation will
eventually simplify the program to the extent that it can be verified
using a back-end SMT solver.  The measure is based on the
characteristics of the difference programs computed by our technique.

Ranking functions have been traditionally used to show program
termination~\cite{rank-fun1,rank-fun2,termination,lexico-cite}.  We
use the notion of ranking functions to measure the progress that our
technique has made towards verifying the given post-condition using
the difference programs.  Several different criteria have been used in
the literature to define ranking functions.  Our ranking function
links with each difference program a value from a well-founded domain.
We assign the minimal rank to programs that can be effectively
verified, for example using a back-end SMT solver.  A difference
program gets a smaller rank compared to another difference program if
it is ``closer'' (in a natural way) to programs that can be proved.
Here, we list some criteria that can be used for defining the ranking
function for our technique based on the syntactic changes in
difference programs vis-a-vis the given program.


The main hurdle in proving the given Hoare Triple, $\{\phi(N)\}$
$\;\PP_N\;$ $\{\psi(N)\}$, are the loops in the given program $\PP_N$.
Once the difference program is loop-free, the post-condition in such
programs can be easily verified by an SMT solver and our technique is
no longer required to recursively apply induction any further for
proving such programs.  Thus, the difference programs for which our
technique terminates are loop-free programs and programs in which
loops can be accelerated or optimized away with known techniques.
Hence, reduction in the number of loops in the difference program
$\partial \PP_N$ vis-a-vis the given program $\PP_N$ is the main
criterion to measure progress in our technique.

Further, the difference computation can potentially reduce the
dependence on the value of $N$.  For programs with expressions that do
not directly or indirectly \footnote{By indirect dependence, we mean
the dependence via another value computed in a peeled or non-peeled
statements in the program.} rely on $N$, the difference program
consists of only the peeled iterations of loops.  Clearly, when the
difference program $\partial \PP_N$ is impervious to the value of $N$,
additional code to rectify the values of variables is no longer
required in the subsequent recursive invocations.  This indicates that
we have made progress.  We thus use the presence of variables in the
program whose value directly or indirectly depends on the value of $N$
as another criteria to measure progress.  As previously stated in
Sect.~\ref{sec:affected}, if the value of a variable/array depends on
$N$ or on a value computed in a peeled statement, then we call such
variables/arrays as \emph{affected variables/arrays}.  Our technique
computes the set of affected variables during each recursive attempt
to verify the post-condition.  The difference program must rectify the
values of these affected variables/arrays.  When the difference
program has no affected variables, the verification attempt can be
terminated after the next invocation of our technique.

We also consider the complexity of expressions in the program and use
it as a criteria for measuring progress.  For every expression
appearing in assignment statements, its expression complexity can be
defined in many possible ways.  Once this complexity is defined for
expressions, we can take the maximum complexity of all the expressions
as the expression complexity of the entire program.  For programs with
polynomial expressions, we can use the highest degree of the affected
variables/arrays in the expression as the expression complexity.
Similarly, several other criteria can be used to define the expression
complexity.  These include nesting levels of array indices,
size/weight of the expression trees in $\partial \PP_N$ vis-a-vis
$\PP_N$, number of variables, operators and constants in the
expressions and so on.  It is worth pointing out that such notions
have been previously studied in term rewriting
systems~\cite{simporder,decomporder}.

Note that each criterion discussed so far, including the number of
loops, the number of affected variables and the expression complexity,
is well-founded.  Hence, the domain of values represented by their
Cartesian product is also well-founded and represents a lexicographic
ordering on the difference programs computed by our method.  Progress
is guaranteed if each recursive invocation of our technique in the
cycle reduces this measure assigned by such a ranking function.  We
argue that the cycle of recursive invocations to our technique must
eventually terminate, as there are no infinite descending chains of
elements in the well-founded domain.  We present an algorithm that can
compute values from this domain on the fly and return the result of
the comparison between the computed quantities.  Note that no user
intervention is required for checking progress.

\begin{algorithm}[!h]
  \caption{\textsc{CheckProgress}({$\PP_N$}: program, {$\partial \PP_N$}: difference program)}
  \label{alg:progress}
  \begin{algorithmic}[1]
    \State $\mathsf{LoopList}$ := \textsc{Loops}($\PP_N$);
    \State $\mathsf{LoopList'}$ := \textsc{Loops}($\partial \PP_N$);
    \If{\#$\mathsf{LoopList'}$ $<$ \#$\mathsf{LoopList}$} \label{algline:lpcmp}
      \State \Return $\true$;
    \EndIf
    \State $\affectedvars$ := \textsc{ComputeAffected}($\PP_N$);
    \State $\affectedvars'$ := \textsc{ComputeAffected}($\partial \PP_N$);
    \If{\#$\affectedvars'$ $<$ \#$\affectedvars$} \label{algline:affcmp}
      \State \Return $\true$;
    \EndIf
    \State $\mathsf{EC}$ := \textsc{ExpressionComplexity}($\PP_N$);
    \State $\mathsf{EC'}$ := \textsc{ExpressionComplexity}($\partial \PP_N$);
    \If{$\mathsf{EC'}$ $<$ $\mathsf{EC}$} \label{algline:degcmp}
      \State \Return $\true$;
    \EndIf
    \State \Return $\false$;
  \end{algorithmic}
\end{algorithm}

The routine \textsc{CheckProgress} in Algorithm~\ref{alg:progress} is
used for checking progress after the difference program is computed.
The algorithm is based on the change in the number of loops, number of
affected variables and the expression complexity of the given program
$\PP_N$ vis-a-vis the difference program $\partial \PP_N$.  First, we
compute the number of loops in programs $\PP_N$ and $\partial \PP_N$.
We compare the number of loops in $\PP_N$ and $\partial \PP_N$ in
line~\ref{algline:lpcmp}.  If the difference program has fewer loops
than $\PP_N$, then we return $\true$ concluding that the
$\partial \PP_N$ is simpler to verify than the given program.  Note
that, we do not consider the glue loops in the difference program that
were introduced during the renaming step to copy values across
versions.  If the number of loops does not decrease in an invocation
of our technique, we check if the number of affected variables has
decreased.  We compute the set of affected variables in $\PP_N$ and
$\partial \PP_N$ using the routine \textsc{ComputeAffected} from
Algorithm~\ref{alg:pdg-affected}.  In line~\ref{algline:affcmp}, we
compare the number of affected variables in both the programs.  The
algorithm returns $\true$ if the difference program has fewer affected
variables than $\PP_N$.  Subsequently, we check if the expressions in
the difference program $\partial \PP_N$ are ``simpler'', and easier to
reason with, than $\PP_N$.  We assume the availability of a
routine \textsc{ExpressionComplexity} that can compute this complexity
measure for programs $\PP_N$ and $\partial \PP_N$.  In
line~\ref{algline:degcmp} we check if the expression complexity of the
difference program $\partial \PP_N$ is less than that of the given
program $\PP_N$, in which case the algorithm returns $\true$.  If none
of these criteria are met, then the algorithm returns $\false$.

\begin{lemma}  \label{lemma:alg-progress}
If \textsc{CheckProgress} in Algorithm~\ref{alg:progress} returns
$\true$, then the difference program $\partial \PP_N$ is ``simpler''
to verify (using the full-program induction technique) as compared to
the given program $\PP_N$.
\end{lemma}
\begin{proof}
The difference program $\partial \PP_N$ has strictly less loops than
$\PP_N$ when the check in line \ref{algline:lpcmp} is satisfied.  In
this case, verifying $\partial \PP_N$ is simpler than verifying
$\PP_N$.  Further, reduction in the number of affected
variables/arrays means less code is retained to rectify their values.
Hence, when the check on line \ref{algline:affcmp} is satisfied,
$\partial \PP_N$ is simpler than $\PP_N$.  By
Lemma \ref{lemma:diff-program-without-affected-vars}, when none of the
variables/arrays of interest are identified as affected, only the
peeled iterations of loops (referred as $\Peel{\PP_N}$) suffice as the
difference program $\partial \PP_N$.  This also makes verifying
$\partial \PP_N$ simpler as compared to $\PP_N$.  Similarly, the last
condition ensures that the expressions in the difference program are
easier to reason with than the given program $\PP_N$.  Further, these
characteristics of the program and the ordering among them as
specified by \textsc{CheckProgress} forms a lexicographic ranking
function \cite{lexico-cite}.  Hence, these quantities are bound to
reduce with each application of our technique, making the difference
program simpler to verify each time.  This concludes the lemma.
\end{proof}


\begin{lemma}  \label{lemma:alg-fpi-termination}
The routine \textsc{FPIVerify} in Algorithm~\ref{alg:fpi} eventual
  terminates.
\end{lemma}
\begin{proof}
Function \textsc{FPIVerify} presented in Algorithm \ref{alg:fpi} can
execute in infinite recursion only when the invocation
of \textsc{CheckProgress} in line \ref{line:fpi:check-progress}
returns $\true$ infinitely often.  From difference program
computation, we know that the number of loops and affected
variables/arrays in the difference program $\partial\PP_N$ never
increase beyond their counts in the given program $\PP_N$, they either
decrease or remain the same.  Further, if the expression complexity of
all the statements that update an affected variable/array does not
decrease then our method returns $\false$, and consequently we report
failure.  Thus, these three characteristics with the specified
ordering among them form a lexicographic ranking
function \cite{lexico-cite}.  Since the value of the lexicographic
ranking function strictly decreases in each recursive application of
our method, it ensures that function \textsc{FPIVerify} eventually
terminates.
\end{proof}

%% file: generalized.tex

For brevity and ease of explanation, we have presented our technique
in simple settings.  We have so far considered Hoare triples that have
a single parameter $N$.  In this section, we show how our technique
can be adapted to Hoare triples with multiple parameters as well as
peeling loops in different directions for our inductive reasoning.  We
also state the limitations of our technique.

Based on the ideas previously described, our technique can already
verify several interesting scenarios in programs.  Our technique can
verify programs that manipulate arrays of different sizes as well as
loops with non-uniform termination conditions that are a linear
function of $N$. It does so by computing a (possibly different) {\em
peel count} for each loop that manipulates different arrays.  For the
ease of presentation, our algorithm computes the rectified values of
variables/arrays in statements with a single operator.  When the
program statements have two or more operators, such statements can be
split into multiple statements, by introducing temporary variables
such that each statement has a single operator, and then computing the
difference program using our algorithm.

\paragraph{\bfseries Multiple independent program parameters.}
Consider proving Hoare triples with multiple parameters
$N_1,N_2,...,N_k$.  Suppose that the values of these parameters are
independent of each other.  Verifying Hoare triples for all values of
these parameters can be done by inducting on one program parameter at
a time while keeping the other parameters fixed.  We explain this with
the help of a simple example with two parameters. To prove that the
Hoare Triple $\{\varphi(N_1,N_2)\}$ $\;\PP_{N_1,N_2}\;$
$\{\psi(N_1,N_2)\}$ for all $N_1 \geq a \wedge N_2 \geq b$, we prove
the following three sub-goals.  First, in the base-case we prove that
the triple $\{\varphi(a,b)\}$ $\;\PP_{a,b}\;$ $\{\psi(a,b)\}$ holds.
Second, induction over the parameter $N_1$, where we assume the Hoare
Triple $\{\varphi(k,N_2)\}$ $\;\PP_{k,N_2}\;$ $\{\psi(k,N_2)\}$ holds
with $k \geq a \wedge N_2 \geq b$, and prove the Hoare Triple
$\{\varphi(k+1,N_2)\}$ $\;\PP_{k+1,N_2}\;$ $\{\psi(k+1,N_2)\}$,
treating $N_2$ as a symbolic parameter unchanged during the induction.
Third, induction over the parameter $N_2$, where we assume that the
Hoare Triple $\{\varphi(N_1,l)\}$ $\;\PP_{N_1, l}\;$ $\{\psi(N_1,l)\}$
holds where $N_1 \geq a \wedge l \geq b$, and prove the Hoare Triple
$\{\varphi(N_1, l+1)\}$ $\;\PP_{N_1, l+1}\;$ $\{\psi(N_1, l+1)\}$,
treating $N_1$ as a symbolic parameter unchanged in the induction.
This can be easily extended to Hoare triples with more than two
parameters.  For programs that manipulate arrays of different
independent sizes, we treat each variable representing the symbolic
size of arrays as a parameter.  As described above, our technique
verifies such programs by inducting on each parameter one at a time.

\paragraph{\bfseries Different direction of peeling loops.}
Recall that the difference program $\partial \PP_N$ is sequentially
composed with $\PP_{N-1}$.  Earlier, we have been peeling the last
iterations of the loops in $\PP_{N}$ so that $\PP_{N-1}$ and $\PP_N$
have the same number of iterations in each loop.  However, there are
programs where peeling the last iterations of loops may not be
possible such that our technique can compute a difference program.  In
such cases, we may need to peel the initial iterations of the loops in
the program.  As an example, consider a loop where the value of the
loop counter decreases in each iteration.  A possible way is to rotate
these loops to fit the template of loops defined in our grammar and
then apply our technique.  However, not all loops are such that they
can be rotated easily using the standard loop transformation
techniques.  For such loops, we may need to peel it at the beginning.
We peel the initial iterations of these loops and add the code that
rectifies values of variables computed in the loop after the peeled
iterations such that $\PP_{N-1} ; \partial \PP_N$ is semantically
equivalent to $\PP_N$.  Thus, by peeling initial iterations of loops
when computing the difference program, our technique can be easily
adapted to programs with such loops in a sound way.

\input{limitations}  

%% file: limitations.tex
\subsection{Limitations}
\label{sec:limitations}

There are several scenarios under which the full-program induction
technique may not produce a conclusive result.

Program computation with side-effects may make it difficult to compute
the difference program such that there is a clear separation between
the program $\PP_{N-1}$ and the rest of the computation that can make
up $P_N$.  Computation that results in side-effects includes I/O
operations, allocation, de-allocation and modification of heap memory
and other operations that modify the environment which is not local to
the given program.  When the given program is not free of such
side-effects, our technique may not be able to decompose it into
$\PP_{N-1}$ and $\partial\PP_N$.  Note that we only disallow the
computation that impacts the post-condition to be proved.  In our
experience, a large class of array manipulating programs are naturally
free of side-effects.  In particular, the programs discussed in this
paper (including Fig. \ref{fig:ex}) and those used for experimentation
are free of side-effects.

Currently our technique is unable to verify programs with branch
conditions that are dependent on the parameter $N$.  Computing the
difference program becomes cumbersome in such cases.  This stems from
the fact that the branch condition may evaluate to different outcomes
in $\PP_N$ and $\PP_{N-1}$, for the same value of $N$, and hence, may
require us to compute the difference of two arbitrary pieces of code
blocks.  We identify such cases while computing the difference program
in Algorithm \ref{alg:diff-generation} and suspend our verification
attempt on line \ref{line:diff:branch} of the routine
\textsc{NodeDiff}.  Note that this does not include the loop
conditions, which are handled by peeling the loop.  To illustrate this
case, consider the Hoare triple shown in Fig. \ref{fig:ex-challenge}.
The first loop in the program initializes array {\tt A} and the second
loop updates array {\tt A} within a branch statement with the
conditional expression {\tt N\%2 == 0}.  It is easy to see that this
branch condition will evaluate to different outcomes in $\PP_N$ and
$\PP_{N-1}$.  As a result, it is difficult for our technique to
compute a difference program.  Invariant generation techniques may be
better suited for verifying this example.  The weakest loop invariants
needed to prove the post-condition in this example are: $\forall j \in
[0,i)\; (A[j] = 0)$ for the first loop and $\forall k \in [0,i)\;
    (A[k]\%2 = N\%2)$ for the second loop.

\input{ex-challenge}

The difference program includes all peeled iterations of $\PP_N$ that
are missed in $\PP_{N-1}$.  Hence, our technique needs to know the
symbolic upper bound on the value of the loop counter to be able to
compute the number of iterations to be peeled from the program.
Further, when programs have loops with non-linear termination
conditions, the construction of the difference program becomes
challenging.  The number of peeled iterations itself may be a function
of $N$ and possibly result in a loop in the difference program.  For
example, consider a loop in $\PP_N$ with the counter $i$ initialized
to $0$ and the loop termination condition ``$i < N^2$''.  The
corresponding loop in $\PP_{N-1}$ has the same initialization but the
termination condition is ``$i < (N-1)^2$''.  ``$2 \times N +1$''
iterations must be peeled from this loop.  For such loop conditions,
an entire loop appears as the peel in the difference program.  Since
the number of iterations to be peeled is not a constant number,
currently while computing this peel (in line \ref{peel:line:fail} of
Algorithm \ref{alg:peelloops}), our technique reports a failure to
handle such programs.  Further, our grammar restricts the shape of
loops that can be verified using our technique.  Most notably, we
analyze programs with only non-nested loops.  We have designed a
variant of the full-program induction technique \cite{diffy-cav21}
that can verify a class of programs with nested loops. The technique
greatly simplifies the computation of difference programs.  It infers
and uses relations between two slightly different versions of the
program during the inductive step.  We refer the interested reader to
\cite{diffy-cav21}.

The inductive reasoning may remain inconclusive when the rank of the
difference programs, as defined in Sect. \ref{sec:progress}, does not
reduce during the successive invocations to verify the post-condition
using our technique.  Continuous reduction in the rank/progress
measure is crucial to the success of full-program induction.  When no
progress is observed, we suspend the verification attempt in line
\ref{line:fpi:ret-prog-inconc} of the routine \textsc{FPIVerify} in
Algorithm \ref{alg:fpi}.  Though the ranking functions can be defined
in many possible ways, there are programs that pose a challenge in
computing the difference program in a way that the rank of the
computed difference program does not reduce.  However, such programs
are rarely seen in practice.

Our technique may fail to verify a correct program if the heuristics
used for weakest pre-condition either fail or return a pre-condition
that causes violation of the base-case checked on
line~\ref{line:fpi:base2} of the routine \textsc{FPIVerify} in
Algorithm \ref{alg:fpi}.

Apart from the conceptual limitations mentioned above, our prototype
implementation has a few limitations.  We currently support
expressions in assignment statements with only $\{ +, -, \times, \div
\}$ operators.  In the implementation we support a single program
parameter and peel only the last iterations of loops.  Despite all
these limitations, our experiments show that full-program induction
performs remarkably well on a large suite of benchmarks.

%% file: ex-challenge.tex
\begin{figure}[h]
\begin{alltt}
// assume(\(\true\))

1.  for(i=0; i<N; i++) \{
2.    A[i] = 0;
3.  \}

5.  for(i=0; i<N; i++) \{
6.    if( N%2 == 0 ) \{
7.      A[i] = A[i] + 2;
8.    \} else \{
9.      A[i] = A[i] + 1;
10.   \}
11. \}

// assert(\(\forall\)i \(\in\) [0,N), A[i]%2 = N%2)
\end{alltt}
\caption{Challenge example}
\label{fig:ex-challenge}
\end{figure}

%% file: experiments.tex
In this section, we present an extensive experimental evaluation of
the {\emph{full-program induction} technique on a large set of array
manipulating benchmarks.

\subsection{Implementation}

We have implemented our technique in a prototype tool called
{\ourtool}.  Our tool and the benchmarks used in the experiments are
publicly available at \cite{vajra-artifact}.  {\ourtool} takes a C
program in SV-COMP format as input.  The tool, written in {\tt C++},
is built on top of the LLVM/CLANG \cite{clang} $6.0.0$ compiler
infrastructure.  We use CLANG front-end to obtain LLVM bitcode.
Several normalization passes such as constant propagation, dead code
elimination, static single assignment (SSA) generation for renaming
variables and arrays, loop normalization for running loop-dependent
passes that identify program constructs such as loop counter, lower
bound and upper bound expressions, branch conditions and so on, are
performed on the bitcode.  We use {\zthree} \cite{z3} v$4.8.7$ as the
SMT solver to prove the validity of the parametric Hoare triples for
loop-free programs and to compute weakest pre-conditions.  We have
also implemented a Gaussian elimination based procedure that
propagates array equalities and simplifies select store nests in the
generated SMT formula to compute weakest pre-conditions.

\subsection{Benchmarks}

We have evaluated {\ourtool} on a test-suite of $231$ benchmarks
inspired from different algebraic functions that compute polynomials
as well as a standard array operations such as copy, min, max and
compare.  Of these there are $121$ safe benchmarks and $110$ unsafe
benchmarks.  All our programs take a symbolic parameter $N$ which
specifies the size of each array as well as the number of times each
loop executes.  Several benchmarks in the test-suite follow different
types of templates wherein either the number of loops in the program
increases or they use potentially different data values.  Program from
the first kind of templates allow us to gauge the scalability aspect
of our technique as the number of loops in the program increases.
Programs from the latter templates allow for checking the robustness
of our technique to the content of arrays and scalars.

Assertions in the benchmarks are either universally quantified or
quantifier free safety properties.  The predicates in these assertions
are (in-)equalities over scalar variables, array elements, and
possibly non-linear polynomial terms over $N$.  Although our technique
can handle some classes of existentially quantified assertions as
discussed in Sect. \ref{sec:diff-pre}, all the examples considered for
our experiments have universally quantified or quantifier-free
assertions.  The approach described in the paper is naturally
applicable to programs with such assertions, given that the underlying
SMT solver can discharge the verification conditions containing
formulas with existential quantification and quantifier alternation
when a loop-free difference program is automatically computed.
Handling post-conditions with existential quantification and
quantifier alternation are part of future work.

\subsection{Setup}

All experiments were performed on a Ubuntu 18.04 machine with 16GB RAM
and running at 2.5 GHz.  We have compared our tool {\ourtool} against
the verifiers for array programs {\viap} (v1.1) \cite{viap},
{\veriabs} (v1.3.10) \cite{veriabs20}, {\booster}
(v0.2) \cite{booster}, {\vaphor} (v1.2) \cite{vaphor} and {\freqhorn}
(v.0.5) \cite{freqhorn}.  C programs were manually converted to
mini-Java as required by {\vaphor} and CHC formulae as required by
{\freqhorn}.  Since {\freqhorn} does not automatically find
counterexamples, so we used the supplementary tool {\expl} from its
repository on unsafe benchmarks as recommend by them.  We have used
the same version of {\veriabs} that was used to perform the
experiments in \cite{tacas20}, since the later version of {\veriabs}
invokes our tool {\ourtool} in its pipeline for verifying array
programs (refer \cite{veriabs20}).  A timeout of $100$ seconds was set
for these experiments.

\begin{figure}[t!]
\centering
\begin{filecontents*}{eval/tabledata.csv}
Tool,Success,CE,Inconclusive,TO
,,Safe,,
Vajra,110,0,11,0
VIAP,58,0,2,61
VeriAbs,50,1,0,70
Booster,36,27,17,41
VapHor,27,9,2,83
FreqHorn,26,0,19,76
,,,,
,,Unsafe,,
Vajra,0,109,1,0
VIAP,1,108,0,1
VeriAbs,0,102,0,8
Booster,0,84,15,11
VapHor,1,106,1,2
FreqHorn,0,99,0,11
\end{filecontents*}
\pgfplotstabletypeset[
every even row/.style={before row={\rowcolor[gray]{0.8}}},
every head row/.style={before row=\toprule,after row=\midrule},
every last row/.style={after row=\bottomrule},
col sep=comma,string type]
{eval/tabledata.csv}
\caption{Summary of the experimental results}
\label{tab:exp-results}
\end{figure}

\subsection{Summary of the Results}

We executed all six tools on the entire set of $231$ benchmarks.  A
table with the summary of obtained results is shown in
Fig. \ref{tab:exp-results}.  We present the results on safe and unsafe
benchmarks separately for a fair representation of each tool on the
set of benchmarks.

\input{safe-plot.tex}  

\input{unsafe-plot.tex}  

\subsection{Analysis on Safe Benchmarks}

{\ourtool} verified $110$ safe benchmarks, compared to $58$ verified
by {\viap}, $50$ by {\veriabs}, $36$ by {\booster}, $27$ by {\vaphor}
and $26$ by {\freqhorn}.  {\ourtool} was inconclusive on $11$
benchmarks.  The reasons for the inability of our tool to generate a
conclusive result are as follows: (1) the difficulty in computing a
difference program due to the presence of a branch condition dependent
on $N$ or complex operations such as modulo, (2) difficulty in
computing the required strengthening of the pre- and post-conditions
and (3) the back-end SMT solver returning an inconclusive result.

{\ourtool} verified $52$ benchmarks on which {\viap} diverged,
primarily due to the inability of {\viap}'s heuristics to get closed
form expressions.  {\viap} verified $5$ benchmarks that could not be
verified by the current version of {\ourtool} due to syntactic
limitations.  {\ourtool}, however, is two orders of magnitude faster
than {\viap} on programs that were verified by both (refer
Fig. \ref{plot:safe}).

{\ourtool} proved $60$ benchmarks on which {\veriabs} diverged.
{\veriabs} ran out of time on programs where loop shrinking and
merging abstractions were not strong enough to prove the assertions.
{\veriabs} reported $1$ program as unsafe due to the imprecision of its
abstractions and it proved $3$ benchmarks that {\ourtool} could not.

{\ourtool} verified $74$ benchmarks that {\booster} could not.
{\booster} reported $27$ benchmarks as unsafe due to imprecise
abstractions, its fixed-point computation engine reported unknown
result on $17$ benchmarks and it ended abruptly on $41$ benchmarks.
{\booster} also proved $2$ benchmarks that could not be handled by the
current version of {\ourtool} due to syntactic limitations.

{\ourtool} verified $83$ benchmarks that {\vaphor} could not.  The
distinguished cell abstraction technique implemented in {\vaphor} is
unable to prove safety of programs, when the value at each array index
needs to be tracked.  {\vaphor} reported $9$ programs unsafe due to
imprecise abstraction, returned unknown on $2$ programs and ended
abruptly on $83$ programs.  {\vaphor} proved $2$ benchmarks that
{\ourtool} could not.

{\ourtool} verified $84$ programs on which {\freqhorn} diverged,
especially when constants and terms that appear in the inductive
invariant are not syntactically present in the program.  {\freqhorn}
ran out of time on $76$ programs, reported unknown result or ended
abruptly on $19$ benchmarks.  {\freqhorn} verified a benchmark with a
single loop that {\ourtool} could not.

All the benchmarks that are uniquely solved by {\ourtool} have
multiple sequentially composed loops and/or a form of
aggregation/cross-iteration dependence via a scalar variable or an
array.

\subsection{Analysis on Unsafe Benchmarks}

{\ourtool} disproved $109$ benchmarks, compared to $108$ disproved by
{\viap}, $102$ by {\veriabs}, $84$ by {\booster}, $106$ by {\vaphor}
and $99$ by {\expl}, the supplementary tool that comes with
{\freqhorn}.  {\ourtool} was unable to disprove $1$ benchmark.

{\ourtool} disproved $1$ benchmark which {\viap} could not.  {\viap}
concluded $1$ benchmark as safe and timed out on $1$ benchmark.  Even
on unsafe benchmarks, {\ourtool}, is an order of magnitude faster than
{\viap} (refer Fig. \ref{plot:unsafe}).
{\ourtool} disproved $7$ benchmarks which {\veriabs} could not.
{\veriabs} ran out of time on $8$ programs.

{\ourtool} disproved $25$ benchmarks that {\booster} could not.
{\booster} reported unknown result on $15$ benchmarks and it timed out
on $11$ benchmarks.
{\ourtool} disproved $3$ benchmarks that {\vaphor} could not.
{\vaphor} proved $1$ program as safe, returned unknown on $2$
programs and timed out on $2$ programs.
{\ourtool} disproved $10$ programs which {\expl} could not.  {\expl}
ran out of time on $11$ programs.

\subsection{Comparing the Performance}

\begin{figure*}[t!]
\centering
\begin{filecontents*}{eval/zerosum.csv}
Instance,Loops,zerosum,zerosum-const,zerosum-m,zerosum-const-m
1,3,0.54,0.49,,
2,5,0.87,0.88,0.86,0.92
3,7,1.28,1.27,1.26,1.21
4,9,1.73,1.94,2.07,1.76
5,11,2.32,2.56,2.31,2.48
6,13,,,2.94,2.95
\end{filecontents*}
\pgfplotstabletypeset[
empty cells with={--},
every even row/.style={before row={\rowcolor[gray]{0.8}}},
every head row/.style={before row=\toprule,after row=\midrule},
every last row/.style={after row=\bottomrule},
col sep=comma, string type]
{eval/zerosum.csv}
\caption{Template-wise analysis of the results for the `zerosum' benchmarks}
\label{tab:zerosum-results}
\end{figure*}

The quantile plots in Figs. \ref{plot:safe} and \ref{plot:unsafe} show
the performance of each tool on all the safe and unsafe benchmarks
respectively in terms of time taken to produce the result.  {\ourtool}
verified/disproved each benchmark within three seconds.  In
comparison, as seen from the plots, other tools took significantly
more time in proving the programs.

As mentioned previously, the test-suite has several benchmarks that
are instantiated from different templates.  For such instantiated
benchmarks that only change the data values in the instances of the
templates, we did not see any change in the performance of {\ourtool}.
Hence, we do not discuss them further.  We now discuss the results for
a set of templates where the number of loops in the benchmarks
instantiated from them increases.  In Fig. \ref{tab:zerosum-results},
we present the results of executing {\ourtool} on the benchmarks
instantiated from the `zerosum' templates.  The first column indicates
the benchmark instance number, the second column indicates the number
of loops in the instantiated benchmark, columns three to six indicate
the benchmark template name that is instantiated and give the time (in
seconds) taken by {\ourtool} to prove the given assertion in the
benchmark instance.  It can be seen from the table that as the number
of loops increase in the benchmark, our tool requires more time in
solving the benchmark.  This is primarily attributed to pre- and
post-condition strengthening step in our technique that requires our
technique to infer and prove auxiliary predicates iteratively during
the inductive step.

%% file: safe-plot.tex
\begin{figure*}[t!]
  \begin{center}
    \begin{tikzpicture}
      \begin{axis}[
          xlabel={\em \#Benchmarks},
          ylabel={\em Time (s)},
          xmin=0, xmax=125,
          ymin=0, ymax=105,
          minor x tick num=1,
          minor y tick num=1,
          axis lines=left,
          legend pos=outer north east,
        ]
        \addplot[color=blue,mark=halfcircle*,mark size=2pt] table {eval/safepd/vajra.csv}; \addlegendentry{\ourtool}
        \addplot[color=lime,mark=halfsquare*,mark size=2pt] table {eval/safepd/viap.csv}; \addlegendentry{\viap}
        \addplot[color=cyan,mark=asterisk,mark size=2pt] table {eval/safepd/veriabs.csv}; \addlegendentry{\veriabs}
        \addplot[color=olive,mark=triangle,mark size=2pt] table {eval/safepd/booster.csv}; \addlegendentry{\booster}
        \addplot[color=teal,mark=halfdiamond*,mark size=2pt] table {eval/safepd/vaphor.csv}; \addlegendentry{\vaphor}
        \addplot[color=purple,mark=star,mark size=2pt] table {eval/safepd/freqhorn.csv}; \addlegendentry{\freqhorn}
      \end{axis}
    \end{tikzpicture}
  \end{center}
    \caption{Quantile plot showing the performance of all the tools on safe benchmarks}
    \label{plot:safe}
\end{figure*}

%% file: unsafe-plot.tex
\begin{figure*}[t!]
  \begin{center}  
    \begin{tikzpicture}
      \begin{axis}[
          xlabel={\em \#Benchmarks},
          ylabel={\em Time (s)},
          xmin=0, xmax=115,
          ymin=0, ymax=105,
          minor x tick num=1,
          minor y tick num=1,
          axis lines=left,
          legend pos=outer north east,
        ]
        \addplot[color=blue,mark=halfcircle*,mark size=2pt] table {eval/unsafepd/vajra.csv}; \addlegendentry{\ourtool}
        \addplot[color=lime,mark=halfsquare*,mark size=2pt] table {eval/unsafepd/viap.csv}; \addlegendentry{\viap}
        \addplot[color=cyan,mark=asterisk,mark size=2pt] table {eval/unsafepd/veriabs.csv}; \addlegendentry{\veriabs}
        \addplot[color=olive,mark=triangle,mark size=2pt] table {eval/unsafepd/booster.csv}; \addlegendentry{\booster}
        \addplot[color=teal,mark=halfdiamond*,mark size=2pt] table {eval/unsafepd/vaphor.csv}; \addlegendentry{\vaphor}
        \addplot[color=purple,mark=star,mark size=2pt] table {eval/unsafepd/freqhorn.csv}; \addlegendentry{\freqhorn}
      \end{axis}
    \end{tikzpicture}
    \end{center}
    \caption{Quantile plot showing the performance of all the tools on unsafe benchmarks}
    \label{plot:unsafe}
\end{figure*}

%% file: related.tex
Earlier work on inductive techniques can be broadly categorized into
those that require loop-specific invariants to be provided or
automatically generated, and those that work without them.  Requiring
a ``good'' inductive invariant for every loop in a program effectively
shifts the onus of assertion checking to that of invariant generation.
Among techniques that do not require explicit inductive invariants or
mid-conditions for each loop, there are some that require loop
invariants to be implicitly generated by a constraint solver.  These
include techniques based on constrained Horn clause
solving \cite{chc,quic3,freqhorn,vaphor}, acceleration and lazy
interpolation for arrays \cite{booster} and those that use inductively
defined predicates and recurrence
solving \cite{trace20,viap,aligators}, among others.


\textsc{QUIC3} \cite{quic3}, {\freqhorn} \cite{freqhorn} and the
technique in \cite{chc} infer universally quantified inductive
invariants of array programs specified as Constrained Horn
Clauses.  \textsc{QUIC3} \cite{quic3} extends the IC3 framework to a
combination of SMT theories and performs lazy quantifier
instantiations.  {\freqhorn} \cite{freqhorn} infers universally
quantified invariants over arrays within its syntax-guided synthesis
framework and can reason with complex array index expressions by
adopting the \emph{tiling} \cite{sas17} ideas.

{\vaphor} \cite{vaphor} transforms array programs to array-free Horn
formulas.  Their technique is parameterized by the number of array
cells to be tracked resulting in an eager quantifier instantiation.

{\booster} \cite{booster} combines
acceleration \cite{acceleration1,acceleration2} and lazy abstraction
with interpolation for arrays \cite{lazyabsarray}.  Performing
interpolation to infer universally quantified array properties is
difficult \cite{Jhala,Monniaux2015}.  The technique does not always
succeed, especially for programs where simple interpolants are
difficult to compute \cite{sas17}.

VIAP \cite{viap} translates the program to an array-free quantified
first order logic formula in the theory of equality and uninterpreted
functions using the scheme proposed in \cite{viaptheory}.  They use
several tactics to simplify the generated formula and apply induction
over array indices to prove the property.
Unlike our method, it does not have heuristics for finding additional
pre-conditions that are required for the induction proof to succeed
which our method successfully infers.

\cite{trace20} uses theorem provers to introduce and
prove lemmas that implicitly capture inductive loop invariants at
arbitrary points in the program described in trace logic.

Thanks to the impressive capabilities of modern constraint solvers and
the effectiveness of carefully tuned heuristics for stringing together
multiple solvers, approaches that rely on constraint solving have shown
a lot of promise in recent years.  However, at a fundamental level,
these formulations rely on solving implicitly specified loop
invariants garbed as constraint solving problems.

Template-based techniques \cite{Gulwani,Srivastava09,Dirk07} search
for inductive invariants by instantiating the parameters of a fixed
set of templates within the abstract interpretation framework.  They
can generate invariants with alternating quantifiers, however, the
user must supply invariant templates and the cost of generating
invariants is quite high.

A large number of techniques have been proposed in literature that use
induction \cite{ind-bdds,ind-wsst,ind-tempo,ind-tlm,ind-te,ind-smt,ind-hc}
and its pragmatically more useful version k-induction
\cite{kind,kind-rtv,kind-lustre,kind-race,kind-par,kind-comb,kind-boost,kind-kinvs,kind-bmc,kind-interpol,kind-invinf,kind-hmc}.
These techniques generate and use loop invariants, especially when
aimed at verifying safety properties of programs.  In contrast, our
novel technique does not rely on generation or use of loop-specific
invariants and differs significantly from these methods in the way in
which the inductive step is formulated using the computed difference
programs and difference pre-conditions.

There are yet other inductive techniques, such as that
in \cite{lopstr12,sas17,squeezing,diffy-cav21}, that truly do not
depend on loop invariants being generated.  In fact, the technique
of \cite{lopstr12} comes closest to our work in principle.
However, \cite{lopstr12} imposes severe restrictions on the input
programs to move the peel of one loop across the next sequentially
composed loop such that the program with the peeled loops composed
with the program fragment consisting of only the peeled iterations is
semantically equivalent to the input program.
They call these restrictions on the input programs
as \emph{commutativity of statements}.  In practice, such restrictive
conditions and data dependencies are not satisfied by a large class of
programs.
For instance, the example in Fig. \ref{fig:ex} and our running example
in Fig. \ref{fig:ss} do not meet these restrictions.  The technique
of \cite{lopstr12} is thus applicable only to a small part of the
program-assertion space over which our technique works.

The tiling \cite{sas17} technique for verifying universally quantified
properties of array programs reasons one loop at a time and applies
only when loops have simple data dependencies across iterations
(called \emph{non-interference} of tiles in \cite{sas17}).  It
effectively uses a slice of the post-condition of a loop as an
inductive invariant.  In the case of sequentially composed loops, it
also requires strong enough mid-conditions to be automatically
generated or supplied by the user.  Our full-program induction
technique circumvents all of these requirements.

The method proposed in \cite{squeezing} proves programs correct by
induction on a rank, chosen as the size of program states.  It
constructs a safety proof by automatically synthesizing a squeezing
function that can map higher-ranked states to a lower-ranked state,
while ensuring that original states are faithfully simulated by their
squeezed counterparts.  This allows the method to shrink program
traces of unbounded length, limiting the reasoning to only
minimally-ranked states.  A guess-and-check approach combined with
heuristics for making educated guesses is employed for computing the
squeezing functions necessary to prove a given program.  Successful
synthesis of a squeezing function is equivalent to establishing the
inductive step.  These functions can be quite useful in practice, for
example, to prove programs that may not have a first-order
representable loop invariant.  In general, squeezing functions are not
easy to synthesize and automatically searching for such functions is a
non-trivial and an exceedingly time consuming task.  Further, the
squeezing functions can only consist of commutative and invertible
operations, restricting their applicability.  The technique may be
used in tandem with the classical loop invariant based methods.  In
comparison, our technique generates and uses difference invariants in
an explicit inductive step and it does not rely on generation and use
of squeezers to shrink the state space of the program.

The technique presented in \cite{diffy-cav21} also performs induction
on the entire program and is a parallel line of work.  As stated
in \cite{diffy-cav21}, full-program induction forms the basis of their
technique but the way in which the inductive step is formulated
differs significantly from ours.  It coins {\em difference invariants}
that relate two slightly different versions of the given program.
They use just the peeled iterations of loops as difference programs
and amend the inductive reasoning using difference invariants.  The
technique supports nested loops as well as branch conditions with
value dependent on the program parameter $N$.  The prototype
tool \textsc{Diffy} \cite{diffy} implements the method.  We believe
that there are programs for which \cite{diffy-cav21} may not be able
to successfully infer and use difference invariants, but full-program
induction (with its recursive invocation ability) will be able to
verify the post-conditions in such programs.

There are several techniques that approximate program computation
during verification.  \cite{prophecytacas21} has proposed a
counterexample-guided abstraction refinement scheme for programs that
manipulate arrays. Their idea relies on prophecy variables to refine
the abstraction.  {\veriabs} \cite{veriabs20} is an abstraction-based
verifier to prove properties of programs.  It implements a portfolio
of abstractions that enable the tool to leverage bounded model
checking.  These abstractions tend to restrict the array manipulating
loops to a fixed number of (possibly initial) iterations.  The tool
makes a series of attempts to prove the property and uses program
features to choose the next abstraction/strategy to be applied.
Fluid updates \cite{fluid} uses bracketing constraints, which are
over- and under-approximations of indices, to specify the concrete
elements being updated in an array without explicit partitioning.
While their abstraction is independent of the given property, they
assume that only a single index expression updates the array in each
loop, severely restricting the technique.
Analyses proposed in \cite{Gopan,Halbwachs} partition the
array into symbolic slices and abstracts each slice with a numeric
scalar variable.
%
Abstract interpretation based techniques \cite{Rival,ArrayCousotCL11}
propose an abstract domain which utilizes cell contents to split array
cells into groups.
In particular, the technique in \cite{Rival} is useful when array
cells with similar properties are non-contiguously present in the
array.
These approaches require the implementation of abstract transformers
for each specialized domain which is not a necessity with our framework.
Other techniques for analyzing array manipulating programs
include \cite{Jhala,verifast}.

Program differencing \cite{paige-differencing}, program
integration \cite{integration} and differential static
analysis \cite{dsa-hoare} have been studied in literature for various
purposes.  Incremental computation of expensive
expressions \cite{liu-incrementalization}, optimizing the execution
time of programs that manipulate arrays \cite{liu-optimization},
reducing the cost of regression testing \cite{diff-testing} and
checking data-structure invariants \cite{ditto07} are some
applications of such techniques. {\symdiff} \cite{symdiff} tool, based
on differential static analysis \cite{dsa-hoare}, displays semantic
differences between different program versions and checks equivalence.
However, the method neither supports checking quantified
post-conditions nor does it support loops and arrays of potentially
unbounded size.  Unfortunately, these techniques do not always
generate code fragments that are well suited for property
verification, especially when the input programs manipulate arrays.
To the best of our knowledge, full-program induction is the first
technique to successfully employ difference computation customized for
verification in an inductive setting.

Full-program induction also offers several other advantages over the
existing techniques.  For instance, it can reason with different
quantifiers over multiple variables, it does not require
implementation of specialized abstract domains for handling quantified
formulas and it can enable the use of existing tools and techniques
for reasoning over arrays.
We believe that verification tools need to have an arsenal of
techniques to be able to efficiently prove a wide range of challenging
problems.
Our novel technique, full-program induction, is a suitable fit for
such an arsenal and has been adopted by verifiers such as {\veriabs}
in practice.  Since the 2020 edition of the international software
verification competition (SV-COMP), {\veriabs} \cite{veriabs20}
invokes our tool {\ourtool} in its pipeline of tools for verifying
programs with arrays from the set of benchmarks in the competition.

%% file: conclusion.tex
We presented a novel property-driven verification technique, called
full-program induction, that performs induction over the entire
program via parameter $N$.  Significantly, our analysis obviates the
need for loop-specific invariants during verification.  The technique
automatically computes the difference program and difference
pre-condition that enable the inductive step of the reasoning.  Our
technique is general and can be applied to programs manipulating
arrays that store integers, matrices, polynomials, vectors and so on.
This give our technique the potential of verifying apis used in
machine learning and cryptography libraries.  Experiments show that
{\ourtool} performs remarkably well vis-a-vis state-of-the-art tools
for analyzing array manipulating programs.

 Possible directions of future work include investigations into
possible ways of incorporating automatically generated and externally
supplied invariants during our analysis, especially for computing
simpler difference programs and handling programs with nested loops.
Automated support for handling assertions with existential
quantification and quantifier alternation and for verifying
heap-manipulating programs as well as programs that operate on tensors
using our technique.  Investigations into the use of synthesis-based
techniques for automatically computing the difference programs and
adapting them to programs from various interesting domains forms
another line of work.  Improvements to the algorithms for simultaneous
strengthening of pre- and post-conditions can be considered.